\RequirePackage[l2tabu,orthodox]{nag}

\documentclass[11pt,a4paper,onecolumn,oneside,notitlepage]{article}
\usepackage[left=3.25cm,right=3.25cm,top=4cm,bottom=4cm]{geometry}
\usepackage[onehalfspacing]{setspace} 

\usepackage{natbib}
\usepackage[american]{babel}
\usepackage[T1]{fontenc}
\usepackage[utf8]{inputenc}
\usepackage{hyphenat}
\usepackage{microtype}
\usepackage{verbatim}
\usepackage{dsfont}

\usepackage{amsmath}
\usepackage{amsfonts}
\usepackage{amssymb}
\usepackage{amsthm}
\usepackage{nicefrac}
\usepackage{bbm}

\usepackage{float}
\usepackage{makeidx}
\usepackage{tabularx}
\usepackage{graphicx}
\usepackage{epstopdf}
\usepackage{booktabs}
\usepackage{pdflscape}
\usepackage{longtable}
\usepackage{seqsplit}
\usepackage{rotating}
\usepackage{adjustbox}
\usepackage{multirow}
\usepackage{multicol}
\usepackage{dcolumn}
\usepackage[section]{placeins}

\usepackage{soul}
\usepackage{blindtext}
\usepackage{todonotes}
\usepackage{appendix}
\usepackage{titlesec}
\usepackage[shortlabels,inline]{enumitem}
\usepackage{thmtools}
\usepackage{tcolorbox}
\usepackage{colortbl} \usepackage[nopar]{lipsum}
\usepackage{enumitem}
\setlist[description]{leftmargin=\parindent,labelindent=\parindent}
\setlist[enumerate]{font=\textnormal}

\usepackage{caption}
\usepackage{subcaption}
\usepackage{xcolor}
\definecolor{darkRed}{RGB}{144,0,0}
\definecolor{darkBlue}{RGB}{0,0,144}
\definecolor{darkGreen}{RGB}{0,144,0}
\definecolor{darkgray}{rgb}{0.7, 0.7, 0.7}
\definecolor{lightgray}{rgb}{0.9, 0.9, 0.9}

\usepackage{hyperref}
\hypersetup{
    colorlinks=true,
    breaklinks=true,
    bookmarksnumbered=true,
    bookmarksopen=true,
    bookmarksdepth=2,
    citecolor=darkBlue,
    linkcolor=darkBlue,
    urlcolor=darkBlue,
    pdftitle={Tree-based Synthetic Control Methods},
    pdfauthor={Nicolaj Mühlbach},
    pdfproducer={}
}

\titleformat{\section}[block]{\large\bfseries\centering}{\thesection.}{0.5em}{}
\titleformat{\subsection}[block]{\normalsize\itshape}{\thesubsection.}{0.5em}{}
\titleformat{\subsubsection}[block]{\normalsize\itshape\centering}{\thesubsubsection.}{0.5em}{}

\renewcommand{\thesection}{\Roman{section}}
\renewcommand{\thesubsection}{\Alph{subsection}}
\renewcommand{\thesubsubsection}{\Alph{subsection}.\arabic{subsubsection}}

\DeclareMathOperator{\Leb}{Leb}

\newcommand{\RomNum}[1]{\uppercase\expandafter{\romannumeral #1\relax}}

\theoremstyle{plain}
\newtheorem{theorem}{Theorem}

\newtheorem{lemma}[theorem]{Lemma}
\newtheorem*{remark}{Remark}

\usepackage[linesnumbered,ruled,vlined]{algorithm2e}
\makeatletter
\newcommand{\nosemic}{\renewcommand{\@endalgocfline}{\relax}}\newcommand{\dosemic}{\renewcommand{\@endalgocfline}{\algocf@endline}}\let\oldnl\nl \newcommand{\nonl}{\renewcommand{\nl}{\let\nl\oldnl}}\makeatother

\usepackage[symbol]{footmisc}
\usepackage{amscd}
\usepackage{amstext}
\usepackage{vmargin}

\usepackage{fancyhdr}

\usepackage[mathscr]{euscript}
\usepackage{mathrsfs}
\usepackage{stmaryrd}

\usepackage{xspace}
\usepackage[draft,danish]{fixme}
\usepackage{mathtools}
\usepackage{url}

\newcommand{\dd}{\operatorname{d\mkern-2.5mu}}

\newcommand{\distEq}[1][]{\overset{\smash{\raisebox{-0.5pt}{\ $
				\scriptscriptstyle \mathscr{D}$
			}}}{=}
}

\newcounter{assump}
\renewcommand{\theassump}{(\Alph{assump})}

\newenvironment{assump}{\refstepcounter{assump}\par\vskip\topsep \noindent\textbf{Assumption~\theassump:}
}{\par\vskip\topskip}

\renewcommand{\thefootnote}{\arabic{footnote}}

\newcommand{\devnull}[1]{}

\begin{document}

\pagenumbering{Alph}
\newcommand{\mytitle}{\Large{\textbf{Tree-based synthetic control methods: \\Consequences of relocating the US embassy}}}
\renewcommand{\refname}{References}
\long\def\symbolfootnote[#1]#2{\begingroup\def\thefootnote{\fnsymbol{footnote}}\footnote[#1]{#2}\endgroup}

\begin{titlepage}

    \thispagestyle{empty}
    \setlength{\parindent}{0cm}
    \renewcommand{\thefootnote}{\fnsymbol{footnote}}
    
    \begin{center}

        \vspace*{-1cm} 
        \mytitle\symbolfootnote[1]
        {
        \noindent
The authors are grateful to Alberto Abadie, Guido Imbens, Stefan Wager, Christian B. Hansen, Christian Bjørnskov, and Bent Jesper Christensen for helpful comments and suggestions, and to the Center for Research in Econometric Analysis of Time Series (CREATES), the Dale T. Mortensen Center, Aarhus University, and the Danish Council for Independent Research (Grant 9056-00011B and Grant 0166-00020B) for research support.

         }
        
\vspace{0.3333in}
{\normalsize
        \hfill
{Nicolaj Søndergaard Mühlbach}\symbolfootnote[2]{
             		Department of Economics, Massachusetts Institute of Technology, and CREATES.
                Email: \href{mailto:muhlbach@mit.edu}{muhlbach@mit.edu}.}
		\hfill
{Mikkel Slot Nielsen}\symbolfootnote[3]{
            	Department of Statistics, Columbia University, and CREATES. 
            	Email: \href{mailto:m.nielsen@columbia.edu}{m.nielsen@columbia.edu}.}            
		 \hfill}      
		
\vspace*{0.75cm}      

        {
            \normalsize
            This version: \today
        }

        \thispagestyle{empty}
        \vspace*{0cm}
        \begin{singlespace}
            \begin{abstract} 
				 \noindent
We recast the synthetic controls for evaluating policies as a counterfactual prediction problem and replace its linear regression with a nonparametric model inspired by machine learning. The proposed method enables us to achieve accurate counterfactual predictions and we provide theoretical guarantees. We apply our method to a highly debated policy: the relocation of the US embassy to Jerusalem. In Israel and Palestine, we find that the average number of weekly conflicts has increased by roughly 103\% over 48 weeks since the relocation was announced on December 6, 2017. By using conformal inference and placebo tests, we justify our model and find the increase to be statistically significant.             \end{abstract}
        \end{singlespace}
    \end{center}
    \vspace{0.5cm}
    \end{titlepage}

\clearpage
\setcounter{page}{1}
\pagenumbering{arabic}

\section{Introduction}\label{sec:Introduction_tbsc}
In social science, we are often interested in the effects of policy
interventions on aggregate entities to evaluate previous, understand
current, or counsel future policies. The aggregate units may be firms,
organizations, geographic areas, etc. Data often stem from observational
studies. The task of estimating such effects has been heavily studied, and various
methods apply to different data available (for reviews, see, e.g., \citet{Imbens2009}
and \citet{Abadie2018}).
One approach is to compare the treated unit
to a control unit not exposed to the event. One of the first examples
is \citet{Card1990}, who uses Southern US cities as a comparison
group to estimate the effect of an unanticipated Cuban migratory influx
in Miami. However, the design of a comparative case study faces certain
challenges.

\vspace{0.166667in}
First, it is not always transparent how specific control
units are chosen, and the appropriate control may be chosen ex-post.
Running several regressions may lead to publication bias due to multiple comparisons \citep{franco2014publication}.
Second, many of the current methods to evaluate policies are based
on regressions that try to maximize the pre-treatment fit, which may
not generalize well out-of-sample. The situation illustrates the classical
bias-variance trade-off, where methods are often chosen to minimize
bias rather than balancing bias for variance to minimize the mean-squared prediction error. We argue that the problem underlying synthetic controls is in fact a prediction problem, and approaches should be designed specifically to accurately predict the outcome of the treated unit post-treatment in a counterfactual state absent from the treatment. This is especially
useful if pre-treatment inference is not a goal in itself (for a discussion
on recasting economic problems as prediction problems, see, e.g.,  \citet{Kleinberg2015}).
Third, the standard approach to comparative case studies is to specify
a linear functional form to capture the relationship between the treated
unit and the control units. This may be restrictive if we are trying
to answer questions for which no economic model exists. In addition,
the standard approach does not take nonlinearities, especially interactions,
into account except those explicitly modeled by the researcher. If
the process that generates the outcomes for the treated unit in the
pre-treatment period is nonlinear in the control outcomes, the resulting
bias may be severe.

\vspace{0.166667in}
\citet{abadie2010synthetic} solve the first challenge by relying on the ideas of \citet{Abadie2003}. In the presence of a single treated unit
and several control candidates, synthetic controls form a set of weights
such that the weighted average of the control units approximately
matches the treated unit in the pre-treatment period. The same weights
are then channeled to the post-treatment period to estimate a synthetic
control group that constitutes the counterfactual state of the world,
in which the treated unit was not exposed to the treatment.  By restricting weights to be non-negative and sum to one, the method already reduces the risk of overfitting, but ultimately, however, the method is focused on pre-treatment fit (and inference), which may not generalize out of sample. 
\cite{Doudchenko2016} modify the synthetic controls by allowing for a more transparent and flexible regularization, where weights may be negative and are not required to sum to one. In particular,  weights are estimated using the elastic net estimator which, compared to ordinary least squares, tends to shrink the weights towards zero and set some
of them exactly to zero. Especially in moderately-high dimensions,
this approach has shown promise in forecasting studies. Also, the
selection property by zeroing out some weights has attractive interpretations
as it allows researchers to pinpoint which control units that have no explanatory
power when forming the counterfactual outcome.

\vspace{0.166667in}
Both of the above methods, however, specify a linear model that is not capable
of automatically detecting nonlinearities between treated and control units.
In particular, we expect many low-order interactions of the control
outcomes to be informative in explaining the outcomes of the treated
unit. For instance, consider the empirical application in \cite{abadie2010synthetic}
regarding cigarette sales in the US. While the sales in California
may be modeled as a weighted average of the sales in New York and
Florida given a common cigarette consumption pattern along the coasts,
a decrease in sales in New York may be associated with an even bigger
decrease in California given a low period of sales in Florida. This
could happen if the people of California see themselves as trendsetters
in regards to health; when people in both New York and Florida are
reducing their cigarette consumption, people of California want
to reduce their consumption even further.
Note that, in the context of forecasting, it is becoming
natural to include interactions and higher-order
terms when relying on regularization-based estimators---in contrast to synthetic controls. However, important interactions
and higher-order terms can be difficult to anticipate ex-ante. The
kitchen sink approach would be to include all higher-order terms up
to a pre-specified order, e.g., to the third order. This approach quickly faces its own problems, since even with 10 control units, all terms up to third-order would count $\binom{10+3}{3}-1=285$, which is infeasible to handle for most parametric estimators given
a reasonable amount of observations. Thus, if nonlinearities are deemed important
or the researcher does not have the domain knowledge
required to specify an economic (parametric) model, it might be more appropriate to apply a flexible nonparametric prediction method. 

\vspace{0.166667in}
Nonparametric approaches to estimating treatment effects do exist
in the econometric toolbox. Indeed, \cite{Athey2016},
\cite{Wager2018}, and \cite{Athey2019} also rely on ideas from
machine learning to study heterogeneous treatment effects using nonparametric
models. They propose various modifications of the random forests algorithm
of \cite{Breiman2001}. Moreover, their
methods are most suitable when a large set of both observations and
covariates is available as they focus on heterogeneous treatment effects,
whereas we focus on average treatment effects. As another example,
\cite{Hartfort2017} use deep neural nets for counterfactual prediction.
We find, however, that many applications in social science and ours
included do not enjoy the luxury of having sufficiently large datasets
available to apply (deep) neural nets.

\vspace{0.166667in}
We recast the problem of \emph{estimating} a synthetic control as
the problem of \emph{predicting} one, similarly to both \cite{Doudchenko2016}
and \cite{Athey2019} who also advocate for powerful prediction methods. This way, we do not have to rely on linear, parametric models that
potentially misspecify the true underlying model. We choose a popular method
from the machine learning literature, namely the random forests algorithm, which handles interactions and
other nonlinearities automatically. In this way we obtain tree-based synthetic control method as an alternative,
particularly suited for applications where the researcher prefers
accurate post-treatment predictions over the ability to do pre-treatment
inference, and when the empirical question is not guided by any economic
model that can justify specific assumptions on the statistical model. This method also allows us to consider
all potential controls in the donor pool transparently; in particular, if some
control units do not contribute to explaining the treated unit, the
method is flexible enough to leave them out. Furthermore, our method naturally captures nonlinearities between treated and control units ---features which alternatively should have been captured by including a (potentially large) number of interactions and higher-order terms. We provide theoretical guarantees for our method as we establish asymptotic unbiasedness and consistency of the random forests predictions as well as consistency of a corresponding estimator of the average treatment effect. 

\vspace{0.166667in}
We showcase the tree-based synthetic control method by estimating
the effect of relocating the US embassy from Tel Aviv to Jerusalem on
the number of weekly conflicts in Israel and Palestine. It is beyond
our interest to judge the particular political decision, rather we
propose a method to estimate its impact. We use conflict data from
December 28, 2015, to November 3, 2018, for Israel and Palestine as
well as for 11 of the remaining countries in the Middle East as controls.
The data are provided by the Armed Conflict Location \& Event Data
Project \citep{raleigh2010introducing}. Our results indicate that
the weekly number of conflicts has increased by 26 incidents on average
after the relocation was announced on December 6, 2017, until November 3,
2018. This corresponds to more than doubling the number of conflicts. This application highlights the need for nonlinear methods. For instance, when we seek to understand
which periods are similar in terms of the level of conflict, it is
difficult to consider conflict levels in Iraq and Saudi Arabia separately
without an interaction between them. Imagine some violent and frequent
conflicts in the South of Iraq in a given period. The regime of Saudi
Arabia may react by increasing the appearance of police forces in
major cities, and as a result, the number of conflicts falls. If such
interactions matter for the conflict level in Israel and Palestine,
we would incur an omitted variable bias by leaving them out.
We use the recently proposed conformal inference test by \cite{Chernozhukov2017b}
to formally justify our results. The increase is statistically significant
at a 1\% level.

\vspace{0.166667in}
The proposed method uses the pre-treatment periods to estimate the
relationship between the treated and control units and it imposes
this relationship in the post-treatment period, similarly to \cite{abadie2010synthetic}
and \cite{Doudchenko2016}. Our model for the conditional expectation, the canonical random forests regression model, have proved successful in many applications (see, e.g., \cite{Montgomery2018}
for a recent paper in political science, or \cite{Ng2019} in IO). Further, variants of random forests have already been employed in
the treatment effects literature either directly \citep{Athey2016,Wager2018,Athey2019}
or indirectly \citep{Chernozhukov2017a,Chernozhukov2018}. Common
to these papers is that they rely on the unconfoundedness assumption
and assume there is a relationship between outcomes for a given unit
over time (estimated by regressing control unit outcomes in treated
periods on lagged outcomes) that is stable across units. In contrast,
the synthetic control literature assumes there is a relationship between
different units (estimated by regressing treated unit outcomes on
control outcomes) that is stable over time. Our approach falls into
the latter. 
Intuitively, for each period where the treated unit is treated, our
model locates a few corresponding pre-treatment periods based on the
control units and uses the average of the pre-treatment outcomes of
the treated unit as a counterfactual prediction in the post-treatment
period. Stated differently, our model aggregates the pre-treatment
periods into similar subgroups based on the control units. Then, it
computes the average of the outcomes of the treated unit in each of
the subgroups. In the post-treatment period, the model remembers how
to group the periods and assigns the corresponding pre-treatment average
to each of the periods. This gives an estimate of the potential outcome
for the treated unit in the absence of the treatment. Having an estimate
for all periods after the intervention, we compute the average of
the differences between the estimate and the actual outcome, similarly
to \cite{Chernozhukov2018ate}.  

\vspace{0.166667in}
Using the application, we compare our proposed method to state-of-the-art methods on several metrics. We show that our method performs at least on par with competing approaches,  while enjoying the benefits of being more off-the-shelf. In particular, we impose fewer assumptions on the relationship between the treated unit of interest and the units in the donor pool, which may become beneficial when no economic model exists to guide the researcher.  When this relationship is indeed linear, our method is still able to recover it, although the standard methods may be more efficient. All methods considered agree on the magnitude of the treatment effect.

\vspace{0.166667in}
The rest of the paper is organized as follows. Section \ref{sec:Tree-based-Control-Methods} introduces the framework underlying synthetic controls and the tree-based synthetic method is formally presented. In addition,  this section presents theoretical guarantees.  Section \ref{sec:Estimating-the-Consequences} considers
the context of Israel and Palestine and presents the results alongside
several robustness checks. Section \ref{sec:Comparing-Methods} compares
our method to state-of-the-art econometric methods. Section \ref{sec:Conclusion_tbsc}
concludes. All proofs are deferred to the Appendix. 

\section{Synthetic Control Methods}\label{sec:Tree-based-Control-Methods}
\subsection{Framework}
We consider $N+1$ cross-sectional units observed in $T$ periods
and assume without loss of generality that only the first unit is
exposed to the treatment, leaving $N$ units as controls\footnote{We will use \emph{treatment} and \emph{intervention} interchangeably.}.  At any given time $t = 1,\dots, T$, we group the observations as $(X_t,Y_t)$ where $X_t\in \mathcal{X}\coloneqq \mathcal{X}_1\times \cdots \times \mathcal{X}_N$ consists of the $N$ control units and $Y_t\in \mathbb{R}$ is the treated unit. Here $\mathcal{X}_i \subseteq \mathbb{R}$ reflects the support of the $i$th control unit.  We assume that the treatment occurs at time $T_{0}<T$,  leaving the first $T_{0}$ periods as the pre-treatment period.  In many applications, ours included, the treatment may have an effect before implementation
via announcement or anticipation, and $T_{0}$ should be redefined
accordingly. We assume implicitly that the treatment does not affect
the outcome for the control units (cf.\ \ref{a1} and~\ref{aa1} of Appendix~\ref{Sec:Appendix_proof}). For a thorough discussion on this assumption, see e.g. \citet{Rosenbaum2007}. 
Let $Y_{t}^{0}$
denote the potential outcome that would be observed for the treated unit at time $t$ in absence of treatment,  and similarly,  let $Y_{t}^{1}$
denote the potential outcome that would be observed if exposed to
the intervention.  In particular,  we have that
\begin{align}
Y_{t}=\begin{cases}
	Y_{t}^{0} & \text{for }t=1,\ldots,T_{0}\\
	Y_{t}^{1} & \text{for }t=T_{0}+1,\ldots,T.
\end{cases}
\end{align}

We define $\tau_{t}=Y_{t}^{1}-Y_{t}^{0}$ as the effect of the intervention at time $t$ for the treated unit. Assuming that both $(X_t,Y^0_t)$ and $(Y^1_t)$ are stationary processes, the average treatment effect (ATE)  can thus be characterized as
\begin{equation}\label{eq:population ATE}
\tau = \mathbb{E}[\tau_t].
\end{equation}
We remark that, under the mild assumption of ergodicity of the sequences $(Y^0_t)$ and $(Y^1_t)$, a simple consistent estimator of $\tau$ is given by the difference in averages before and after treatment, that is, 
\begin{equation}\label{tauTilde}
\tilde{\tau} = \frac{1}{T-T_0}\sum_{t=T_0 + 1}^T Y^1_t - \frac{1}{T_0}\sum_{t=1}^{T_0}Y^0_t.
\end{equation}
This estimator is based solely on $Y_1,\dots, Y_T$ and does not make use of information carried in the control units $X_t$ for $t=1,\ldots,T$. Alternatively, to make use of the information carried in the control units, we rely on a prediction-based estimator of $\mathbb{E}[Y^0_1]$ where the unobserved (no intervention) outcomes $Y^0_t$, $t=T_0+1,\dots, T$ are imputed by means of regression. The simple unbiased estimate often turns out to be of limited practical value as it fails to deliver a precise estimate of the average treatment effect (for a related discussion in the context of randomized control trials, see, e.g., \cite{Deaton2018}). Another benefit of our approach relative to the simple average is that it unlocks estimates of $\tau_t$ for all $t=T_0+1,\ldots,T$, which are intrinsically interesting to study as the effects may increase or decrease over time.

\vspace{0.166667in}
Regression imputation is a well-known strategy in the context of missing data (see, e.g., \cite{musil2002comparison,shao2002sample}). The basic idea behind this method is that, with $f(x) = \mathbb{E}[Y^0_1\mid X_1=x]$ being the regression function, $f(X_t)$ delivers a ``good'' proxy of $Y^0_t$ for $t=T_0+1,\dots, T$.  Note that $X_{t}$ could include covariates other than the control units as long as they are not affected
by the intervention. For instance, we would not be able to include
stock market indicators for Israel and Palestine. For simplicity,
however, we follow \citet{abadie2010synthetic} and focus on using
the control units. Also note that $f(X_t)$ has the same distribution as $\mathbb{E}[Y_1^0\mid X_1]$ and not $Y^0_1$. However, since their means coincide, this is not an issue in the estimation of the ATE. The resulting estimator $\hat{\tau}$ of $\tau$ that we will employ in the empirical application is simply obtained by replacing the second average in \eqref{tauTilde} by the average over predictions of $Y^0_{T_0+1},\dots, Y^0_T$, i.e.,
\begin{equation}\label{tauHat}
\hat{\tau} = \frac{1}{T-T_0} \sum_{t=T_0+1}^{T} \hat{\tau}_t,
\end{equation}
where $\hat{\tau}_t = Y^1_t - \hat{f}(X_t)$, where $\hat{f}$ is a suitable estimator of $f$. 

\subsection{The tree-based synthetic control method}\label{subsec: Tree-based Control Methods}
The framework outlined above coincides with the idea of \citet{abadie2010synthetic}
to the extent that we also use regressions to estimate the
relationship between treated and control units in the
pre-treatment period and assume that the estimated relationship continues
into the post-treatment period.  However, to estimate the ATE using in \eqref{tauHat},  we choose a popular and flexible estimator, $\hat{f}$, for the regression function called random forests regression \citep{Breiman2001}.  Using random forests as an estimator for $f$ and under suitable regularity assumptions, we prove in the following that $\hat{\tau}$ is a consistent estimator of $\tau$ in the sense that 
\begin{equation}
\hat{\tau}\to \tau\qquad \text{in probability as $T_0,T-T_0 \to \infty$.}
\end{equation}
However, we begin by providing some intuition behind how random forests work in the context of the tree-based synthetic control method.

\vspace{0.166667in}
The cornerstone of random forests is a single regression tree. Regression trees are obtained by recursively partitioning the input space (i.e., the possible values of the control units) into cleverly chosen subsets (nodes), and then they output a constant value for all inputs within the same
terminal node (also called a leaf). Specifically, any (post-treatment) outcome of the control units belongs exactly to one particular leaf, and to form the counterfactual prediction for the treated unit, the model uses the average pre-treatment outcome of the treated unit based on the corresponding outcomes of the control units falling into the same leaf. The recursive partitioning is formed in such a way that the leaf associated to a given outcome can be identified by asking a sequence of questions such as ``Is the outcome of the $i$th control unit above $10$?'', ``Is it below $20$?'', ``Is the outcome of the $j$th control above $15$?'', etc. Indeed, starting from the entire input space, the procedure works by selecting a node to split, a split direction (which of the $N$ control units to ask a question about), and a split position (the level to exceed or stay below). The node to split is called the parent node of the resulting two subnodes, which are also called child nodes. The random forests estimate is simply obtained by averaging over $B$ regression trees, which differ due to exogenous randomness injected in the recursive partitioning procedure.

\vspace{0.166667in}
Figure \ref{fig: decision tree example}
shows an example related to our application. In the example, we divide
the weekly level of conflicts in Israel and Palestine at each period
$t\leq T_{0}$ into bins based on the weekly level of conflicts in
Saudi Arabia and Iraq. Given an observation of the weekly level
of conflicts in Saudi Arabia and Iraq at a new point in time,
say $t'>T_{0}$, we decide which of the four leaves that $t'$ belongs
to. As an example, suppose this observation ends up in the first leaf, Subgroup 1. Our prediction of
the weekly level of conflicts in Israel and Palestine is then the
average of all observations that fall into Subgroup 1 in the pre-treatment
period. Hence, the outcomes for Saudi Arabia and Iraq enter only
in the stratification and, thus, the approach also allows the inclusion
of other covariates.
\begin{figure}[!t]
		\begin{adjustbox}{max totalsize = {\textwidth}{0.3\textheight}, center}
			\includegraphics[width = \textwidth]{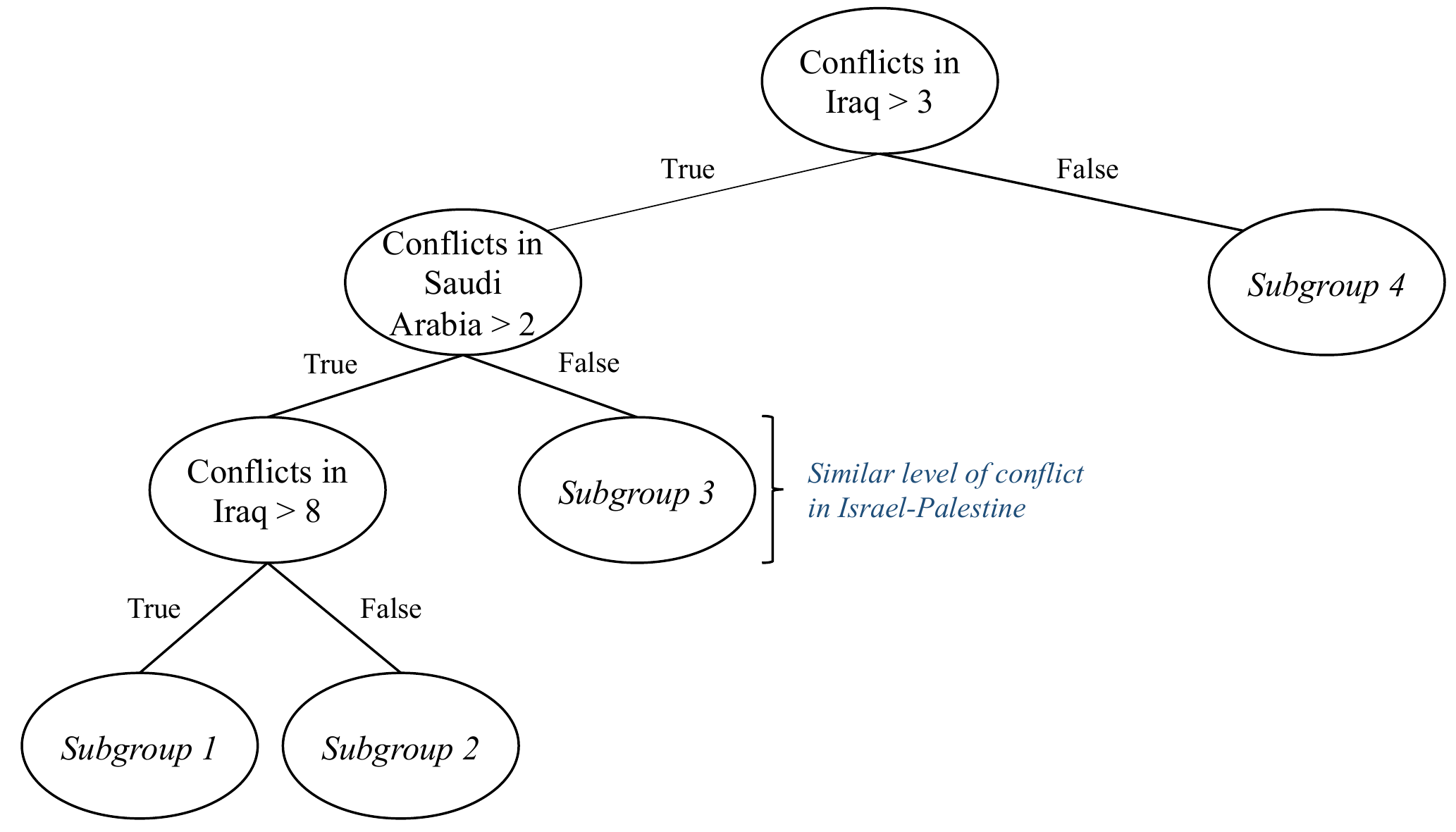}
		\end{adjustbox}
\caption{An Example of a Decision Tree}
\label{fig: decision tree example}
\vspace{0.166667in}
{\footnotesize
\textit{Notes:} As input variables we
consider the level of weekly conflicts in Saudi Arabia and Iraq.
First, we stratify observations depending on whether or not the level
of weekly conflicts in Iraq is above three (thus, ``Iraq'' is the split direction and ``Three'' is the split position). This will place any
observation in one of the resulting two child nodes. Next, we partition one of these nodes by asking whether or not the weekly level of conflicts in Saudi Arabia is above two,
etc. The recursive partitioning procedure results in four distinct subgroups/leaves.\par}
\end{figure}

To continue the example above, a possible
data-generating process (DGP) that fits nicely to this tree-based framework would be
\begin{equation}
Y_{IP,t}^{0}=\beta_{1}Y_{SA,t}+\beta_{2}Y_{IR,t}+\beta_{3}Y_{SA,t}Y_{IR,t}+\varepsilon_{0,t},\label{eq:ex_nonlinear_dgp}
\end{equation}
where $Y_{\cdot,t}$ denotes the conflict level in period $t<T_{0}$,
$IP$ abbreviates Israel-Palestine, $SA$ Saudi Arabia, and $IR$
Iraq. An additive model that does not explicitly take the interaction into account would
suffer from omitted variable bias. On the other hand, the random forests model requires no such knowledge, since it will
automatically detect the (unknown) functional form of the regression function as long as $(Y_{SA,t},Y_{IR,t})$ are included as control units---this will be shown in the theoretical results below.

\subsubsection{Theoretical results}\label{consSection}
Let $\mathcal{D}_{T_0} = \{(X_1,Y^0_1),\dots, (X_{T_0},Y^0_{T_0})\}$ denote the data from the pre-treatment period used to build the estimator $\hat{f}$. In the following  we argue that, under suitable assumptions, various types of random forests are consistent estimators of $f$.  Then, we establish consistency of the tree-based synthetic control method in the context ATE estimation. Details of the presented results, as well as their proofs, can be found in Appendix~\ref{Sec:Appendix_proof}. 

\vspace{0.166667in}
In line with \cite{davis2020rf} and \cite{wager2015adaptive} we consider a subclass of random forests, which we call $(\alpha,k,m)$-forests, parameterized through $\alpha \in (1/2)$, $k\geq 1$, and $m \geq 2k$. The triple $(\alpha,k,m)$ indicates that the trees of the forest obey the following rules:
\begin{enumerate}[(i)]
	\item\label{as1} All leaves contain strictly less than $m$ pre-treatment observations.
	\item\label{as2} No leaf contains less than $k$ pre-treatment observations.
	\item\label{as3} The split position is chosen such that the resulting child nodes contain at least a fraction $\alpha \in (0,1/2)$ of the data points in the parent node.
	\item\label{as4} The probability that a given node is split along the $i$th direction is bounded from below by a strictly positive constant across $i \in \{1,\dots, N\}$.
\end{enumerate}
We will now attach some intuition to \ref{as1}--\ref{as4}. The first rule~\ref{as1} ensures that leaves are not too large. Since we also impose the (rather technical) rule \ref{as3} that splits cannot be too ``unbalanced'', it effectively means that we require a large number splits before getting to a leaf (i.e., many questions should be asked about the outcome of the control units before the associated leaf can be identified). Concerning \ref{as2}, this is imposed to ensure that sample averages within leaves stabilize such that they are not too far from their theoretical counterparts. Finally, \ref{as4} implies that many splits will be placed across any of the $N$ directions of the input space; in other words, the partition associated to the given tree is fine in all directions. 

\paragraph*{Consistency of random forests regression}
Consistency of similar forests was established in \cite{wager2015adaptive} in the setting where $(X_1,Y^0_1),\dots, (X_{T_0},Y^0_{T_0})$ are assumed to be i.i.d., and \cite{davis2020rf} considered an autoregressive setup, $X_t = (Y^0_{t-1},\dots, Y^0_{t-p})$, when $Y^0_t$ is a $p$-th order Markov chain.  None of these settings, however, are suited for our application. Therefore,  in Appendix \ref{Sec:Appendix_proof}, we formally introduce the assumptions needed on the data-generating process to establish Theorem~\ref{consistency} (see Assumption~\ref{dataAssump}).  Most of the assumptions are not restrictive from a practical point of view and are mainly imposed for the sake of simplicity and to avoid other, less transparent, assumptions.  Note that we assume that the sequence $(X_t,Y^0_t)_{t\in \mathbb{Z}}$ is stationary and has exponentially decaying strong mixing coefficients. This is classical when proving asymptotic results, particularly when the results contain information on convergence rates as is the case in Lemma~\ref{treesConcentration} (in such situation, ergodicity is not sufficient). It is satisfied for a wide range of stationary processes; e.g., ARMA processes, Markov chains, and other short-memory time series are included in this setting. We are now ready to formulate our consistency result for $(\alpha, k ,m)$-forests.

\begin{theorem}\label{consistency}
	Let $\hat{f} = \hat{f}(\: \cdot \: ; \mathcal{D}_{T_0})$ be an $(\alpha, k,m)$-forest and suppose that Assumption~\ref{dataAssump} is satisfied. Suppose also that $k/(\log T_0)^4\to \infty$ and $\log (T_0/m)/\log (\alpha^{-1})\to \infty$ as $T_0 \to \infty$. Then
	\begin{equation*}
	\vert \hat{f}(x) - f(x)\vert \leq \delta_1 + \delta_2 (x),
	\end{equation*}
	where
	\begin{enumerate}[(1)]
		\item $\delta_1$ and $\delta_2(x)$, $x\in \mathcal{X}$, are (uniformly) bounded by a constant,
		\item $\delta_1$ does not depend on $x$ and $\delta_1 \to 0$ in probability as $T_0\to 0$, and
		\item $\delta_2(x)\to 0$ as $T_0 \to \infty$ almost surely for each $x \in\mathcal{X}$.
	\end{enumerate}
	In particular, $\hat{f}$ is a pointwise consistent estimator of $f$ in the sense that
	\begin{equation}\label{classicalCons}
		\hat{f}(x) \longrightarrow f(x)\qquad \text{in probability as $T_0\to \infty$}
	\end{equation}
	for any $x\in \mathcal{X}$.
\end{theorem}

\begin{remark}
Under Assumption~\ref{dataAssump} (part \ref{a2}), both $\vert \hat{f}(x)\vert$ and $\vert f(x)\vert$ are bounded by $M$, so the convergence in probability \eqref{classicalCons} is equivalent to convergence in $\gamma$th order mean, i.e., 
\begin{equation*}
\mathbb{E}[\vert \hat{f}(x) - f(x)\vert^\gamma]\longrightarrow 0,\qquad T_0\to \infty,
\end{equation*}
for an arbitrary $\gamma \in (0,\infty)$. In particular, the estimator $\hat{f}(x)$ is asymptotically unbiased; $\mathbb{E}[\hat{f}(x)]\to f(x)$ as $T_0 \to \infty$.
\end{remark}

While it is always of particular interest to know about convergence rates as well, we note that this depends heavily on the rate of the random forests estimator $\hat{f}$. Such results are very difficult to prove and exist only in idealized settings---in particular, results are only available in the case of independent observations and the restrictions on the random forests algorithm are often rather strict and unrealistic in practice.

\paragraph*{Consistency of the tree-based synthetic control method}
We can now apply Theorem~\ref{consistency} to prove consistency of $\hat{\tau}$ as given in \eqref{tauHat}. We will, however, need a slightly stronger assumption than Assumption~\ref{dataAssump}. In particular,  we assume that the sequence $(X_t,Y_t^0)_{t\in \mathbb{Z}}$ has exponentially decaying $\beta$-mixing coefficients,  which is imposed to be able to estimate certain expectations.  Many stationary processes satisfy the $\beta$-mixing condition as well; e.g., ARMA and Markov processes. In addition, we assume that the sequence $(Y^1_t)_{t\in \mathbb{Z}}$ is ergodic, and $\mathbb{E}[\vert Y^1_t\vert ]<\infty$, which is needed for the sample average $\sum_{t=T_0+1}^TY^1_t$ to converge to $\mathbb{E}[Y^1_1]$ (and for the latter to be well-defined and finite). The full set of assumptions needed to establish Theorem \ref{tauConsistency} is provided in Appendix \ref{Sec:Appendix_proof} (see Assumption~\ref{dataAssump2}). We can now formulate our consistency result for tree-based synthetic control methods.

\begin{theorem}\label{tauConsistency}
	Let $\hat{f}=\hat{f}(\: \cdot \: ; \mathcal{D}_{T_0})$ be an $(\alpha , k,m)$-forest, and let $\hat{\tau}$ be given by \eqref{tauHat}. Suppose that Assumption~\ref{dataAssump2} is satisfied and that $k/(\log T_0)^4\to \infty$, $\log (T_0/m)/\log (\alpha^{-1})\to \infty$, and $T-T_0\to \infty$ as $T \to \infty$. Then, $\hat{\tau}$ is a consistent estimator of the ATE $\tau$ in the sense that
	\begin{equation*}
	\hat{\tau} \longrightarrow \tau\qquad \text{in probability as $T \to \infty$.}
	\end{equation*}
\end{theorem}

\section{Estimating the Effects of relocating the Embassy}\label{sec:Estimating-the-Consequences}
\subsection{Background}
Monday afternoon December 6, 2017, the US President fulfilled a major
campaign promise by announcing the relocation of the embassy from Tel Aviv
to Jerusalem, which took place on May 14, 2018. Many international media
reported intensively on the move that broke with decades of US policy
by recognizing Jerusalem as the capital of Israel, although former
US presidents have also been commenting on the relocation. For instance,
Bill Clinton supported recognizing Jerusalem as the capital and the
principle of moving the embassy there. George W. Bush said before
taking office that he intended to move the embassy, and Barack Obama
spoke of Jerusalem as the capital of Israel that ought to remain undivided.
However, the former presidents all consistently signed waivers to
postpone the move.

\vspace{0.166667in}
The relocation should be viewed as the most recent event in the ongoing
Israeli-Palestinian conflict, dating back to the mid-20th century
in which the Jewish immigration and the sectarian conflict in Mandatory
Palestine between Jews and Arabs took place. In 1948, the establishment
of the State of Israel alongside the State of Palestine was proclaimed,
and US President at the time Harry S. Truman recognized the new nation.
Since 1967, Israel has held all of the pre-war cities of West and
East Jerusalem, and in addition, the Gaza Strip has been under Israel's
control. Ever since, several wars have been fought between the Arab
countries and Israel, and a permanent solution is still to be found.
For a complete review and analysis of the Israeli-Palestinian conflict,
see \cite{Frisch2004} and \cite{Eriksson2018}.

\subsection{Data and sample}
\begin{table}[!t]
\caption{Summary Statistics of Weekly Conflicts in the Middle East (excl. Iran and Syria)}
\begin{adjustbox}{max width = \textwidth, center}
\begin{tabular}{lrrrrrrr}
\toprule
\textbf{Country} & \textbf{Mean} & \textbf{Sd.} & \textbf{Min} & \textbf{Q1} & \textbf{Median} & \textbf{Q3} & \textbf{Max} \\
\midrule 
Israel-Palestine &  32.9 & 18.7 &  8.0 &  20.0 &  29.0 &  41.0 & 106.0 \\    
Bahrain &   6.8 &  6.9 &  0.0 &   1.0 &   5.0 &  11.0 &  31.0 \\    
Iraq &  96.8 & 33.8 & 32.0 &  65.0 &  97.0 & 120.0 & 186.0 \\    
Jordan &   1.4 &  2.6 &  0.0 &   0.0 &   1.0 &   2.0 &  21.0 \\    
Kuwait &   0.1 &  0.4 &  0.0 &   0.0 &   0.0 &   0.0 &   2.0 \\    
Lebanon &   6.2 &  4.8 &  0.0 &   3.0 &   5.0 &   9.0 &  25.0 \\    
Oman &   0.0 &  0.2 &  0.0 &   0.0 &   0.0 &   0.0 &   2.0 \\    
Qatar &   0.0 &  0.1 &  0.0 &   0.0 &   0.0 &   0.0 &   1.0 \\    
Saudi Arabia &  27.8 & 15.8 &  0.0 &  17.0 &  27.0 &  39.0 &  75.0 \\    
Turkey &  46.0 & 75.4 &  6.0 &  22.0 &  34.0 &  51.0 & 777.0 \\    
United Arab Emirates &   0.0 &  0.1 &  0.0 &   0.0 &   0.0 &   0.0 &   1.0 \\    
Yemen & 168.7 & 39.5 & 72.0 & 137.0 & 173.0 & 197.0 & 313.0 \\    
Average (excl. Israel-Palestine) &  32.2 &  8.4 & 17.6 &  28.8 &  31.1 &  34.3 & 100.1 \\ 
\bottomrule
\end{tabular} 
\end{adjustbox}\vspace{0.166667in}
\label{tab:summary_stat}
{\footnotesize 
\textit{Notes:} Summary statistics of
the weekly conflicts in the Middle East, excl. Iran and Syria. Measures
in order of appearance include mean, standard deviation, minimum,
first quartile, median, third quartile, and maximum. The countries
other than Israel-Palestine are grouped as \emph{Average (excl. Israel-Palestine)}.\par}
\end{table}
We use daily country-level panel data in the period December 28, 2015,
to November 3, 2018, on conflicts reported by the Armed Conflict Location
\& Event Data Project \citep{raleigh2010introducing}. The conflicts
cover riots, protests, strategic development, remote violence, violence
against civilians, various types of battles, and headquarter or base
establishments. We consider the aggregate of all conflicts and leave
the disaggregating for further research. The data consist of multiple
daily observations which we aggregate into weekly observations to smooth the daily variations.
We have no other data on a daily or weekly frequency. The treated countries
considered are Israel and Palestine, which we aggregate into one treated
unit to take into account the interdependency of the two countries
\citep{Arnon2001}.\footnote{We sometimes refer to Israel and Palestine
as Israel-Palestine.} Aggregating them into one treated unit rather
than having one of them, say Israel, as a potential control is necessary
to meet the assumption of no interference between units. One may be
interested in the effects on Israel and Palestine separately, leaving
out the other country completely to avoid interference.
Another reason to aggregate Israel and Palestine into one is because several of the reported conflicts happen at the border between the two countries, which favors the aggregation.

\vspace{0.166667in}
An interesting hypothesis is whether the conflicts in Palestine accelerate earlier
than the conflicts in Israel. However, this is hard to measure, as
the conflicts in both countries may be initiated by people from either
place, making it difficult to disentangle the effect in Israel from
the effect in Palestine. As we are interested in the overall effect
in the area, we aggregate the countries for now and leave the other
hypothesis for future research.

\vspace{0.166667in}
The control countries we consider are all the
remaining countries in the Middle East but Syria and Iran, which include
Bahrain, Iraq, Jordan, Kuwait, Lebanon, Oman, Qatar, Saudi Arabia,
Turkey, United Arab Emirates, and Yemen, giving us a total of 11 control
countries. The data coverage for Syria starts from January 2017, and
instead of restricting our sample to begin here, we choose to exclude
Syria. We also exclude Iran because of its involvement in the Israeli-Palestinian
conflict and its relation to the US, which make it too difficult to
justify the assumption of no inference between units (see \cite{Buonomo2018}
for an analysis of the Iran-US relation).

\vspace{0.166667in}
In fact, if we compare the trends in the weekly level of conflicts in Iran and Israel-Palestine
before and after the move of the embassy, the co-movement is clear.
We document the trends in the weekly number of conflicts for all countries
in the Middle East except Syria in Appendix \ref{Sec:Appendix_tbsc}.
The pre-intervention period covers 101 weeks, starting December 28,
2015, and ending December 3, 2017, just before the announcement. The
post-intervention period begins on December 4, 2017, and ends on November
3, 2018, leaving 48 weeks for estimating the average level of conflicts
in Israel and Palestine in the counterfactual situation where the
US embassy is not relocated. Summary statistics for the weekly number
of conflicts across the Middle East countries are provided in Table
\ref{tab:summary_stat}.

\vspace{0.166667in} 
Further, we show the distribution of the
weekly number of conflicts in Israel-Palestine in both the pre-treatment
and post-treatment period in Figure \ref{fig: hist_israel_palestine}.
It follows from Figure \ref{fig: hist_israel_palestine} that the
distribution is shifted to the right in the post-treatment period,
which tentatively suggests that violent weeks tend to occur more often
in the post-treatment period. Last, Figure \ref{fig: weekly aggregates of conflicts in Palestine and Israel}
shows the level of conflicts over time in Israel-Palestine as well
as the average of the remaining countries.  As noted in Section \ref{sec:Tree-based-Control-Methods},  under mild assumption of ergodicity of the sequences $(Y^0_t)$ and $(Y^1_t)$, a simple before-after comparison is sufficient to identify
the average treatment effect, which in this case would be 23.5 weeks.
This simple yet unbiased estimate is roughly in line with the results
we show next. 
\begin{figure}[!t]
		\begin{adjustbox}{max totalsize = {0.7\textwidth}{0.9\textheight}, center}
			\includegraphics[width = \textwidth]{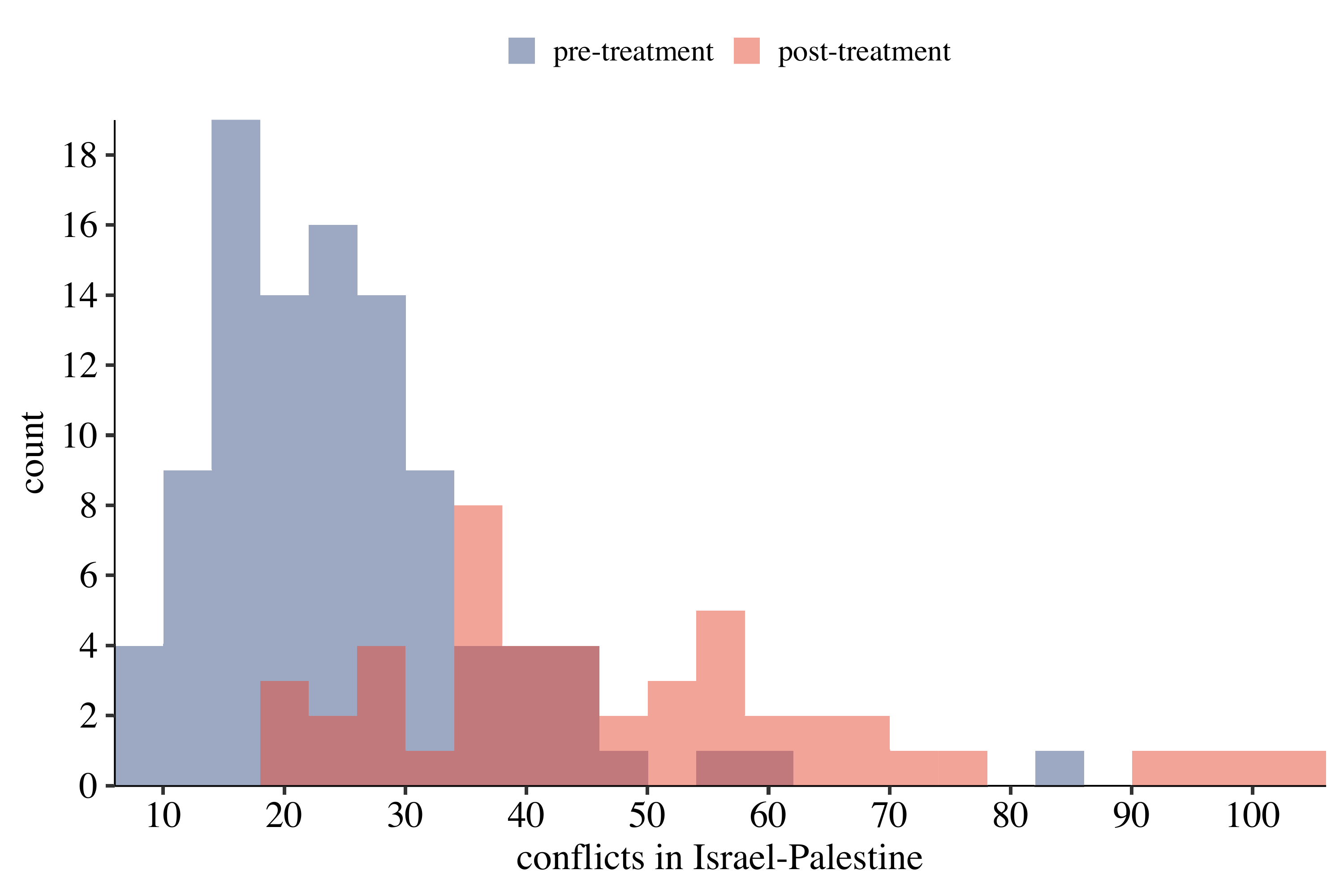}
		\end{adjustbox}
\caption{Distribution of Weekly Conflicts in Israel-Palestine}
\label{fig: hist_israel_palestine}
\vspace{0.166667in}
{\footnotesize
\textit{Notes:} Distribution of weekly
conflicts in Israel and Palestine pre-treatment (blue) and post-treatment
(red). The conflicts cover riots/protests, strategic development,
remote violence, violence against civilians, various types of battles,
and headquarters or base established.\par}
\end{figure}

\subsection{Results}
Our application is motivated by Figure \ref{fig: weekly aggregates of conflicts in Palestine and Israel},
showing the weekly number of conflicts in Israel-Palestine over the
entire sample period. The two vertical lines indicate the date when
the relocation of the US embassy was announced and the date of the actual
move, respectively, and also, we plot the average of the remaining
countries. A couple of observations are worth noting. First, visual
inspection suggests that the average weekly number of conflicts in
Israel-Palestine has in fact increased subsequent to the announcement.
In contrast, the average number of weekly conflicts over the remaining
countries in the Middle East does not appear to follow the same upward
shift after the announcement. We formalize this shortly.

\vspace{0.166667in}
Second, the volatility of the weekly number of conflicts in Israel-Palestine seems
much higher after the announcement, supporting the histogram in Figure
\ref{fig: hist_israel_palestine}. This has important economic implications,
as it indicates that conflicts tend to cluster and that misfortunes never
come singly. Considering the conflicts more closely, for instance
analyzing the degree of violence in the clusters, is interesting, but
we postpone this for future research.

\vspace{0.166667in}
Finally, note the large spike
in the average number of conflicts across the remaining countries
in the Middle East around July 2016. Specifically, the week with the
highest average number of conflicts runs from July 18 to July 24,
which is just after the military coup was attempted in Turkey on July
15 against state institutions, including the government and President
Erdo\u{g}an. During the coup, more than 2,100 people were injured
and over 300 were killed. This rare event shows up in the estimation
for some methods that are exposed to outliers.
\begin{figure}[!t]
		\begin{adjustbox}{max totalsize = {0.8\textwidth}{0.9\textheight}, center}
			\includegraphics[scale=1,width = \textwidth]{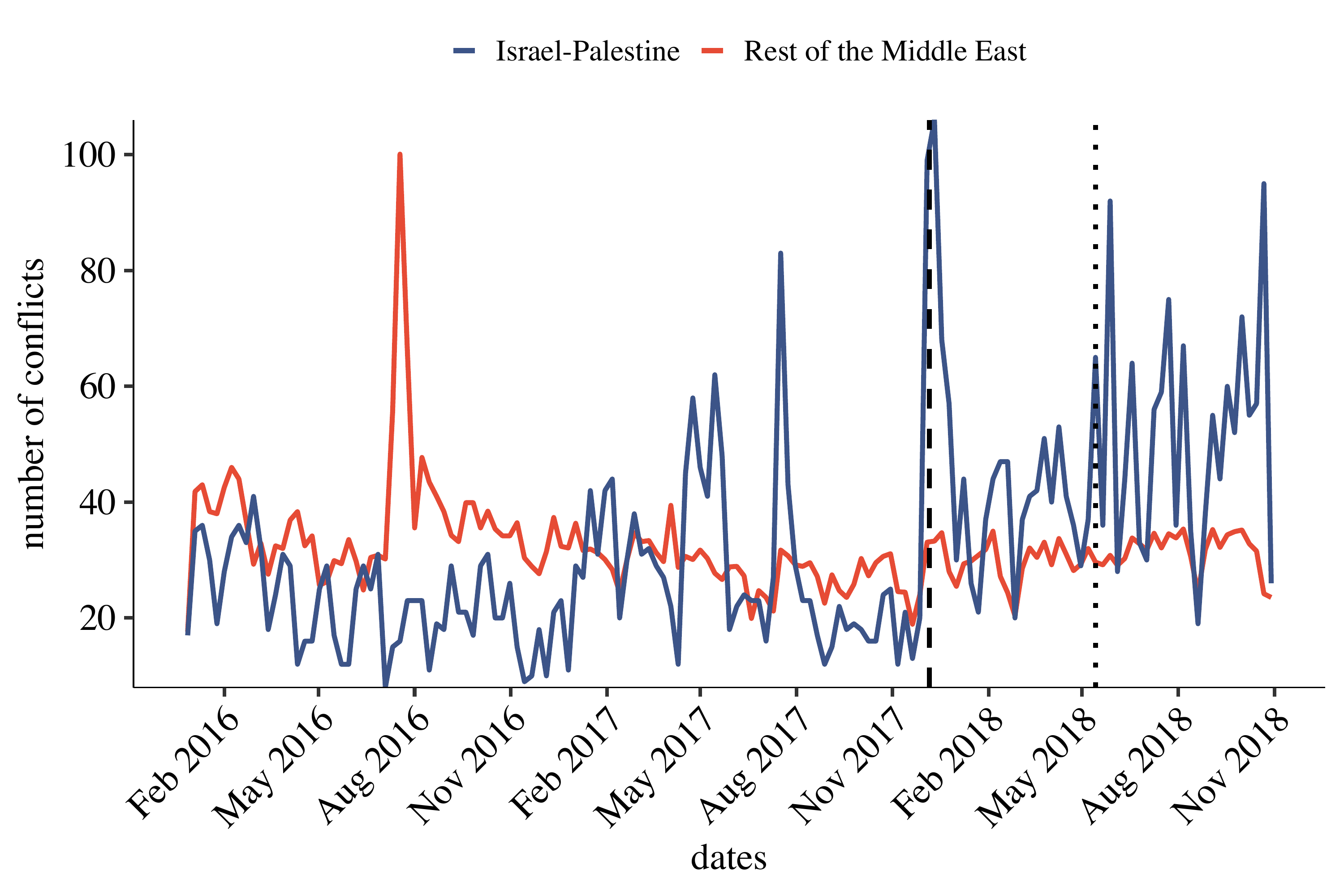}
		\end{adjustbox}
\caption{Weekly Number of Conflicts in Israel-Palestine and the Middle East}
\label{fig: weekly aggregates of conflicts in Palestine and Israel}
\vspace{0.166667in}
{\footnotesize
\textit{Notes:} Weekly number of conflicts
in Israel and Palestine (blue line) in addition to the average of
the remaining countries in the Middle East (red line). The vertical
dashed and dotted lines represent the date when the move of the US
embassy was announced and the date of the actual move, respectively.\par}
\end{figure}

\vspace{0.166667in}
Figure \ref{fig: weekly level of conflict Palestine-Israel and synthetic counterpart}
displays the weekly number of conflicts for Israel-Palestine and its
estimated counterpart during the period December 28, 2015, to November
3, 2018. The observed level of conflicts in Israel-Palestine is closely
followed by the estimated counterpart in the entire pre-intervention
period until the move was announced on December 4, 2017. This suggests
that the time periods before the announcement can be grouped together
into homogeneous subgroups based on the level of conflicts in the
neighboring countries, and for these subgroups of time periods, the
level of conflicts in Israel and Palestine is relatively constant.
In fact, the average of the observed weekly number of conflicts in
the pre-intervention period is 25.32, whereas the estimated counterpart
is 25.41, indicating an accurate fit on average. Note that the estimated
counterpart to Israel-Palestine is always closer to the average level
of weekly conflicts instead of capturing the spikes to the fullest
extent. The is an attractive feature of the averaging that happens
in our model as the model implicitly becomes conservative.

\vspace{0.166667in}
Altogether, we take this as evidence that the tree-based synthetic
control method can be used to predict a counterfactual Israel-Palestine,
which provides a sensible approximation to the true level of conflicts
that would have occurred in that region in absence of the move. Thus,
we next use the tree-based synthetic control method to estimate the
average treatment effect of moving the embassy.
\begin{figure}[!t]
		\begin{adjustbox}{max totalsize = {0.8\textwidth}{0.9\textheight}, center}
			\includegraphics[scale=1,width = \textwidth]{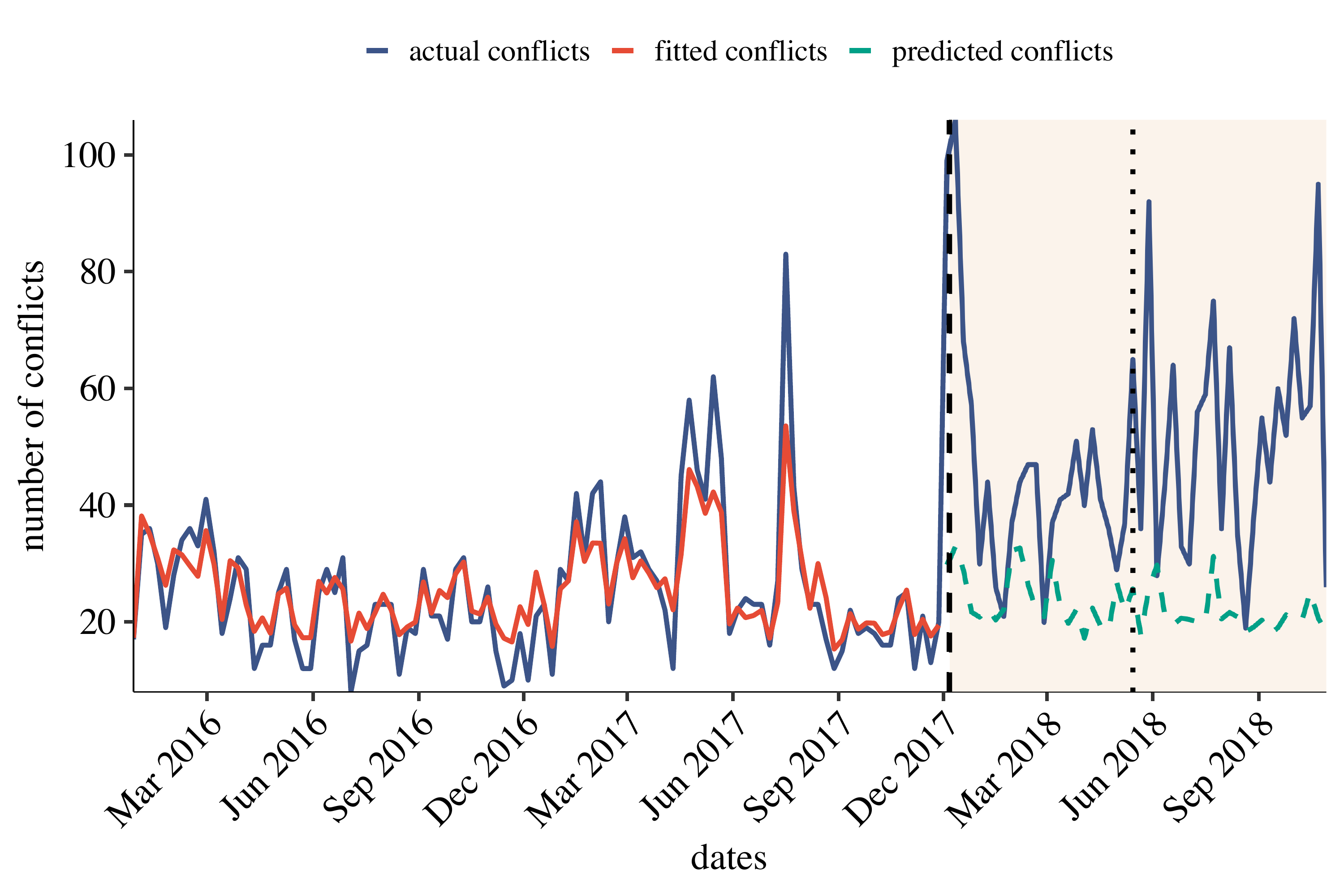}
		\end{adjustbox}
\caption{Weekly Number of Conflicts in Israel-Palestine and its Estimated Counterpart}
\label{fig: weekly level of conflict Palestine-Israel and synthetic counterpart}
\vspace{0.166667in}
{\footnotesize
\textit{Notes:} Weekly number of conflicts
in Israel and Palestine (blue line) and its estimated counterpart
in the pre-intervention period (red line) and post-treatment period
(green dashed line). The vertical dashed and dotted lines represent
the date when the move of the US embassy was announced and the date
of the actual move, respectively.\par}
\end{figure}

\vspace{0.166667in}
We estimate the effect of the relocation of the US embassy for each of the
48 weeks after the announcement as the difference between the observed
level of conflicts in Israel-Palestine and its counterfactual analog.
The differences follow as the discrepancies between the two lines
in the shaded area of Figure \ref{fig: weekly level of conflict Palestine-Israel and synthetic counterpart}.
Immediately after the move is announced, both the observed and counterfactual
level of conflicts increase, but to very different degrees, and in
fact the observed level of weekly conflicts in Israel and Palestine
reaches its maximum level across the entire sample within the first
week of the announcement.

\vspace{0.166667in}
For the rest of the post-announcement period,
the observed level of conflicts sees a higher base level with
distinctly conflict-ridden weeks, whereas the counterfactual Israel-Palestine
maintains the lower base level from the pre-announcement period. Specifically,
the average of the observed number of weekly conflicts in the post-intervention
period is 48.88, whereas the estimated counterpart is 22.78, indicating
a significant difference. This suggests that the relocation of the embassy
has a numerically positive effect on the level of conflicts in Israel
and Palestine, meaning that the level generally increases in the entire
post-announcement period.
\begin{figure}[!t]
		\begin{adjustbox}{max totalsize = {0.8\textwidth}{0.9\textheight}, center}
			\includegraphics[scale=1,width = \textwidth]{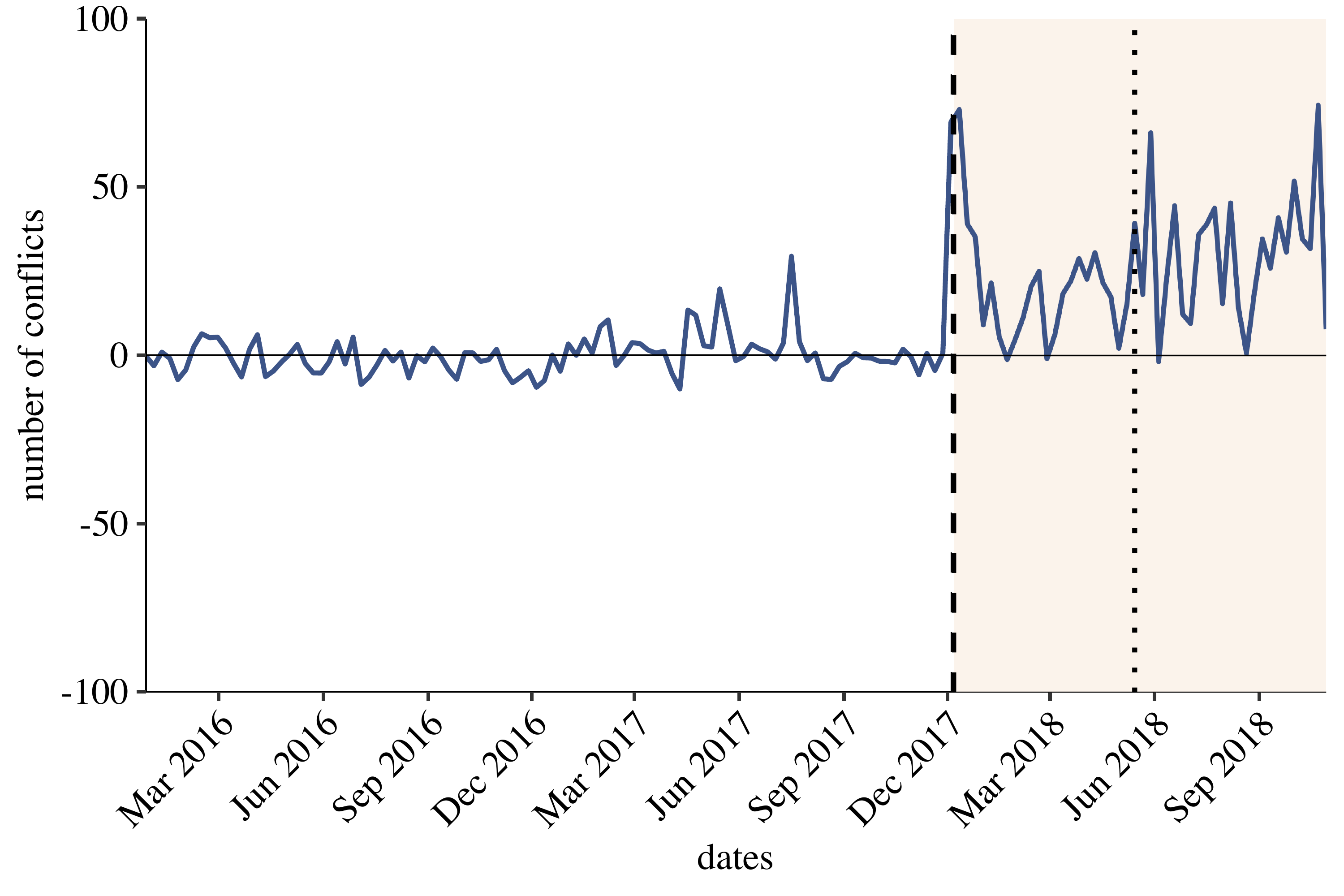}
		\end{adjustbox}
\caption{Discrepancies between the Observed and Estimated Conflicts in Israel-Palestine}
\label{fig: weekly gap tree}
\vspace{0.166667in}
{\footnotesize
\textit{Notes:} Weekly gaps between the
number of observed and estimated conflicts in Israel and Palestine.
The vertical dashed and dotted lines represent the date when the move
of the US embassy was announced and the date of the actual move, respectively.\par}
\end{figure}

\vspace{0.166667in}
We assess the weekly estimates of the impact directly in Figure \ref{fig: weekly gap tree},
where we plot the differences between the observed and estimated number
of weekly conflicts in Israel and Palestine. Figure \ref{fig: weekly gap tree}
unveils the same story as Figure \ref{fig: weekly level of conflict Palestine-Israel and synthetic counterpart}.
The gap of approximately zero on average in the pre-intervention period
indicates that the tree-based synthetic control method is able to
approximate well the true level of conflicts, albeit very fluctuating.
To be precise, the average difference between the observed and estimated
weekly number of conflicts in the pre-intervention period is only
-0.09. This is obviously a heuristic as average discrepancy does not per se capture accurate fit, and this number should considered jointly with Figure \ref{fig: weekly gap tree}.

\vspace{0.166667in}
Using all 48 weeks after the announcement, our results show
that the level of conflicts in Israel and Palestine is increased by
an average of more than 26 incidents per week, which corresponds to
an increase of approximately 103\%. The estimated average effect is
associated with a bootstrapped standard error of 2.67 using 10,000
block bootstrap samples with block length equal to 3. That is, the
95\% bootstrap confidence interval of the weekly increase is between
20.88 and 31.36. This translates into a percentage point change between
roughly 82-124\%. We acknowledge that the confidence interval is rather
wide, which is not surprising due to the volatility in the number
of conflicts across weeks. The results are insensitive to the choice
of block length.

\vspace{0.166667in}
Naturally, the assumption of no interference between the treated and
control units can be violated in several ways in the context of analyzing
the effect of moving the US embassy. The Israeli-Palestinian conflict
is an issue in all of the region, and the ties between the countries
are complex to understand. For instance, we choose to exclude Iran
in the sample, because the Iranian government has played an active
role in the conflict. The results with and without Iran are, however,
not significantly different, because the tree-based synthetic control
method averages over the number of conflicts in Israel-Palestine and
uses only the neighboring countries, i.e., the controls, to stratify the time periods.

\vspace{0.166667in}
This feature of the method makes it more robust to the potential violations
compared to methods that base the estimates on the outcomes for the
control units. Further, the average weekly number of conflicts across
all control countries does not differ between the pre- and post-intervention
period. In particular, the average over the control countries in the
pre-intervention period is 32.80, whereas the same figure is 30.82
in the post-intervention period. The small difference is likely to
be driven by the coup attempt in Turkey.

\vspace{0.166667in}
The placebo tests we review shortly reveal that no other relevant country experienced the same
effect of the relocation of the US embassy. Last, the conformal inference
test in Section \ref{subsec:Inference} provides evidence that our
model is correctly specified and that the increase is statistically
significant. Taken altogether, it is our judgment that the potential
violations do not appear to be severe in this context. 

\subsection{Inference}\label{subsec:Inference}
We want to assess how much our results are driven by mere chance.
If we are able to obtain estimated effects of the same magnitude for
the control countries as for Israel-Palestine by relabeling treatment
and control unit, we would not be able to interpret our analysis as
providing any significant effects. To make inference about the effect
of the embassy relocation, we follow the strategy outlined in \cite{abadie2010synthetic},
\cite{Bertrand2004}, and \cite{Abadie2003} and run placebo tests.

\vspace{0.166667in}
Placebo tests re-do the original analysis, but switch the roles between
the treated unit and a randomly chosen control unit, the rationale
being that using the control unit not exposed to treatment should
lead to an estimated effect of approximately zero. By applying the
tree-based synthetic control method individually to all the countries
in the donor pool, we can therefore evaluate the significance of our
analysis. We expect one of two outcomes. If the placebo tests deliver
estimates of the average effect of similar magnitude as for Israel-Palestine,
we cannot rightfully interpret our results as evidence for a significant
effect. If, on the other hand, that none of the placebo tests for
the countries in which the US embassy was not moved lead to similar
estimated effects, then we take this as evidence that our tree-based
analysis documents a significant effect of moving the US embassy in
terms of an increased level of conflicts. One important condition, however,
is that the pre-intervention fit to the weekly number of conflicts
is precise for the country in question when we run the placebo test.

\vspace{0.166667in}
To assess the significance of our estimates, we perform a series of
placebo test for which we create a counterfactual state of the world.
That is, we iteratively treat each control country in the remaining
Middle East as if it had experienced a move of the US embassy at exactly
the same time as the move in Israel, while we also reassign both Israel
and Palestine to the control group. In each iteration, we apply tree-based
controls to the respective country to estimate the impact of the fictive
embassy move on the weekly number of countries. The series of placebo
tests gives us a distribution of differences between the observed
and estimated number of conflicts over the countries.
\begin{figure}[!t]
		\begin{adjustbox}{max totalsize = {0.8\textwidth}{0.9\textheight}, center}
			\includegraphics[scale=1,width = \textwidth]{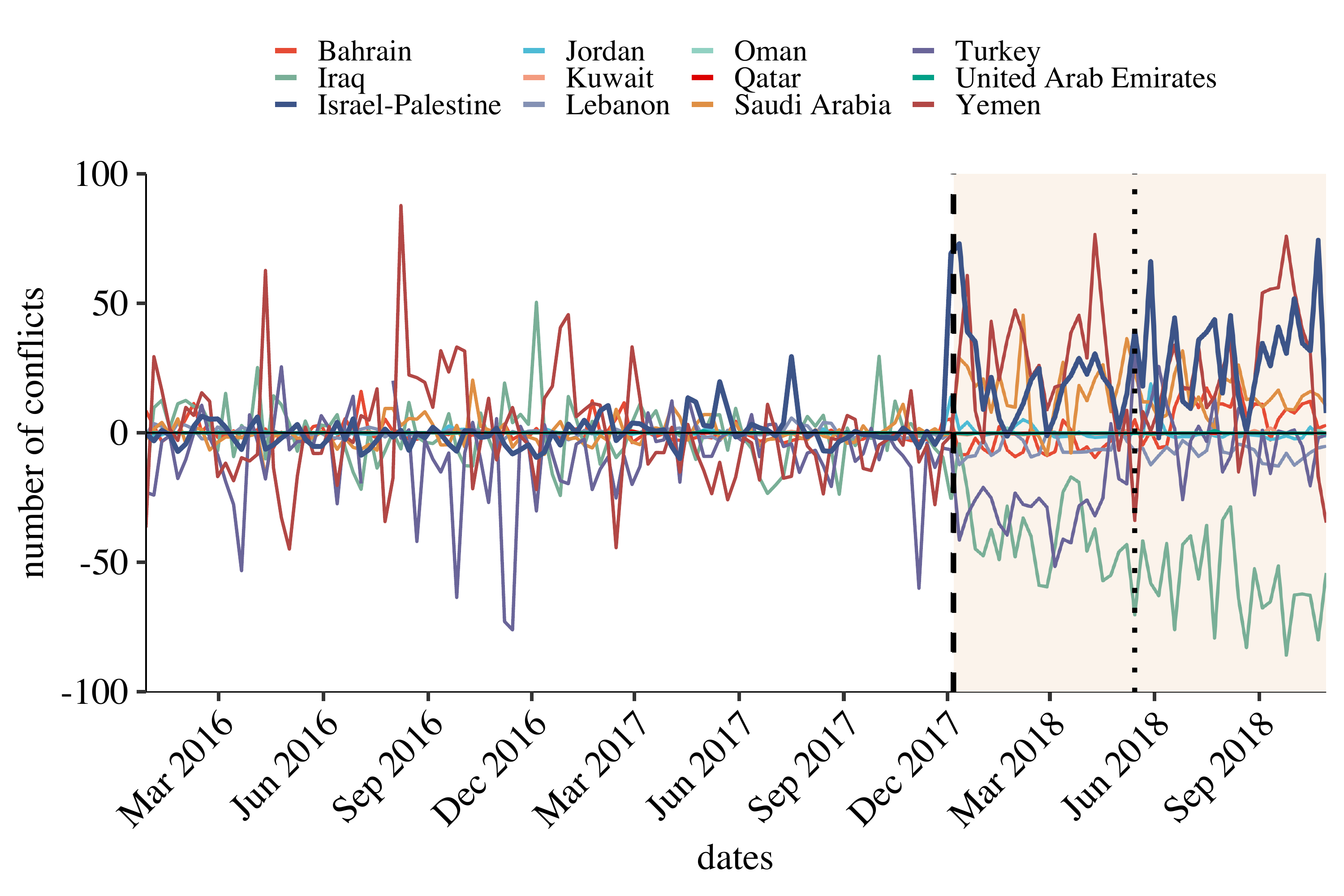}
		\end{adjustbox}
\caption{Discrepancies between the Observed and Estimated Conflicts in the
Middle East}
\label{fig: placebo all}
\vspace{0.166667in}
{\footnotesize
\textit{Notes:} Weekly gaps between the
number of observed and estimated conflicts for all countries considered
in the placebo tests. The blue line represents the differences for
Israel and Palestine, whereas the other lines represent the differences
for the control units defined temporarily as treated units. The vertical
dashed and dotted lines represent the date when the move of the US
embassy was announced and the date of the actual move, respectively.\par}
\end{figure}

\vspace{0.166667in}
Figure \ref{fig: placebo all} plots the differences in the observed
and estimated number of conflicts for all the placebo analyses and
the original analysis. The blue line shows the case for Israel-Palestine,
reproducing Figure \ref{fig: weekly gap tree}. The other lines show
the same differences estimated by the tree-based synthetic control
method, but for each of the 11 control countries in the donor pool.
Figure \ref{fig: placebo all} indicates that the tree-based synthetic
control method provides an accurate fit in the pre-intervention period
for Israel and Palestine as well as for most of the control countries.

\vspace{0.166667in}
In particular, the pre-intervention root mean squared prediction error
(RMSPE) for Israel-Palestine is 5.77, where RMSPE is computed as the
root average of the squared differences between the observed and estimated
weekly number of conflicts. The pre-intervention median RMSPE for
the control countries is 1.71. This should not be taken as evidence
that the ability to fit the pre-intervention is higher for the control
countries than for Israel-Palestine. In fact, mean RMSPE over the
control countries is 9.51, indicating that a few control countries
stand out in terms of high RMSPE, while for most control countries,
we achieve a very low RMSPE. This is supported by Figure \ref{fig: placebo all}
from which it is apparent that the pre-intervention fit is very imprecise
for some countries.

\vspace{0.166667in}
The country with the worst fit is Turkey with
an RMSPE of 61.88. This result, however, is not surprising due to
the attempted military coup in 2016 that led to an extreme spike in
the number of conflicts. As this coup attempt was, of course, unanticipated,
the conflict situation in the other countries was normal, and therefore,
no statistical method would be able to capture this outlier. Similar
problems arise for Iraq and Yemen, which are the countries with the
overall highest variation in the weekly number of conflicts. This
high variation makes it difficult for the tree-based synthetic control
method, and likely any other method, to produce a valid fit in the
pre-intervention period without imposing too much flexibility. As
a result, the RMSPE for Turkey, Iraq, and Yemen are all more than
double that of Israel-Palestine and any other control country.
\begin{figure}[!t]
		\begin{adjustbox}{max totalsize = {0.8\textwidth}{0.9\textheight}, center}
			\includegraphics[scale=1,width = \textwidth]{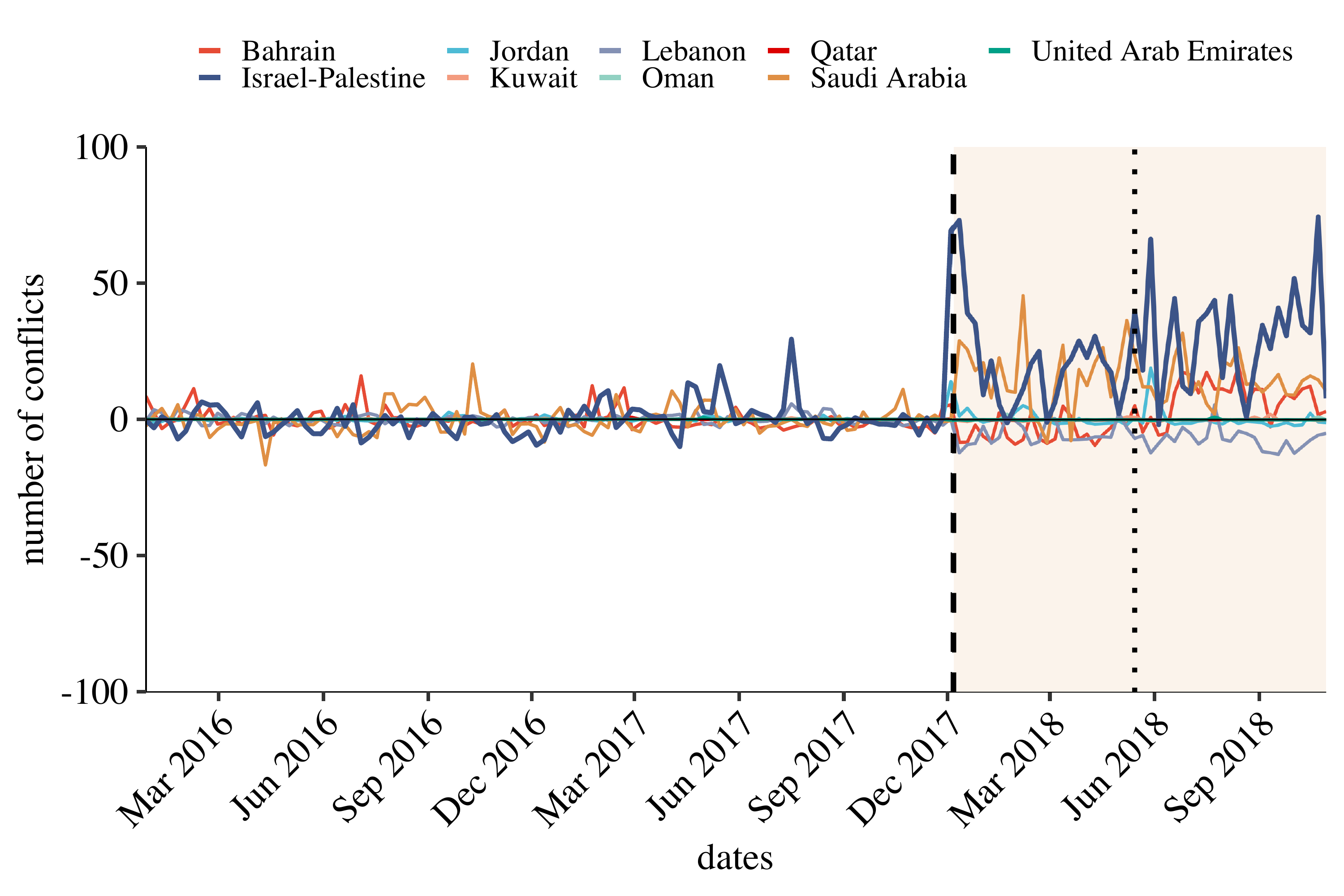}
		\end{adjustbox}
\caption{Discrepancies between the Observed and Estimated Conflicts in the
Middle East (excl. Turkey, Iraq, and Yemen)}
\label{fig: placebo restrictive}
\vspace{0.166667in}
{\footnotesize
\textit{Notes:} Weekly gaps between the
number of observed and estimated conflicts for all countries considered
in the placebo tests except Turkey, Iraq, and Yemen. The blue line
represents the differences for Israel and Palestine, whereas the other
lines represent the differences for the control units defined temporarily
as treated units. The vertical dashed and dotted lines represent the
date when the move of the US embassy was announced and the date of
the actual move, respectively.\par}
\end{figure}

\vspace{0.166667in}
To handle the countries for which the tree-based synthetic control
method gives a poor fit, we follow an argument provided in \cite{abadie2010synthetic}
as they encounter the same issue for some of the states. If the tree-based
synthetic control method had failed to deliver a reasonable fit to
the observed weekly level of conflicts in the pre-intervention period
for Israel-Palestine, we would treat the lack of fit as evidence that
the estimated increase in the weekly number of conflicts in the post-intervention
period was arbitrary and not caused by the move of the US embassy.
Analogously, we cannot take into account the estimated effects in
the post-intervention period for Turkey, Iraq, and Yemen when assessing
the degree of chance in our results for Israel-Palestine.

\vspace{0.166667in}
Consequently, we provide another version of Figure \ref{fig: placebo all} in which
we have excluded the placebo tests for Turkey, Iraq, and Yemen. This
effectively corresponds to removing countries for which the RMSPE
is more than double the one for Israel-Palestine. Figure \ref{fig: placebo restrictive}
provides the restricted version of Figure \ref{fig: placebo all}
from which we have excluded Turkey, Iraq, and Yemen. The median RMSPE
over the remaining countries in the Middle East drops to 0.35, and
the corresponding mean drops to 1.37.

\vspace{0.166667in}
Removing the countries for which
the tree-based synthetic control method would be ill-advised tells
a clear message. The largest estimated effect on the weekly number
of conflicts in the post-intervention period is to be found for Israel-Palestine.
More precisely, while the average estimated effect for Israel-Palestine
is 26.12 in the post-intervention period, the corresponding figure
over the placebo tests is 1.38. For the pre-intervention period, the
estimated gaps are -0.09 and -0.02, respectively.

\vspace{0.166667in}
The use of placebo tests as a mode of inference for synthetic controls is heavily debated (see, e.g., \cite{Hahn2017}). We emphasize that placebo tests evaluate significance
\emph{relative} to a benchmark distribution (here, a uniform distribution) for the given assignment
mechanism in the data. This permutation-based test is conditional on the sample and exploits the randomness induced by the (placebo) assignment mechanism.
In contrast, sample-based tests are conditional on the assignment mechanism and exploit the randomness in the DGP. But because the sample mechanism is not well-defined and the sample is in fact the population in the cross section (all countries in the Middle East are included), sample-based tests are often complicated in this settings \citep{abadie2019using}.\footnote{We thank Alberto Abadie for pointing this out.}
\begin{table}[!t]
\caption{Summary of Performance Measures across Countries Pre-treatment and
Post-treatment}
\begin{adjustbox}{max width = \textwidth, center}
\begin{tabular}{llcccccccc}
\toprule 
 &  & \multicolumn{2}{c}{Ratio} &  & \multicolumn{2}{c}{Pre-intervention} &  & \multicolumn{2}{c}{Post-intervention}\tabularnewline
\cmidrule{3-4} \cmidrule{4-4} \cmidrule{6-7} \cmidrule{7-7} \cmidrule{9-10} \cmidrule{10-10} 
 &  & MAE & RMSPE &  & MAE & RMSPE &  & MAE & RMSPE\tabularnewline
\midrule
Israel \& Palestine &  & 6.59 & 5.61 &  & 3.99 & 5.77 &  & 26.28 & 32.38\tabularnewline
Bahrain &  & 3.03 & 2.40 &  & 7.40 & 3.58 &  & 2.44 & 8.61\tabularnewline
Iraq &  & 5.84 & 4.57 &  & 48.64 & 11.36 &  & 8.33 & 51.89\tabularnewline
Jordan &  & 5.12 & 6.78 &  & 2.09 & 0.57 &  & 0.41 & 3.89\tabularnewline
Kuwait &  & 3.60 & 3.47 &  & 0.22 & 0.13 &  & 0.06 & 0.43\tabularnewline
Lebanon &  & 4.92 & 4.32 &  & 6.60 & 1.71 &  & 1.34 & 7.37\tabularnewline
Oman &  & 8.40 & 7.38 &  & 0.07 & 0.04 &  & 0.00 & 0.32\tabularnewline
Qatar &  & 3.15 & 0.97 &  & 0.05 & 0.07 &  & 0.02 & 0.07\tabularnewline
Saudi Arabia &  & 4.40 & 3.76 &  & 15.46 & 4.80 &  & 3.52 & 18.06\tabularnewline
Turkey &  & 0.86 & 0.37 &  & 18.44 & 61.88 &  & 21.41 & 22.94\tabularnewline
UAE &  & 5.31 & 2.23 &  & 0.08 & 0.08 &  & 0.01 & 0.17\tabularnewline
Yemen &  & 1.91 & 1.68 &  & 28.48 & 20.35 &  & 14.95 & 34.20\tabularnewline
\bottomrule
\end{tabular}
\end{adjustbox}\vspace{0.166667in}
\label{tab: ratios}
{\footnotesize 
\textit{Notes:} Summary of measures used
to assess the significance of the results obtained for Israel and
Palestine. Measures include mean absolute error and root mean squared
prediction error between the observed and estimated weekly number
of conflicts for both the pre- and post-intervention period. We also
include the ratios of post-/pre-intervention measures. All measures
are reported for Israel and Palestine, and for each of the placebo
runs.\par}
\end{table}

\vspace{0.166667in}
We consider another approach to assessing the significance of our
results, namely computing ratios of post-/pre-intervention measures
both for Israel-Palestine and the control countries. As \cite{abadie2010synthetic},
we compute the ratios in terms of RMSPE. Arguably, the advantage of
comparing ratios relative to post-intervention gaps is that we do
not necessarily have to exclude ill-fitting placebo runs in an iterative
way as demonstrated by figures \ref{fig: placebo all} and \ref{fig: placebo restrictive}.
For instance, although the RMSPE for Turkey is the highest across
all in the pre-intervention period, it is similarly high in the post-intervention
period, and the ratio will be more robust to this. 

\vspace{0.166667in}
The only countries with a higher ratio of post-/pre-intervention RMSPE than Israel-Palestine
are Jordan and Oman. This observation, however, does not cause much
concern when we take into account the gaps in both periods. For Jordan,
the pre-intervention gap between the observed and estimated weekly
number of conflicts is -0.02, whereas the same figure is 0.46 in the
post-intervention period. Likewise, the figures for Oman are -0.00
and 0.06, respectively. Thus, the high ratios of post/pre-intervention
RMSPE for the two countries are likely driven by a few very conflict-ridden
weeks after the intervention.

\vspace{0.166667in}
In addition to the ratios of post-/pre-intervention
RMSPE used in \cite{abadie2010synthetic}, we also compute the ratios
of post-/pre-intervention mean absolute error (MAE) between the observed
and estimated weekly number of conflicts. Using either the ratio of
post-/pre-intervention RMSPE or MAE has different advantages. RMSPE
penalizes large errors more than MAE, but MAE is more interpretable.
We provide both ratios for each country in Table \ref{tab: ratios},
in which we also provide the respective pre- and post-intervention
measures. Note from Table \ref{tab: ratios} than Oman is the only
country with a higher ratio of post-/pre-intervention MAE than Israel-Palestine.
In absolute terms, again, the result for Oman is not too disturbing
for our analysis.

\subsubsection{Exact and robust conformal inference}
We consider one last approach to draw inference about our results.
Recall that our proposed method as well as the other methods considered
relies on cross-sectional regressions. Whenever the joint distribution
of the data is not well-approximated by cross-sectional regressions,
the model will provide a poor global fit in the sense that not all
$N$ controls will fit the model, which is exactly the case in our
application as well as in \cite{abadie2010synthetic}. In this situation,
\cite{Chernozhukov2017b} propose an exact and robust conformal inference
method along with an associated validity test.  In the following, we rely on the validity test of the needed assumptions rather than going into the theoretical aspects.

\vspace{0.166667in}
The method requires only a good \emph{local} instead of a
good \emph{global} fit, as it relies solely on a suitable model for
the treated unit and it focuses on the time-series dimension. Essentially,
the procedure postulates a null trajectory and tests the sharp null hypothesis $\mathcal{H}_{0}:\tau_t=\tau_t^{o}$ for $t=T_0+1,\ldots,T$. For the test to be valid, the estimator of the counterfactual outcome
for the treated unit needs to be consistent and stable and be able to provide residuals that are exchangeable.

\vspace{0.166667in}
To assess the plausibility of the key assumptions, \cite{Chernozhukov2017b} provide placebo
specification tests. The conditions result in non-asymptotic validity
of the test, meaning that the $p$-value is approximately unbiased
in size \citep[Theorem 1, p. 23,][]{Chernozhukov2017b}. The proposed
inference method is valid for stationary and weekly dependent data.
\begin{table}[!t]
\caption{Placebo Specification Test}
\begin{adjustbox}{max width = \textwidth, center}
\begin{tabular}{llllllllllllllll}
\toprule
\textbf{Placebo Specification} &  & & & & & & & & & \\
\midrule
$\kappa$ & 1 & 2 & 3 & 4 & 5 & 6 & 7 & 8 & 9 & 10 \\
i.i.d. Perm. & 0.902 & 0.664 & 0.850 & 0.678 & 0.832 & 0.883 & 0.902 & 0.933 & 0.952 & 0.974 \\ 
Moving Block Perm. & 0.901 & 0.594 & 0.782 & 0.614 & 0.762 & 0.812 & 0.851 & 0.891 & 0.901 & 0.941 \\
\bottomrule
\end{tabular}
\end{adjustbox}\vspace{0.166667in}
\label{tab:placebo_specification_test}
{\footnotesize 
\textit{Notes:} Placebo specification
test $p$-values over varying $\kappa$ from 1 to 10 based on both
the i.i.d. and the moving block permutations. We fail to reject the
null hypothesis at any significance level above 60\textbackslash\%.
Failure to reject the null hypothesis provides evidence for correct
specification. In the  i.i.d. case, we randomly sample 10,000 elements
from the set of all permutations with replacement.\par}
\end{table}

\vspace{0.166667in}
We are interested in testing the hypothesis that the trajectory of
the policy effects in the post-treatment is zero. Hence, our main
hypothesis is
\begin{equation}
\mathcal{H}_{0}:\tau_t=\tau_t^{o},\quad\text{for }t=T_0+1,\ldots,T\label{eq:main_null}
\end{equation}
The test statistic $S$ is based on the $\left(\left(T-T_{0}\right)\times1\right)$
vector of residuals of our model $\hat{u}_t$ for $t=T_0+1,\ldots,T$.
The test statistic is then defined by
\begin{equation}
S\left(\hat{u}_{T_0+1},\ldots,\hat{u}_{T}\right)=S_{q}\left(\hat{u}_{T_0+1},\ldots,\hat{u}_{T}\right)=\left(\frac{1}{\sqrt{T-T_{0}}}\sum_{t=T_{0}+1}^{T}\left|\hat{u}_{t}\right|^{q}\right)^{\nicefrac{1}{q}},\label{eq:test_statistic}
\end{equation}
where we set $q=1$. To compute $p$-values, the test relies on two
different sets of permutations, the i.i.d permutations denoted $\Pi_{\text{i.i.d}}$
and the moving block permutations denoted $\Pi_{\rightarrow}$. The
moving block permutations are necessary if the sequence of residuals
exhibits serial dependence. The $p$-value is estimated as $\hat{p}=1-\hat{F}\left(S\left(\hat{u}_{T_0+1},\ldots,\hat{u}_{T}\right)\right)$,
where 
\begin{equation}
\hat{F}\left(x\right)=\frac{1}{\left|\Pi\right|}\sum_{\pi\in\Pi}\mathds{1}\left\{S\left(\hat{u}_{\pi,T_0+1},\ldots,\hat{u}_{\pi,T}\right) <x)\right\} .\label{eq:cdf_test}
\end{equation}
To assess the validity of the assumptions underlying the test, the
first step is to perform a placebo specification test. Based on the
outlined procedure, the idea is to test the null hypothesis that
\begin{equation}
\mathcal{H}_{0}:\tau_{T_{0}-\kappa+1}=\cdots=\tau_{T_{0}}=0\label{eq:placebo_null}
\end{equation}
for a given $\kappa\geq1$ based on pre-treatment data. The null hypothesis
\eqref{eq:placebo_null} is true if the underlying assumptions are
correct. Thus, rejecting the null provides evidence against a correct
specification. For proofs and additional details, we refer to \cite{Chernozhukov2017b}.\footnote{Note that \cite{Chernozhukov2017b} also provide a test for the average effect over time. However, this requires the total number of periods
to be much larger than the post-treatment periods, which is not the case in our application.}

\vspace{0.166667in}
We begin the analysis by testing the underlying assumptions of our
proposed method, i.e., consistency, stability, and exchangeability
of the residuals. We apply both i.i.d. permutations and the moving
block permutations. We use $\kappa=10$ and randomly sample 10,000
elements from the set of all permutations with replacement for the
i.i.d. permutations. The resulting $p$-values follow from Table \ref{tab:placebo_specification_test}.
All $p$-values from both permutation schemes are above 60\% and most
of them are above 80\%, and thus, we fail to reject the null hypothesis.
This serves as evidence for a correct model specification. We further
see that the $p$-values differ slightly between the i.i.d. permutations
and the moving block permutations, where the $p$-values tend to be
lower using moving block permutations. This provides evidence for
some serial dependence in the residuals.

\vspace{0.166667in}
Next, we turn to our main hypothesis in \eqref{eq:main_null}. We
consider again both the i.i.d. permutations with 10,000 random samples
as well as the moving block permutations. The $p$-value based on
the i.i.d. permutations is 0.000, whereas the $p$-value based on
the moving block permutations is 0.007. We reject the null hypothesis
in both cases given both $p$-values are below 1\%, providing evidence
that the trajectory of the policy effects from the embassy relocation is
different from zero. The formal test results thus appear to be in
agreement with the other inference results provided in this section. 

\section{Comparing Methods}\label{sec:Comparing-Methods}
In Section \ref{sec:Estimating-the-Consequences}, we provide evidence
that the decision to move the US embassy from Tel Aviv to Jerusalem
has resulted in a significant increase in the weekly number of conflicts
in Israel and Palestine. We assess the robustness of our results in
several ways, including performing formal inference tests, conducting
a series of placebo runs, and evaluating the fit on different measures
such as ratios of post-/pre-intervention RMSPE and MAE. In this section,
we compare the tree-based synthetic control method to three state-of-the-art
methods in the econometric literature.  We begin by introducing the methods.

\subsection{Competing methods}
\citet{abadie2010synthetic} also consider a version of \eqref{tauHat},  but assume linearity of $f$ in $X_{t}$.  In particular, \citet{abadie2010synthetic} assume that there exists a set of perfect weights
$\omega^{*}=\left(\omega_{1}^{*},\ldots,\omega_{N}^{*}\right)'$
such that $\left\langle \omega^{*},X_{t}\right\rangle =Y_{t}$ for $t=1,\ldots,T_0\le T_{0}$.
Considering $Y_{t}-\left\langle \omega^{*},X_{t}\right\rangle $, \citet{abadie2010synthetic}
prove that its mean is approximately zero under standard conditions,
which suggests using $\hat{\tau}_{t}=Y_{t}-\left\langle \omega^{*},X_{t}\right\rangle $
as an estimator for $\tau_{t}$ in periods $t>T_{0}$. The weights
are then estimated by 
\begin{equation}
\hat{\omega}=\arg\min_{\omega\in\mathbb{R}^{N}}\left\{ \sum_{t=1}^{T_{0}}\left(Y_{t}-\left\langle \omega,X_{t}\right\rangle \right)^{2}\right\} \quad\text{st. }\sum_{i=1}^{N}\omega_{i}=1,\;\omega_{i}\ge0\quad \text{for } i=1,\ldots,N,\label{eq: Synthetic control weights}
\end{equation}
which in practice can estimated by constrained least squares. 
The synthetic control method is mainly tailored for empirical settings
with relatively more time periods than control units, i.e., $T_0 \gg N$.

\vspace{0.166667in}
\citet{Doudchenko2016} propose a regularized extension to synthetic
controls, namely the elastic net estimator. The optimization problem
is similar to \eqref{eq: Synthetic control weights} but adds a regularization
term to the objective function with inspiration from shrinkage estimation.
Let $\left(\lambda,\alpha\right)\in\mathbb{R}\times\mathbb{R}$ be
a given pair of hyperparameters to be tuned, and let $\mu\in\mathbb{R}$
be an intercept, capturing the possibility that the outcomes for the
treated unit are systematically different from the other units. Then,
\citet{Doudchenko2016} propose to estimate the weights by 
\begin{equation}
\left(\mu,\hat{\omega}\right)=\arg\min_{\mu,\omega}\left\{ \sum_{t=1}^{T_{0}}\left(Y_{t}-\mu-\left\langle \omega,X_{t}\right\rangle \right)^{2}+\lambda\left(\frac{1-\alpha}{2}\sum_{i=1}^{N}\omega_{i}^{2}+\alpha\sum_{i=1}^{N}\left|\omega_{i}\right|\right)\right\} .\label{eq: Regularized Synthetic Control Method weights}
\end{equation}
Note that \eqref{eq: Regularized Synthetic Control Method weights}
neither requires zero intercept, weights summing to one, nor non-negative
weights. The elastic net estimator enjoys the selection property known
from Lasso by the $\ell_{1}$-penalty term \citep{Tibshirani1996,Zou2005}.
Essentially, some weights are likely to be zeroed out, meaning that
some control units are not predictive of the treated unit.

\vspace{0.166667in}
Both the synthetic control and the elastic net estimator may be viewed
as cross-sectional regressions in which the outcome of the treated
unit is regressed on the outcomes of the control units in the pre-treatment
period. Assuming stability over time, the cross-sectional pattern
is then carried over into the post-treatment period, based on which
the counterfactual outcome for the treated unit is predicted using
the control units. This form of regression in causal panel data models
is known as vertical regressions, a term coined by \citet{Athey2018matrix}.
The (almost) symmetric formulation is known as horizontal regressions,
where the post-treatment outcomes are regressed on the pre-treatment
outcomes using only the control units. This time-series approach estimates
a relationship which is then applied to the treatment unit assuming
stability across units and requires $N\gg T$. It is not a symmetric
problem because the order of $T$ matters in contrast to the order
of $N$.

\vspace{0.166667in}
However, both methods have a disadvantage in cases with $T\approx N$
as they do not fully exploit the panel structure by running either
cross-sectional or time-series regressions. A recent approach to causal
panel data models that takes both sources of variation into account
is the matrix completion method by \citet{Athey2018matrix}, treating
$Y_{t}^{0}$ for $t>T_{0}$ as missing.  We are now ready to compare the methods introduced.

\subsection{Comparison}
First, we apply the synthetic control method, serving as a baseline model. Then, we apply the regularized
counterpart, i.e., the elastic net estimator.  The matrix completion method combines elements from vertical
and horizontal regressions, and it is the last method we include.
\begin{figure}[!t]
\begin{subfigure}{.5\textwidth}
		\begin{adjustbox}{max totalsize = {\textwidth}{0.9\textheight}, center}
			\includegraphics[width = \textwidth]{Figures/counterfactual_tree}
		\end{adjustbox}
		\caption{Estimation by tree-based controls}
		\label{Fig:counterfactual_tree}
	\end{subfigure}\hfill
\begin{subfigure}{.5\textwidth}
		\begin{adjustbox}{max totalsize = {\textwidth}{0.9\textheight}, center}
			\includegraphics[width = \textwidth]{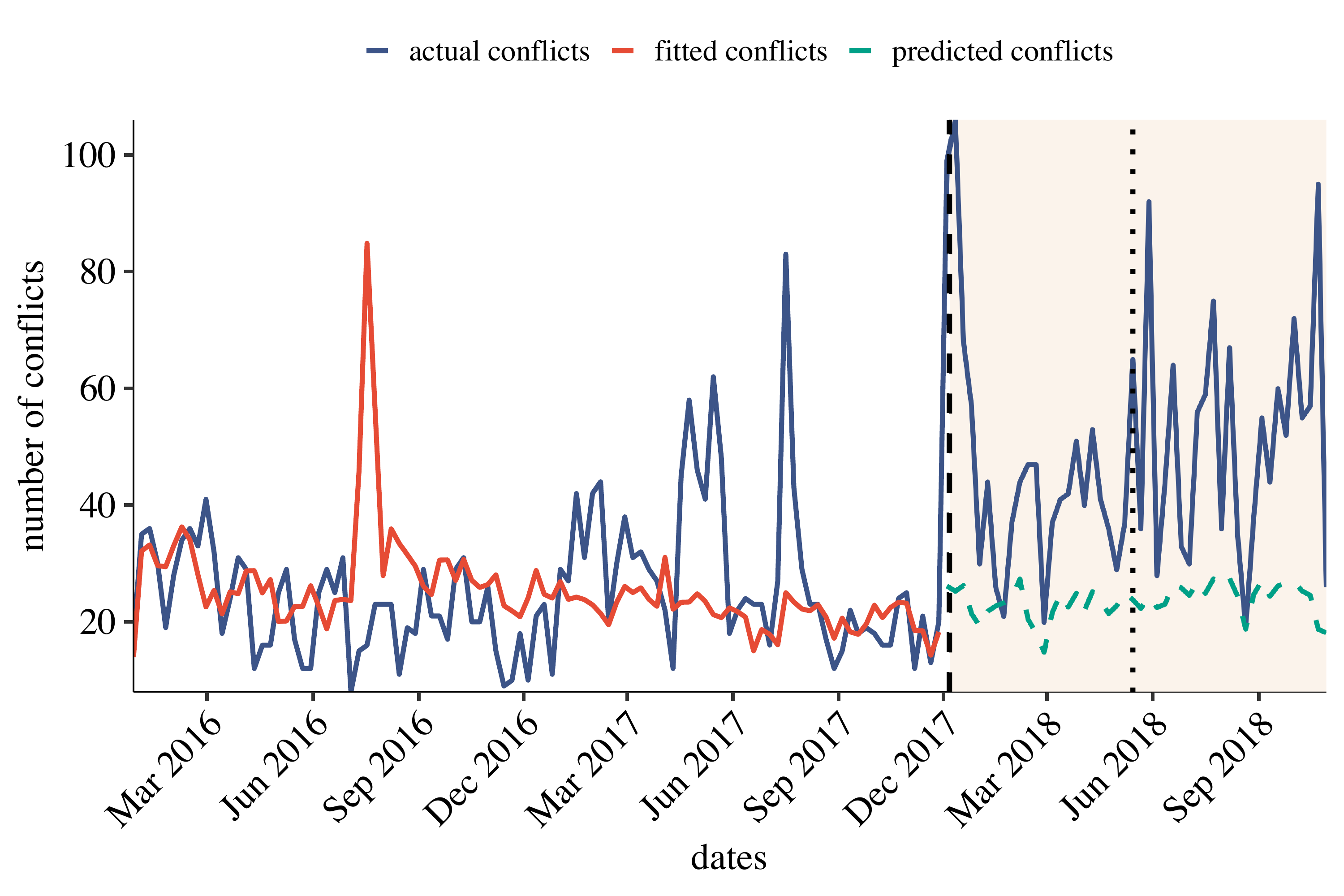}
		\end{adjustbox}	
		\caption{Estimation by synthetic controls}
	\label{Fig:counterfactual_synthetic}
	\end{subfigure}
\begin{subfigure}{.5\textwidth}
		\begin{adjustbox}{max totalsize = {\textwidth}{0.9\textheight}, center}
			\includegraphics[width = \textwidth]{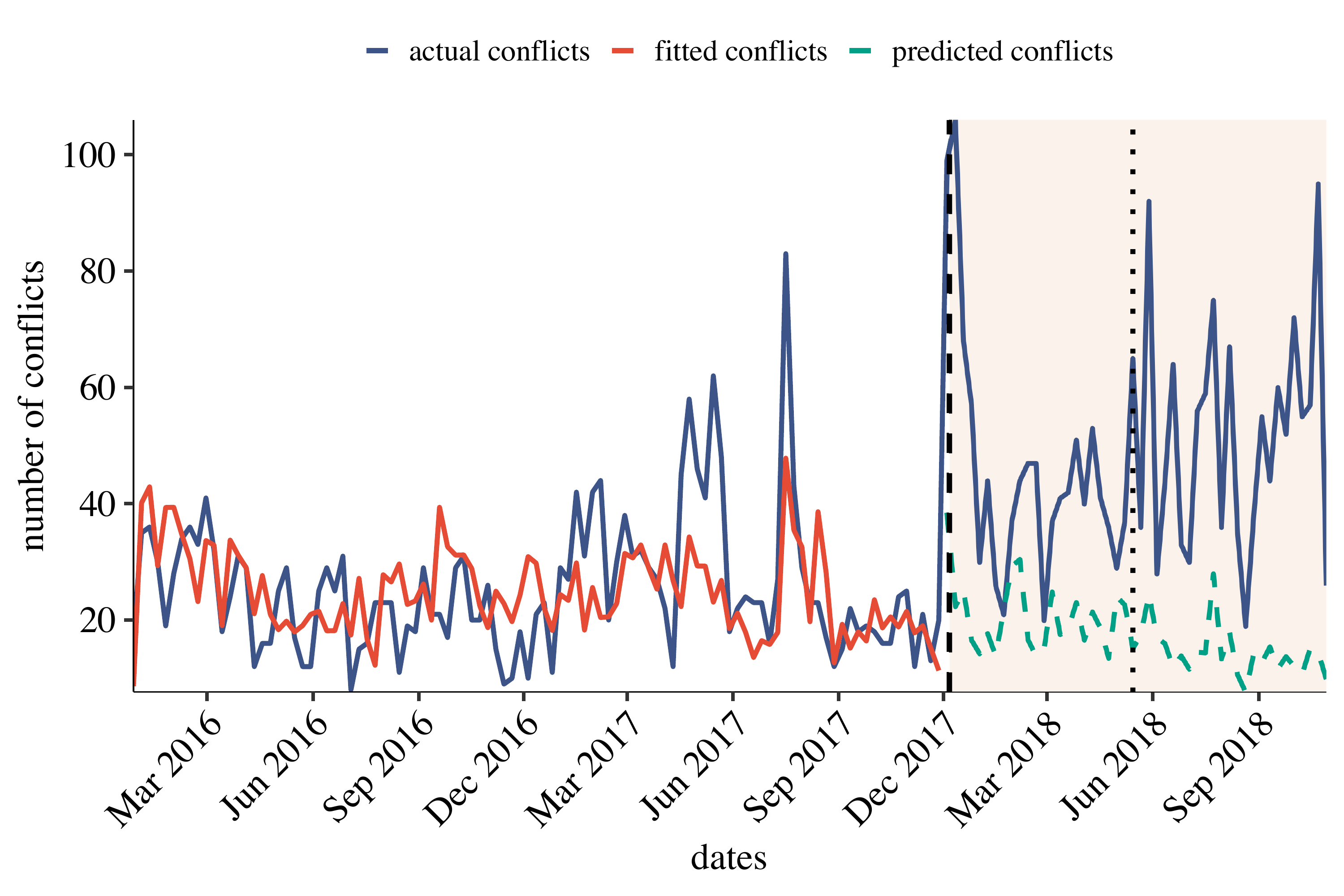}
		\end{adjustbox}	
		\caption{Estimation by elastic net}
		\label{Fig:counterfactual_elastic}
	\end{subfigure}\hfill
\begin{subfigure}{.5\textwidth}
		\begin{adjustbox}{max totalsize = {\textwidth}{0.9\textheight}, center}
			\includegraphics[width = \textwidth]{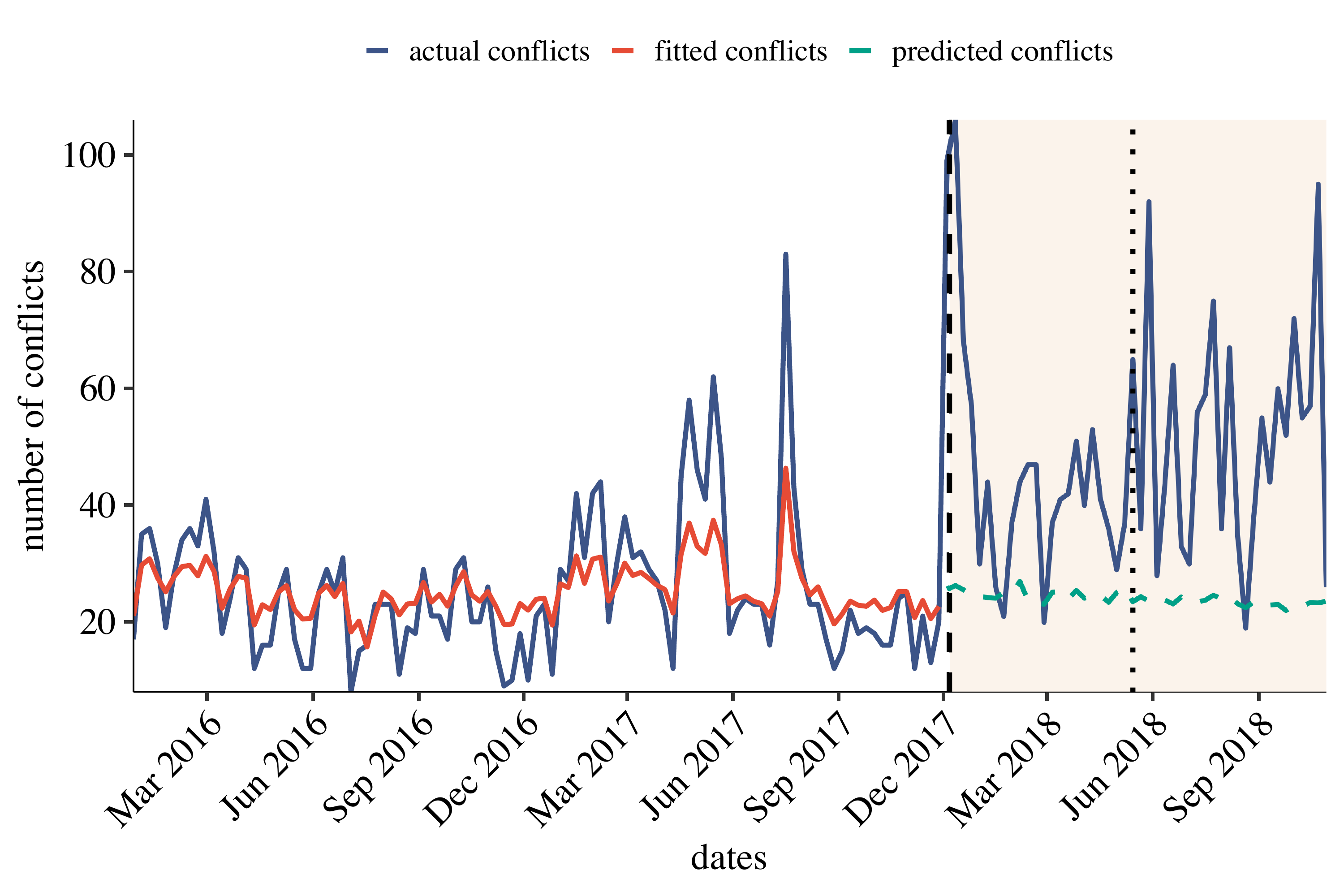}
		\end{adjustbox}
		\caption{Estimation by matrix completion}
		\label{Fig:counterfactual_matrix}
	\end{subfigure}
\caption{Comparison of the Four Methods based on the Observed and Estimated
Conflicts in Israel-Palestine}
\label{fig: comparing treated and control on actuals}
\vspace{0.166667in}
{\footnotesize
\textit{Notes:} Comparison of the four
methods showing the weekly number of conflicts in Israel and Palestine
(blue line) and its estimated counterpart in the pre-intervention
period (red line) and post-intervention period (green dashed line).
The vertical dashed and dotted lines represent the date when the move
of the US embassy was announced and the date of the actual move, respectively.
(a) shows the result of the tree-based controls, (b) for the synthetic
controls, (c) for the elastic net, and (d) for the matrix completion.\par}
\end{figure}

\vspace{0.166667in}
Figure \ref{fig: comparing treated and control on actuals} shows
the observed and estimated number of weekly conflicts in Israel-Palestine
for all four methods, and two features of the methods are noticeable.
First, the fit in the pre-intervention period gives an idea of the
ability to approximate the weekly level of conflicts in Israel-Palestine,
which is highly fluctuating. The synthetic control method, the elastic
net estimator, and the matrix completion method are comparable in
terms of pre-intervention fit, the matrix completion method being
marginally in the lead. The reason the elastic net estimator performs
slightly better compared to the synthetic control method is likely
because the elastic net is less restrictive when estimating weights.
None of the comparison methods, however, are able to approximate the
weekly level of conflicts in the pre-intervention period as well as
the tree-based control method.

\vspace{0.166667in}
Second, the variation in the estimated counterfactuals in the post-intervention
period hints at the degree of overfitting, particularly, if there is no or limited variation.
If a given method fits only to noise in the pre-treatment period, the post-treatment predictions will be roughly constant because the associated noise do not match the fitted noise.
Both the elastic net estimator and the tree-based synthetic control method appear to deliver
reasonable variation in the estimates. They are able to fit the shape
and pattern, but not the level of the observed conflicts. The ability
to fit shape but not level is exactly what leads us to estimate a significant
effect of the embassy move.

\vspace{0.166667in}
In contrast, the estimates by the synthetic
control method and the matrix completion method have little variation
and are closely centered around the average weekly number of conflicts
in the pre-intervention period. This is a sign of overfitting.\footnote{We thank Stefan Wager for pointing this out.} However, given the data available and in particular the number of control units,
this is not surprising. Recall that the matrix completion method combines
elements from vertical and horizontal regressions. For the horizontal
part, it tries to fit the post-intervention outcomes to the pre-intervention
outcomes using only 11 control countries. As the number of weeks is
much greater than the number of control countries, it is not surprising
that vertical regressions do perform better.
\begin{figure}[!t]
\begin{subfigure}{.5\textwidth}
		\begin{adjustbox}{max totalsize = {\textwidth}{0.9\textheight}, center}
			\includegraphics[width = \textwidth]{Figures/gap_tree}
		\end{adjustbox}
		\caption{Estimation by tree-based controls}
		\label{Fig:gap_tree}
	\end{subfigure}\hfill
\begin{subfigure}{.5\textwidth}
		\begin{adjustbox}{max totalsize = {\textwidth}{0.9\textheight}, center}
			\includegraphics[width = \textwidth]{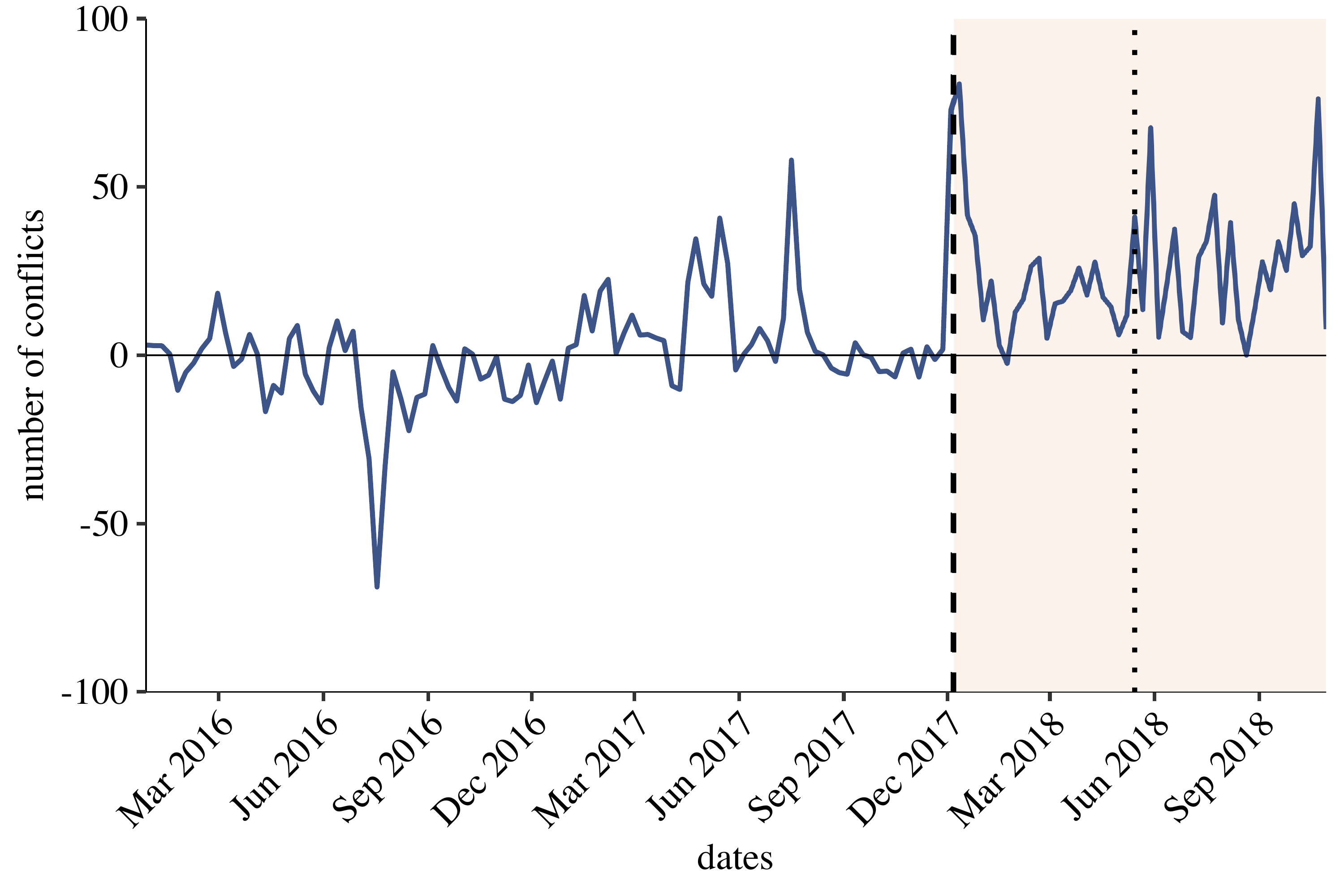}
		\end{adjustbox}	
		\caption{Estimation by synthetic controls}
	\label{Fig:gap_synthetic}
	\end{subfigure}
\begin{subfigure}{.5\textwidth}
		\begin{adjustbox}{max totalsize = {\textwidth}{0.9\textheight}, center}
			\includegraphics[width = \textwidth]{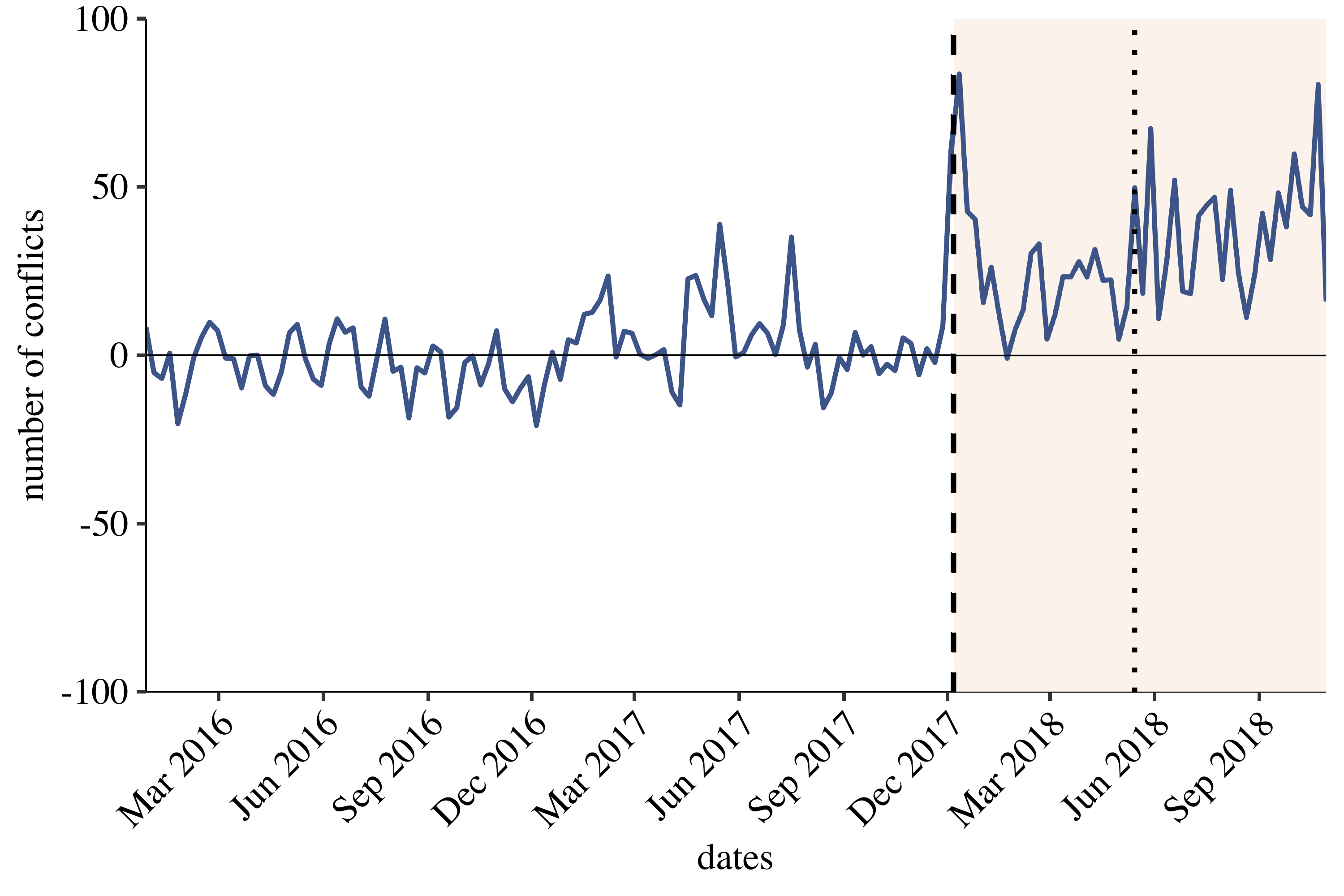}
		\end{adjustbox}	
		\caption{Estimation by elastic net}
		\label{Fig:gap_elastic}
	\end{subfigure}\hfill
\begin{subfigure}{.5\textwidth}
		\begin{adjustbox}{max totalsize = {\textwidth}{0.9\textheight}, center}
			\includegraphics[width = \textwidth]{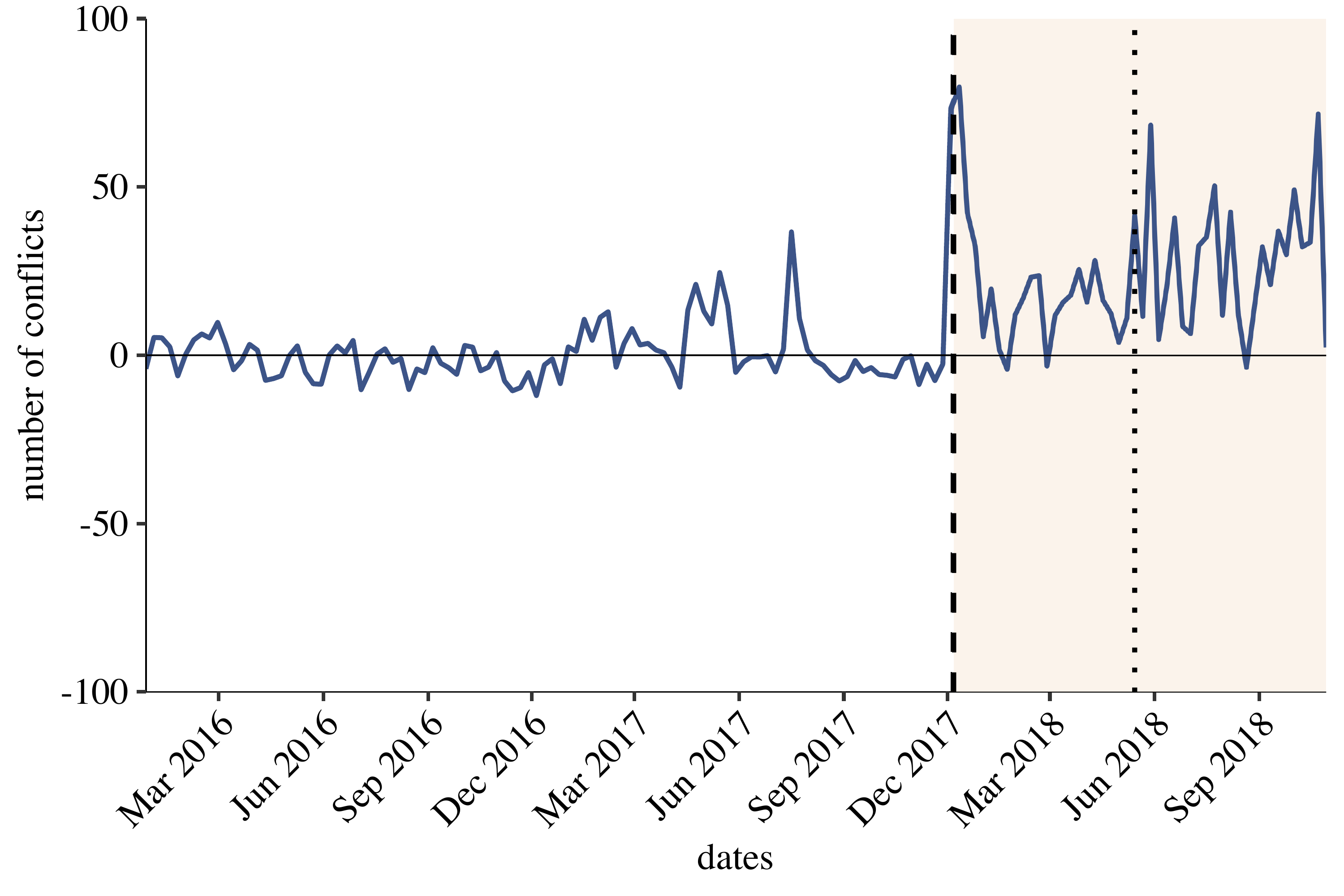}
		\end{adjustbox}
		\caption{Estimation by matrix completion}
		\label{Fig:gap_matrix}
	\end{subfigure}
\caption{Comparison of the Four Methods based on Discrepancies between the
Observed and Estimated Conflicts in Israel-Palestine.}
\label{fig: comparing treated and control on gaps}
\vspace{0.166667in}
{\footnotesize
\textit{Notes:} Comparison of the four
methods showing gaps between the observed and estimated weekly number
of conflicts in Israel and Palestine (blue line). The vertical dashed
and dotted lines represent the date when the move of the US embassy
was announced and the date of the actual move, respectively. (a) shows
the result of the tree-based controls, (b) for the synthetic controls,
(c) for the elastic net, and (d) for the matrix completion.\par}
\end{figure}

\vspace{0.166667in}
Figure \ref{fig: comparing treated and control on gaps} conveys the
same insights as Figure \ref{fig: comparing treated and control on actuals},
but instead of showing the observed and estimated number of weekly
conflicts separately, it displays the differences between the two.
Considering the differences instead of actuals provides an easier
approach to evaluating pre-intervention fit. Again, a good ability
to approximate the pre-intervention level of conflicts corresponds
to differences closely around zero. As apparent in Figure \ref{fig: comparing treated and control on gaps},
the tree-based synthetic control method delivers the best pre-intervention
fit, followed by the matrix completion method, the elastic net estimator,
and the synthetic control method. It is, however, impossible to assess
the overfitting indicated by little post-intervention variation from
Figure \ref{fig: comparing treated and control on gaps}.

\vspace{0.166667in}
From Figure \ref{fig: comparing treated and control on actuals}
and \ref{fig: comparing treated and control on gaps}, we have argued
that the tree-based synthetic control method performs at least as
well as state-of-the-art methods. Supporting this, Table \ref{tab: comparing methods}
provides the various measures that follow from the figures. In particular,
we compute the RMSPE and MAE in the pre-intervention period for all
the methods considered. Both measures capture the ability to approximate
the observed weekly level of conflicts in Israel-Palestine. The tree-based
synthetic control method outperforms all other methods on these metrics.
We also report the standard deviation of the estimated number of weekly
conflicts in the counterfactual Israel-Palestine absent of the embassy
move. The elastic net estimator is the only comparison method that
delivers higher variation than the tree-based synthetic control method.
The matrix completion method delivers almost no variation in the estimates.
\begin{table}[!t]
\caption{Summary of Performance Measures across Models Pre-treatment and Post-treatment}
\begin{adjustbox}{max width = \textwidth, center}
\begin{tabular}{lcccccc}
\toprule 
 &  & \multicolumn{2}{c}{Pre-intervention} &  & \multicolumn{2}{c}{Post-intervention}\tabularnewline
\cmidrule{3-4} \cmidrule{4-4} \cmidrule{6-7} \cmidrule{7-7} 
 &  & MAE & RMSPE &  & Std & Ave. gap\tabularnewline
\midrule
tree-based controls &  & 3.99 & 5.77 &  & 4.34 & 26.12\tabularnewline
synthetic controls &  & 9.62 & 14.73 &  & 2.86 & 25.14\tabularnewline
elastic net &  & 7.88 & 10.67 &  & 6.08 & 31.32\tabularnewline
matrix completion &  & 5.53 & 7.65 &  & 1.05 & 24.80\tabularnewline
\bottomrule
\end{tabular}
\end{adjustbox}\vspace{0.166667in}
\label{tab: comparing methods}
{\footnotesize 
\textit{Notes:} Summary of measures used
to assess the performance of the results obtained for Israel and Palestine.
Measures include mean absolute error and root mean squared prediction
error between the observed and estimated weekly number of conflicts
for both the pre- and post-intervention period. We also include the
estimated standard deviation of the estimates and the average gap
in the post-intervention period. We include the measures for the tree-based
control method and the comparison methods.\par}
\end{table}

\vspace{0.166667in}
Evaluating the degree of overfitting by computing standard errors
is clearly insufficient. One final approach to simultaneously assessing 
the ability of the methods to approximate the weekly number of conflicts
in Israel-Palestine and the degree of overfitting is to repeat the
analysis, but hold out a subsample of the pre-intervention period and
compute the RMSPE and MAE on this subsample. The hold-out sample serves
as a test sample, but in contrast to the post-intervention period,
we observe $Y_{t}^{0}$ as if the intervention has not yet occurred.
This allows us to evaluate the predictive ability. Specifically, we
hold out the last 10\% of the observations in the pre-intervention
period, resulting in an estimation sample and a validation sample.
Then, we re-run all methods on the estimation sample.

\vspace{0.166667in}
For the methods that require tuning of hyperparameters, namely the tree-based synthetic
control method, the elastic net estimator, and the matrix completion
method, we further split the estimation sample using an 80/20\% split
as in the original analysis. We use the 20\% to select the hyperparameters
rather than selecting hyperparameters on the full estimation sample.
For the synthetic control method, we use the whole estimation sample
to estimate the weights for each country as it does not require any
hyperparameters. Having estimated all parameters, we apply all the
methods to the validation sample for which we know the true outcome
and compute RMSPE and MAE.

\vspace{0.166667in}
Table \ref{tab: comparing methods validation} shows the results of
the hold-out sample approach. The elastic net estimator performs best
in terms of both metrics, followed by the tree-based synthetic control
method, the synthetic control method, and lastly the matrix completion
method. Our suspicion that the matrix completion method overfits as
seen in Figure \ref{fig: comparing treated and control on actuals}
appears to be confirmed. We emphasize that this is not an objection
to the method, but rather a result of the structure of the data, namely
$T\gg N$. The elastic net estimator performs very well on the validation
sample, and in fact better than evaluated on the entire pre-intervention
period.

\vspace{0.166667in}
Normally, we would take this as a sign of underfitting, but
as we run more than 20 different specifications of the elastic net
estimator in the pre-intervention period, it is more likely caused
by the validation sample being too small. The tree-based control performs
comparably in the validation sample as in using the entire pre-intervention
period, which indicates that neither overfitting nor underfitting
takes place. Being a nonparametric method, however, it requires more
data and, the fact that we only estimate the hyperparameters using
roughly 70\% of the pre-treatment data seems critical in this assessment
of the fit. Ideally, we would use a larger validation sample to compare
the methods on validation RSMPE and MAE.
\begin{table}[!t]
\caption{Summary of Performance Measures in Validation Sample}
\begin{adjustbox}{max width = \textwidth, center}
\begin{tabular}{lccc}
\toprule 
 &  & RMSPE & MAE\tabularnewline
\midrule
tree-based controls &  & 6.60 & 5.02\tabularnewline
synthetic controls &  & 7.72 & 6.96\tabularnewline
elastic net &  & 4.81 & 4.33\tabularnewline
matrix completion &  & 8.95 & 7.99\tabularnewline
\bottomrule
\end{tabular}
\end{adjustbox}\vspace{0.166667in}
\label{tab: comparing methods validation}
{\footnotesize 
\textit{Notes:} Summary of measures used
to assess the performance of the results obtained for Israel and Palestine.
Measures include mean absolute error and root mean squared prediction
error between the observed and estimated weekly number of conflicts
on a validation sample from the pre-intervention period. We include
the measures for the tree-based control method and the comparison
methods.\par}
\end{table} 

\section{Conclusion}\label{sec:Conclusion_tbsc}
The synthetic control method is an effective method in comparative
case studies in which relatively more time periods than potential
control units are available. The main advantage is the data-driven
approach to control unit selection. Since the estimation of the synthetic
controls is performed to maximize the pre-treatment fit to the treated
unit, however, the fit may not carry over into the post-treatment
period. One can argue that synthetic controls are not designed to
balance bias for variance, which may lead to overfitting to the pre-treatment
period despite the importance of high predictive performance in the
post-treatment period.

\vspace{0.166667in}
The elastic net estimator is an extension that
regularizes the weights on the control units to improve the post-treatment
fit. Both methods, however, impose a linear model that may not be
guided theoretically. In addition, if interactions and higher-order
terms of the control units are important to approximate the treated
unit but difficult to anticipate, the estimators may suffer from bias.
We recast the problem of estimating a counterfactual state as a prediction
problem. Specifically, we provide a data-driven method that balances
bias and variance to achieve post-treatment accuracy and is able to
capture nonlinearities without the need for a researcher specifying
them.

\vspace{0.166667in}
Our method can be applied in domains without theoretical guidelines
and is also able to recover linear models. We achieve predictive accuracy
because we replace the linear component of the synthetic controls
with a powerful model inspired by machine learning, namely the random
forests model. The ability to capture nonlinearities in a data-driven
way is a special feature of this model. This makes the tree-based
synthetic control method powerful, yet simple. We provide that the random forests regression model is asymptotically unbiased as well as consistent, which we use to establish consistency of the tree-based synthetic control method.

\vspace{0.166667in}
To demonstrate the applicability of the tree-based synthetic control
method, we evaluate the relocation of the US embassy from Tel Aviv to Jerusalem.
Specifically, we estimate the weekly number of conflicts in Israel
and Palestine in the counterfactual state of the world absent of the
embassy move. The estimates cover the period from the announcement
of the move on December 6, 2017, until November 3, 2018. Comparing
the estimates to the observed numbers, we find that the average number
of weekly conflicts in Israel and Palestine has increased by more
than 26 incidents since the move was announced. By placebo tests,
we show that the estimated effect of the embassy relocation is very unlikely
to be replicated if one were to arbitrarily relabel the treated unit
in the data given that the pre-treatment fit is reasonable.
To formally justify our results, we apply exact and robust conformal inference
tests and find statistical significance at the 1\% level.

\vspace{0.166667in}
We further compare the tree-based controls to state-of-the-art methods and conclude
that our method is data-driven and needs no linearity assumptions,
while it is not dominated even by the best of the comparison methods.
All comparison methods agree on the magnitude of the effect. 

\clearpage
\phantomsection
\addcontentsline{toc}{section}{References}
\bibliographystyle{ecta}

\begin{thebibliography}{52}
\newcommand{\enquote}[1]{``#1''}
\expandafter\ifx\csname natexlab\endcsname\relax\def\natexlab#1{#1}\fi

\bibitem[\protect\citeauthoryear{Abadie}{Abadie}{2019}]{abadie2019using}
\textsc{Abadie, A.} (2019): \enquote{Using synthetic controls: feasibility,
  data requirements, and methodological aspects,} \emph{Journal of Economic
  Literature (Forthcoming)}.

\bibitem[\protect\citeauthoryear{Abadie and Cattaneo}{Abadie and
  Cattaneo}{2018}]{Abadie2018}
\textsc{Abadie, A. and M.~D. Cattaneo} (2018): \enquote{Econometric methods for
  program evaluation,} \emph{Annual Review of Economics}, 10, 465--503.

\bibitem[\protect\citeauthoryear{Abadie, Diamond, and Hainmueller}{Abadie
  et~al.}{2010}]{abadie2010synthetic}
\textsc{Abadie, A., A.~Diamond, and J.~Hainmueller} (2010): \enquote{{Synthetic
  control methods for comparative case studies: estimating the effect of
  Californias tobacco control program},} \emph{Journal of the American
  Statistical Association}, 105, 493--505.

\bibitem[\protect\citeauthoryear{Abadie and Gardeazabal}{Abadie and
  Gardeazabal}{2003}]{Abadie2003}
\textsc{Abadie, A. and J.~Gardeazabal} (2003): \enquote{{The economic costs of
  conflict: a case study of the Basque country},} \emph{American Economic
  Review}, 93, 113--132.

\bibitem[\protect\citeauthoryear{Arnon and Weinblatt}{Arnon and
  Weinblatt}{2001}]{Arnon2001}
\textsc{Arnon, A. and J.~Weinblatt} (2001): \enquote{{Sovereignty and economic
  development: the case of Israel and Palestine},} \emph{The Economic Journal},
  111, 291--308.

\bibitem[\protect\citeauthoryear{Athey, Bayati, Doudchenko, Imbens, and
  Khosravi}{Athey et~al.}{2020}]{Athey2018matrix}
\textsc{Athey, S., M.~Bayati, N.~Doudchenko, G.~Imbens, and K.~Khosravi}
  (2020): \enquote{Matrix completion methods for causal panel data models,}
  \emph{arXiv Working Paper}, arXiv:1710.10251v3.

\bibitem[\protect\citeauthoryear{Athey and Imbens}{Athey and
  Imbens}{2016}]{Athey2016}
\textsc{Athey, S. and G.~Imbens} (2016): \enquote{Recursive partitioning for
  heterogeneous causal effects,} \emph{Proceedings of the National Academy of
  Sciences}, 113, 7353--7360.

\bibitem[\protect\citeauthoryear{Athey, Tibshirani, and Wager}{Athey
  et~al.}{2019}]{Athey2019}
\textsc{Athey, S., J.~Tibshirani, and S.~Wager} (2019): \enquote{Generalized
  random forests,} \emph{The Annals of Statistics}, 47, 1148--1178.

\bibitem[\protect\citeauthoryear{Bertrand, Duflo, and Mullainathan}{Bertrand
  et~al.}{2004}]{Bertrand2004}
\textsc{Bertrand, M., E.~Duflo, and S.~Mullainathan} (2004): \enquote{How much
  should we trust differences-in-differences estimates?} \emph{The Quarterly
  Journal of Economics}, 119, 249--275.

\bibitem[\protect\citeauthoryear{Breiman}{Breiman}{2001}]{Breiman2001}
\textsc{Breiman, L.} (2001): \enquote{Random forests,} \emph{Machine Learning},
  45, 5--32.

\bibitem[\protect\citeauthoryear{Buonomo}{Buonomo}{2018}]{Buonomo2018}
\textsc{Buonomo, T.} (2018): \enquote{{Iran's supreme leader: an analysis of
  his hostility toward the US and Israel},} \emph{Middle East Policy}, 25,
  33--45.

\bibitem[\protect\citeauthoryear{Card}{Card}{1990}]{Card1990}
\textsc{Card, D.} (1990): \enquote{{The impact of the mariel boatlift on the
  Miami labor market},} \emph{Industrial and Labor Relations Review}, 43,
  245--257.

\bibitem[\protect\citeauthoryear{Cavallo, Galiani, Noy, and Pantano}{Cavallo
  et~al.}{2013}]{Cavallo2013}
\textsc{Cavallo, E., S.~Galiani, I.~Noy, and J.~Pantano} (2013):
  \enquote{Catastrophic natural disasters and economic growth,} \emph{The
  Review of Economics and Statistics}, 95, 1549--1561.

\bibitem[\protect\citeauthoryear{Chernozhukov, Chetverikov, Demirer, Duflo,
  Hansen, and Newey}{Chernozhukov
  et~al.}{2017{\natexlab{a}}}]{Chernozhukov2017a}
\textsc{Chernozhukov, V., D.~Chetverikov, M.~Demirer, E.~Duflo, C.~Hansen, and
  W.~Newey} (2017{\natexlab{a}}): \enquote{Double/debiased/neyman machine
  learning of treatment effects,} \emph{American Economic Review}, 107,
  261--265.

\bibitem[\protect\citeauthoryear{Chernozhukov, Chetverikov, Demirer, Duflo,
  Hansen, Newey, and Robins}{Chernozhukov et~al.}{2018}]{Chernozhukov2018}
\textsc{Chernozhukov, V., D.~Chetverikov, M.~Demirer, E.~Duflo, C.~Hansen,
  W.~Newey, and J.~Robins} (2018): \enquote{Double/debiased machine learning
  for treatment and structural parameters,} \emph{The Econometrics Journal},
  21, C1--C68.

\bibitem[\protect\citeauthoryear{Chernozhukov, Demirer, Duflo, and
  Fernandez-Val}{Chernozhukov et~al.}{2019}]{Chernozhukov2018a}
\textsc{Chernozhukov, V., M.~Demirer, E.~Duflo, and I.~Fernandez-Val} (2019):
  \enquote{Generic machine learning inference on heterogenous treatment effects
  in randomized experiments,} \emph{arXiv Working Paper}, arXiv:1712.04802v4.

\bibitem[\protect\citeauthoryear{Chernozhukov, Wuthrich, and Zhu}{Chernozhukov
  et~al.}{2017{\natexlab{b}}}]{Chernozhukov2017b}
\textsc{Chernozhukov, V., K.~Wuthrich, and Y.~Zhu} (2017{\natexlab{b}}):
  \enquote{An exact and robust conformal inference method for counterfactual
  and synthetic controls,} \emph{arXiv Working Paper}, arXiv:1712.09089v7.

\bibitem[\protect\citeauthoryear{{Chernozhukov}, {Wuthrich}, and
  {Zhu}}{{Chernozhukov} et~al.}{2017}]{Chernozhukov2018ate}
\textsc{{Chernozhukov}, V., K.~{Wuthrich}, and Y.~{Zhu}} (2017):
  \enquote{Practical and robust $t$-test based inference for synthetic control
  and related methods,} \emph{arXiv Working Paper}, arXiv:1812.10820v4.

\bibitem[\protect\citeauthoryear{Davis and Nielsen}{Davis and
  Nielsen}{2020}]{davis2020rf}
\textsc{Davis, R.~A. and M.~S. Nielsen} (2020): \enquote{Modeling of time
  series using random forests: {T}heoretical developments,} \emph{Electron. J.
  Stat.}, 14, 3644--3671.

\bibitem[\protect\citeauthoryear{De'ath}{De'ath}{2002}]{Glenn2002}
\textsc{De'ath, G.} (2002): \enquote{Multivariate regression trees: a new
  technique for modeling species-environment relationships,} \emph{Ecology},
  83, 1105--1117.

\bibitem[\protect\citeauthoryear{Deaton and Cartwright}{Deaton and
  Cartwright}{2018}]{Deaton2018}
\textsc{Deaton, A. and N.~Cartwright} (2018): \enquote{Understanding and
  misunderstanding randomized controlled trials,} \emph{Social Science {\&}
  Medicine}, 210, 2--21.

\bibitem[\protect\citeauthoryear{Doudchenko and Imbens}{Doudchenko and
  Imbens}{2017}]{Doudchenko2016}
\textsc{Doudchenko, N. and G.~W. Imbens} (2017): \enquote{Balancing,
  regression, difference-in-differences and synthetic control methods: a
  synthesis,} \emph{arXiv Working Paper}, arXiv:1610.07748v2.

\bibitem[\protect\citeauthoryear{Eriksson}{Eriksson}{2018}]{Eriksson2018}
\textsc{Eriksson, J.} (2018): \enquote{{Master of none: Trump, jerusalem and
  the prospects of israeli-palestinian peace},} \emph{Middle East Policy}, 25,
  51--63.

\bibitem[\protect\citeauthoryear{Franco, Malhotra, and Simonovits}{Franco
  et~al.}{2014}]{franco2014publication}
\textsc{Franco, A., N.~Malhotra, and G.~Simonovits} (2014):
  \enquote{Publication bias in the social sciences: unlocking the file drawer,}
  \emph{Science}, 345, 1502--1505.

\bibitem[\protect\citeauthoryear{Frisch and Sandler}{Frisch and
  Sandler}{2004}]{Frisch2004}
\textsc{Frisch, H. and S.~Sandler} (2004): \enquote{{Religion, state, and the
  international system in the Israeli--Palestinian conflict},}
  \emph{International Political Science Review}, 25, 77--96.

\bibitem[\protect\citeauthoryear{Gu, Kelly, and Xiu}{Gu et~al.}{2020}]{Gu2020}
\textsc{Gu, S., B.~Kelly, and D.~Xiu} (2020): \enquote{Empirical asset pricing
  via machine learning,} \emph{The Review of Financial Studies}, 33,
  2223--2273.

\bibitem[\protect\citeauthoryear{Guha and Ng}{Guha and Ng}{2019}]{Ng2019}
\textsc{Guha, R. and S.~Ng} (2019): \emph{{A Machine Learning Analysis of
  Seasonal and Cyclical Sales in Weekly Scanner Data}}, University of Chicago
  Press.

\bibitem[\protect\citeauthoryear{Hahn and Shi}{Hahn and Shi}{2017}]{Hahn2017}
\textsc{Hahn, J. and R.~Shi} (2017): \enquote{Synthetic control and inference,}
  \emph{Econometrics}, 5, 52.

\bibitem[\protect\citeauthoryear{Hainmueller}{Hainmueller}{2012}]{Hainmueller2012}
\textsc{Hainmueller, J.} (2012): \enquote{Entropy balancing for causal effects:
  a multivariate reweighting method to produce balanced samples in
  observational studies,} \emph{Political Analysis}, 20, 25--46.

\bibitem[\protect\citeauthoryear{Hartford, Lewis, Leyton-Brown, and
  Taddy}{Hartford et~al.}{2017}]{Hartfort2017}
\textsc{Hartford, J., G.~Lewis, K.~Leyton-Brown, and M.~Taddy} (2017):
  \enquote{Deep {IV}: a flexible approach for counterfactual prediction,} in
  \emph{Proceedings of the 34th International Conference on Machine Learning},
  ed. by D.~Precup and Y.~W. Teh, International Convention Centre, Sydney,
  Australia: PMLR, vol.~70 of \emph{Proceedings of Machine Learning Research},
  1414--1423.

\bibitem[\protect\citeauthoryear{Howard and Bowles}{Howard and
  Bowles}{2012}]{howard2012two}
\textsc{Howard, J. and M.~Bowles} (2012): \enquote{The two most important
  algorithms in predictive modeling today,} in \emph{Strata Conference
  presentation, February}, vol.~28.

\bibitem[\protect\citeauthoryear{Imbens and Wooldridge}{Imbens and
  Wooldridge}{2009}]{Imbens2009}
\textsc{Imbens, G.~W. and J.~M. Wooldridge} (2009): \enquote{Recent
  developments in the econometrics of program evaluation,} \emph{Journal of
  Economic Literature}, 47, 5--86.

\bibitem[\protect\citeauthoryear{Kleinberg, Ludwig, Mullainathan, and
  Obermeyer}{Kleinberg et~al.}{2015}]{Kleinberg2015}
\textsc{Kleinberg, J., J.~Ludwig, S.~Mullainathan, and Z.~Obermeyer} (2015):
  \enquote{Prediction policy problems,} \emph{American Economic Review}, 105,
  491--95.

\bibitem[\protect\citeauthoryear{Lundberg and Lee}{Lundberg and
  Lee}{2017}]{Lundberg2017}
\textsc{Lundberg, S.~M. and S.-I. Lee} (2017): \enquote{A unified approach to
  interpreting model predictions,} in \emph{Advances in Neural Information
  Processing Systems}, 4765--4774.

\bibitem[\protect\citeauthoryear{Medeiros, Vasconcelos, Veiga, and
  Zilberman}{Medeiros et~al.}{2019}]{Medeiros2019}
\textsc{Medeiros, M.~C., G.~F.~R. Vasconcelos, {\'{A}}.~Veiga, and
  E.~Zilberman} (2019): \enquote{Forecasting inflation in a data-rich
  environment: the benefits of machine learning methods,} \emph{Journal of
  Business {\&} Economic Statistics}, 1--22.

\bibitem[\protect\citeauthoryear{Meinshausen}{Meinshausen}{2006}]{Meinshausen2006}
\textsc{Meinshausen, N.} (2006): \enquote{Quantile regression forests,}
  \emph{Journal of Machine Learning Research}, 7, 983--999.

\bibitem[\protect\citeauthoryear{Merlev{\`e}de, Peligrad, Rio
  et~al.}{Merlev{\`e}de et~al.}{2009}]{merlevede2009bernstein}
\textsc{Merlev{\`e}de, F., M.~Peligrad, E.~Rio, et~al.} (2009):
  \enquote{Bernstein inequality and moderate deviations under strong mixing
  conditions,} in \emph{High dimensional probability V: the Luminy volume},
  Institute of Mathematical Statistics, 273--292.

\bibitem[\protect\citeauthoryear{Montgomery and Olivella}{Montgomery and
  Olivella}{2018}]{Montgomery2018}
\textsc{Montgomery, J.~M. and S.~Olivella} (2018): \enquote{Tree-based models
  for political science data,} \emph{American Journal of Political Science},
  62, 729--744.

\bibitem[\protect\citeauthoryear{Musil, Warner, Yobas, and Jones}{Musil
  et~al.}{2002}]{musil2002comparison}
\textsc{Musil, C.~M., C.~B. Warner, P.~K. Yobas, and S.~L. Jones} (2002):
  \enquote{A comparison of imputation techniques for handling missing data,}
  \emph{Western Journal of Nursing Research}, 24, 815--829.

\bibitem[\protect\citeauthoryear{Pierdzioch and Risse}{Pierdzioch and
  Risse}{2018}]{Pierdzioch2018}
\textsc{Pierdzioch, C. and M.~Risse} (2018): \enquote{Forecasting precious
  metal returns with multivariate random forests,} \emph{Empirical Economics},
  1--18.

\bibitem[\protect\citeauthoryear{Raleigh, Linke, Hegre, and Karlsen}{Raleigh
  et~al.}{2010}]{raleigh2010introducing}
\textsc{Raleigh, C., A.~Linke, H.~Hegre, and J.~Karlsen} (2010):
  \enquote{{Introducing ACLED: an armed conflict location and event dataset:
  special data feature},} \emph{Journal of Peace Research}, 47, 651--660.

\bibitem[\protect\citeauthoryear{Rio}{Rio}{1993}]{rio1993covariance}
\textsc{Rio, E.} (1993): \enquote{Covariance inequalities for strongly mixing
  processes,} in \emph{Annales de l'IHP Probabilit{\'e}s et statistiques},
  vol.~29, 587--597.

\bibitem[\protect\citeauthoryear{Robbins, Saunders, and Kilmer}{Robbins
  et~al.}{2017}]{Robbins2017}
\textsc{Robbins, M.~W., J.~Saunders, and B.~Kilmer} (2017): \enquote{A
  framework for synthetic control methods with high-dimensional, micro-level
  data: evaluating a neighborhood-specific crime intervention,} \emph{Journal
  of the American Statistical Association}, 112, 109--126.

\bibitem[\protect\citeauthoryear{Rosenbaum}{Rosenbaum}{2007}]{Rosenbaum2007}
\textsc{Rosenbaum, P.~R.} (2007): \enquote{Interference between units in
  randomized experiments,} \emph{Journal of the American Statistical
  Association}, 102, 191--200.

\bibitem[\protect\citeauthoryear{Segal and Xiao}{Segal and
  Xiao}{2011}]{Segal2011}
\textsc{Segal, M. and Y.~Xiao} (2011): \enquote{Multivariate random forests,}
  \emph{Wiley Interdisciplinary Reviews: Data Mining and Knowledge Discovery},
  1, 80--87.

\bibitem[\protect\citeauthoryear{Shao and Wang}{Shao and
  Wang}{2002}]{shao2002sample}
\textsc{Shao, J. and H.~Wang} (2002): \enquote{Sample correlation coefficients
  based on survey data under regression imputation,} \emph{Journal of the
  American Statistical Association}, 97, 544--552.

\bibitem[\protect\citeauthoryear{Strobl, Boulesteix, Kneib, Augustin, and
  Zeileis}{Strobl et~al.}{2008}]{Strobl2008}
\textsc{Strobl, C., A.-L. Boulesteix, T.~Kneib, T.~Augustin, and A.~Zeileis}
  (2008): \enquote{Conditional variable importance for random forests,}
  \emph{{BMC} Bioinformatics}, 9.

\bibitem[\protect\citeauthoryear{Tibshirani}{Tibshirani}{1996}]{Tibshirani1996}
\textsc{Tibshirani, R.} (1996): \enquote{Regression shrinkage and selection via
  the {L}asso,} \emph{Journal of the Royal Statistical Society. Series B
  (Statistical Methodology)}, 267--288.

\bibitem[\protect\citeauthoryear{Wager and Athey}{Wager and
  Athey}{2018}]{Wager2018}
\textsc{Wager, S. and S.~Athey} (2018): \enquote{Estimation and inference of
  heterogeneous treatment effects using random forests,} \emph{Journal of the
  American Statistical Association}, 113, 1228--1242.

\bibitem[\protect\citeauthoryear{Wager and Walther}{Wager and
  Walther}{2015}]{wager2015adaptive}
\textsc{Wager, S. and G.~Walther} (2015): \enquote{Adaptive concentration of
  regression trees, with application to random forests,} \emph{arXiv preprint
  arXiv:1503.06388}.

\bibitem[\protect\citeauthoryear{Yu}{Yu}{1994}]{yu1994rates}
\textsc{Yu, B.} (1994): \enquote{Rates of convergence for empirical processes
  of stationary mixing sequences,} \emph{The Annals of Probability}, 94--116.

\bibitem[\protect\citeauthoryear{Zou and Hastie}{Zou and
  Hastie}{2005}]{Zou2005}
\textsc{Zou, H. and T.~Hastie} (2005): \enquote{Regularization and variable
  selection via the {Elastic Net},} \emph{Journal of the Royal Statistical
  Society. Series B (Statistical Methodology)}, 67, 301--320.

\end{thebibliography}

\newpage

\newpage 
\appendix

\renewcommand{\thesection}{\Alph{section}}
\numberwithin{equation}{section}

\clearpage
\section{Common trends in the Middle East}\label{Sec:Appendix_tbsc}
\begin{figure}[H] \begin{subfigure}{.5\textwidth}
		\begin{adjustbox}{max totalsize = {\textwidth}{0.9\textheight}, center}
			\includegraphics[width = \textwidth]{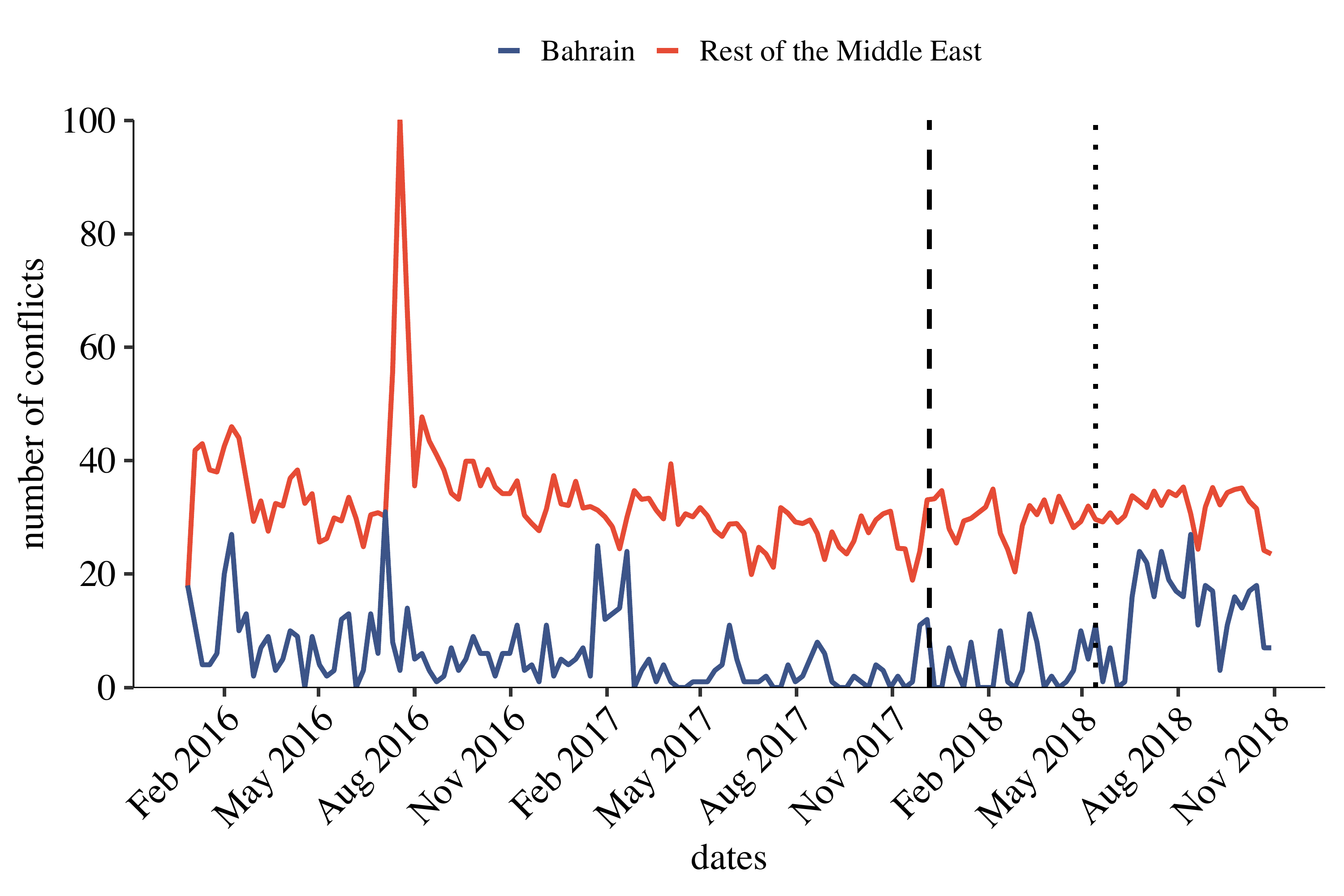}
		\end{adjustbox}
		\caption{Conflict trends in Bahrain}
		\label{Fig:descriptive_Bahrain}
	\end{subfigure}\hfill
\begin{subfigure}{.5\textwidth}
		\begin{adjustbox}{max totalsize = {\textwidth}{0.9\textheight}, center}
			\includegraphics[width = \textwidth]{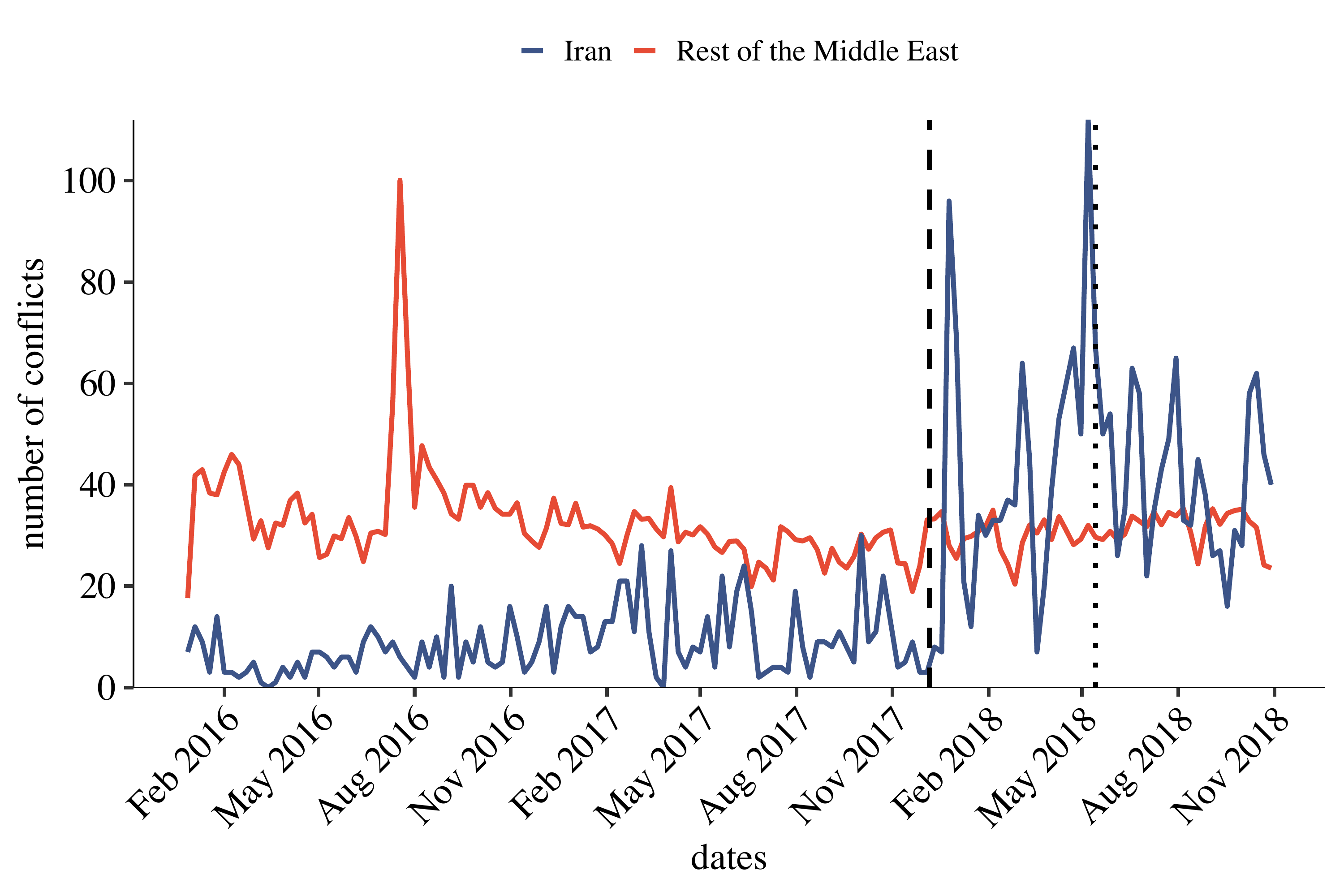}
		\end{adjustbox}	
		\caption{Conflict trends in Iran}
	\label{Fig:descriptive_Iran}
	\end{subfigure}
\begin{subfigure}{.5\textwidth}
		\begin{adjustbox}{max totalsize = {\textwidth}{0.9\textheight}, center}
			\includegraphics[width = \textwidth]{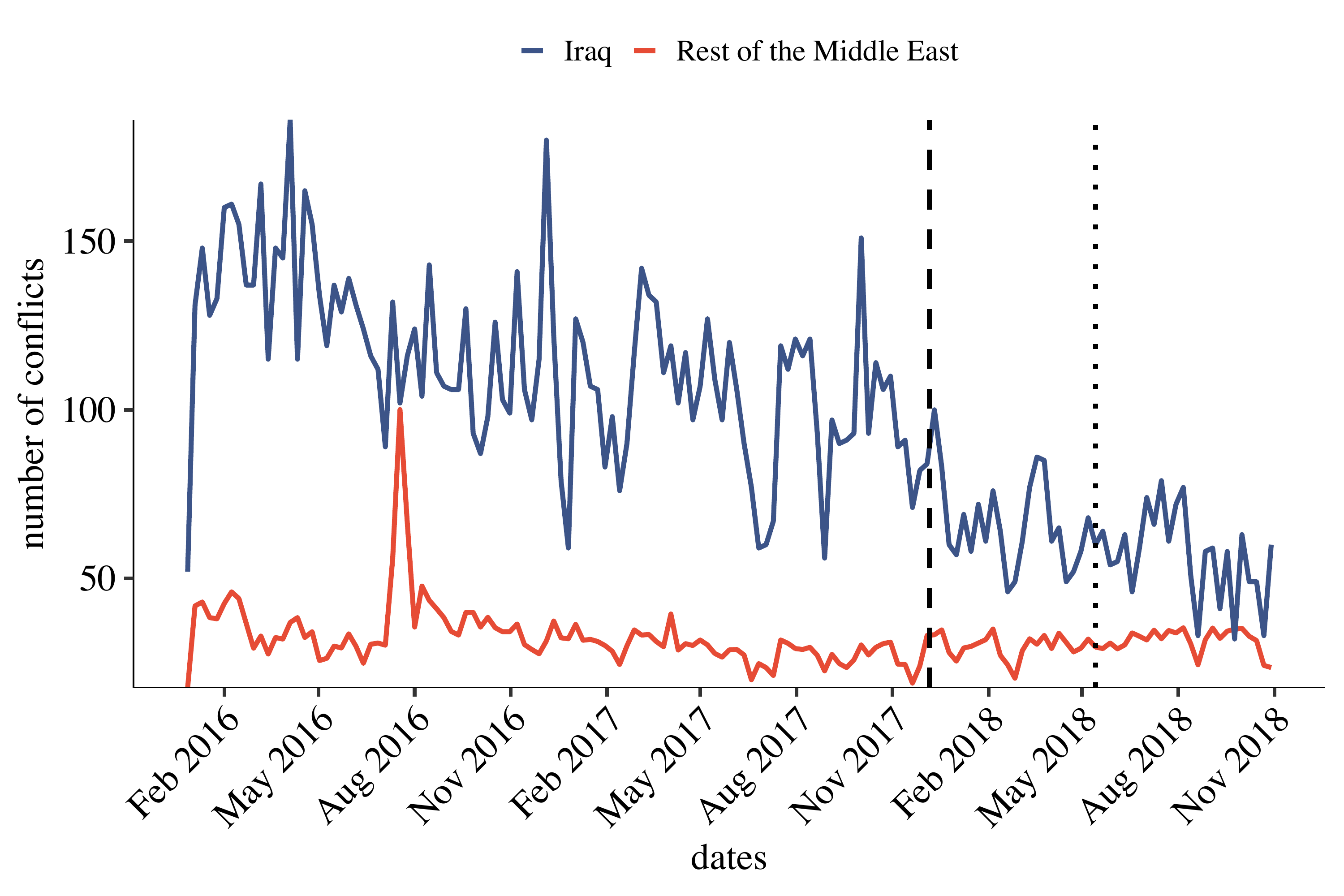}
		\end{adjustbox}	
		\caption{Conflict trends in Iraq}
		\label{Fig:descriptive_Iraq}
	\end{subfigure}\hfill
\begin{subfigure}{.5\textwidth}
		\begin{adjustbox}{max totalsize = {\textwidth}{0.9\textheight}, center}
			\includegraphics[width = \textwidth]{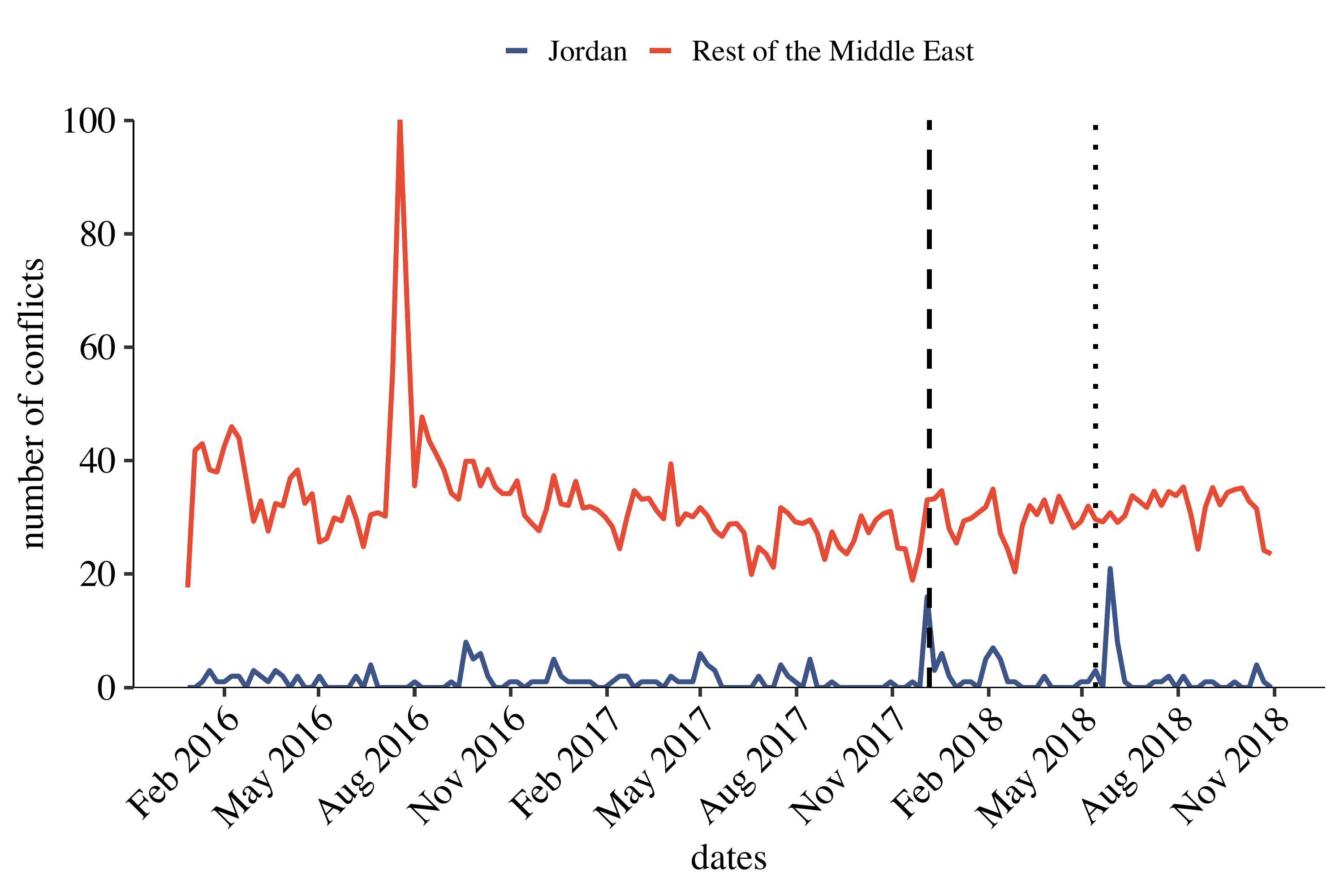}
		\end{adjustbox}
		\caption{Conflict trends in Jordan}
		\label{Fig:descriptive_Jordan}
	\end{subfigure}
\begin{subfigure}{.5\textwidth}
		\begin{adjustbox}{max totalsize = {\textwidth}{0.9\textheight}, center}
			\includegraphics[width = \textwidth]{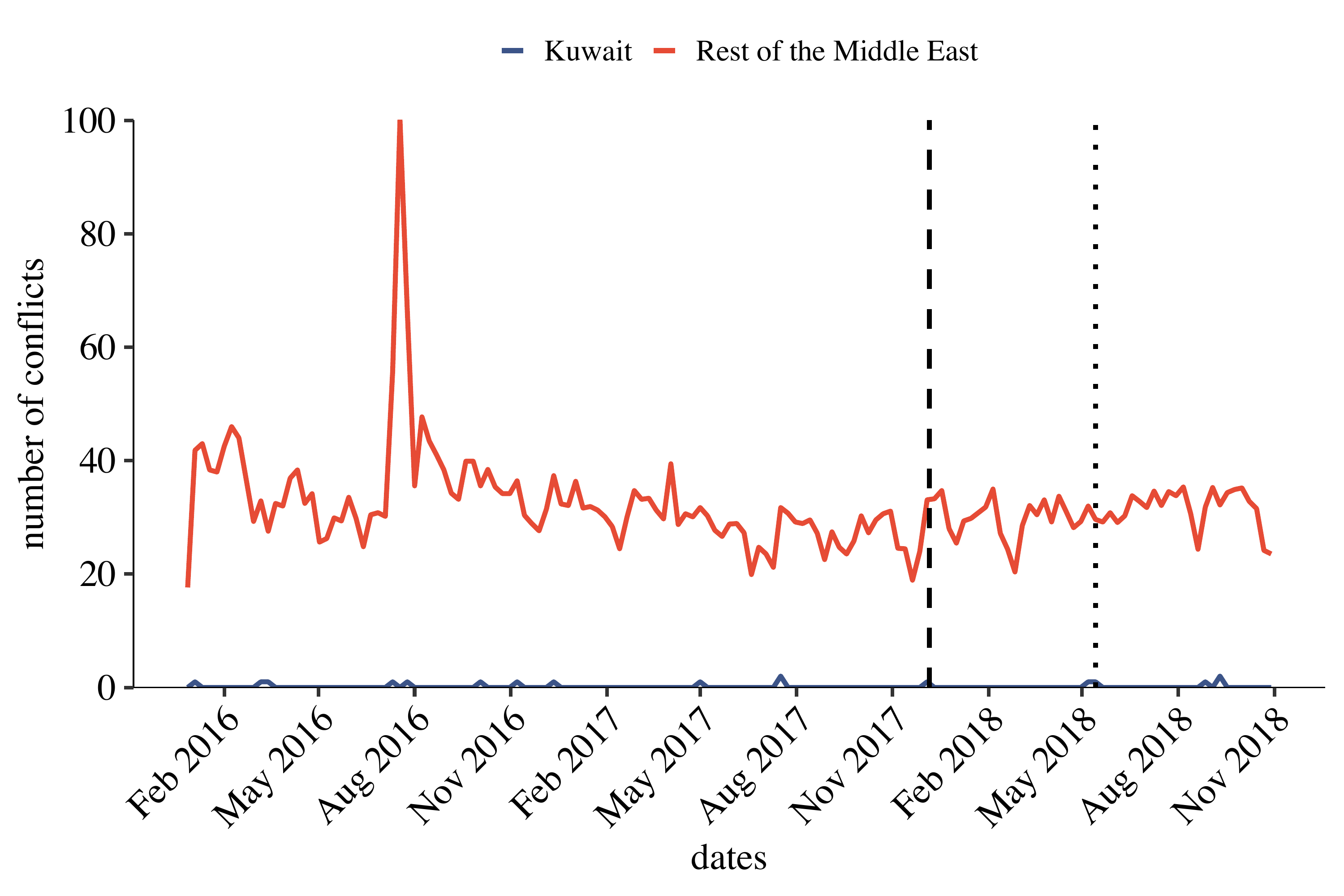}
		\end{adjustbox}	
		\caption{Conflict trends in Kuwait}
		\label{Fig:descriptive_Kuwait}
	\end{subfigure}\hfill
\begin{subfigure}{.5\textwidth}
		\begin{adjustbox}{max totalsize = {\textwidth}{0.9\textheight}, center}
			\includegraphics[width = \textwidth]{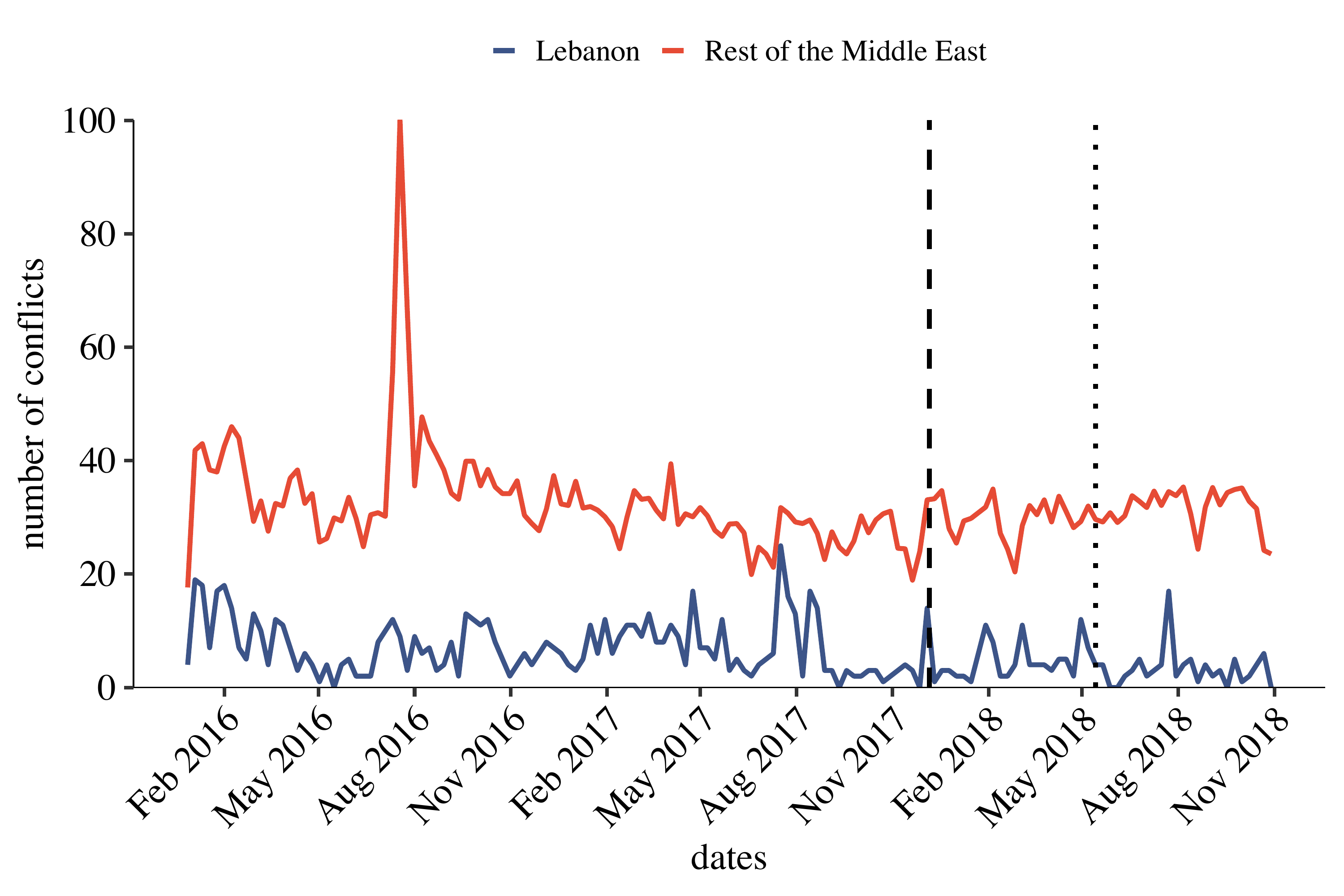}
		\end{adjustbox}
		\caption{Conflict trends in Lebanon}
		\label{Fig:descriptive_Lebanon}
	\end{subfigure}
\caption{Weekly Number of Conflicts in the Middle East (\RomNum{1})}
\label{fig: descriptive_conflicts_middle_east_1}
\vspace{0.166667in}
{\footnotesize
\textit{Notes:} Weekly number of conflicts
in each of the control countries in the Middle East together with
Iran (blue line) in addition to the average of the control countries
in the Middle East (red line). The vertical dashed and dotted lines
represent the date when the relocation of the US embassy was announced and
the date of the actual relocation, respectively.\par}
\end{figure}
\begin{figure}[H]
\begin{subfigure}{.5\textwidth}
		\begin{adjustbox}{max totalsize = {\textwidth}{0.9\textheight}, center}
			\includegraphics[width = \textwidth]{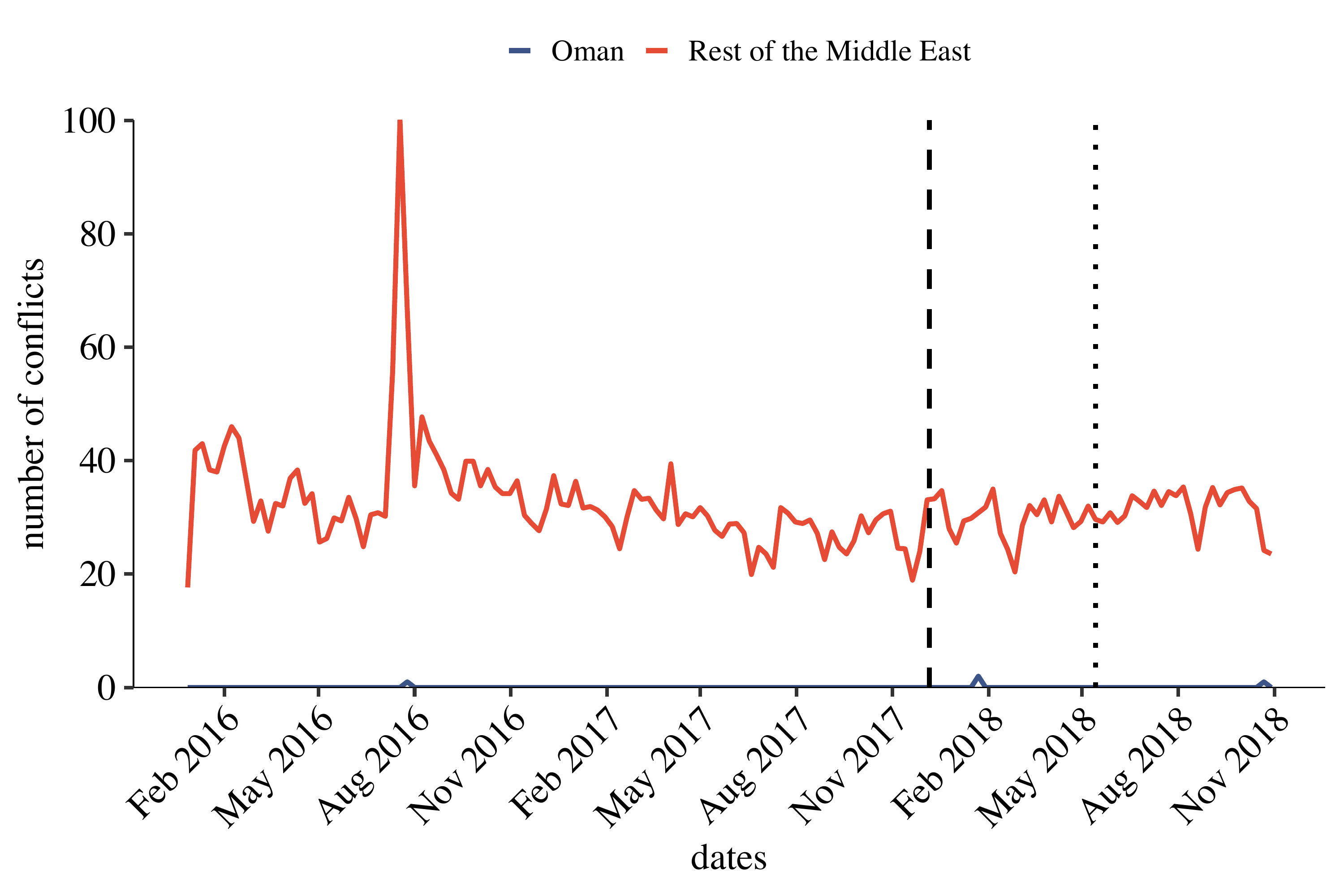}
		\end{adjustbox}
		\caption{Conflict trends in Oman}
		\label{Fig:descriptive_Oman}
	\end{subfigure}\hfill
\begin{subfigure}{.5\textwidth}
		\begin{adjustbox}{max totalsize = {\textwidth}{0.9\textheight}, center}
			\includegraphics[width = \textwidth]{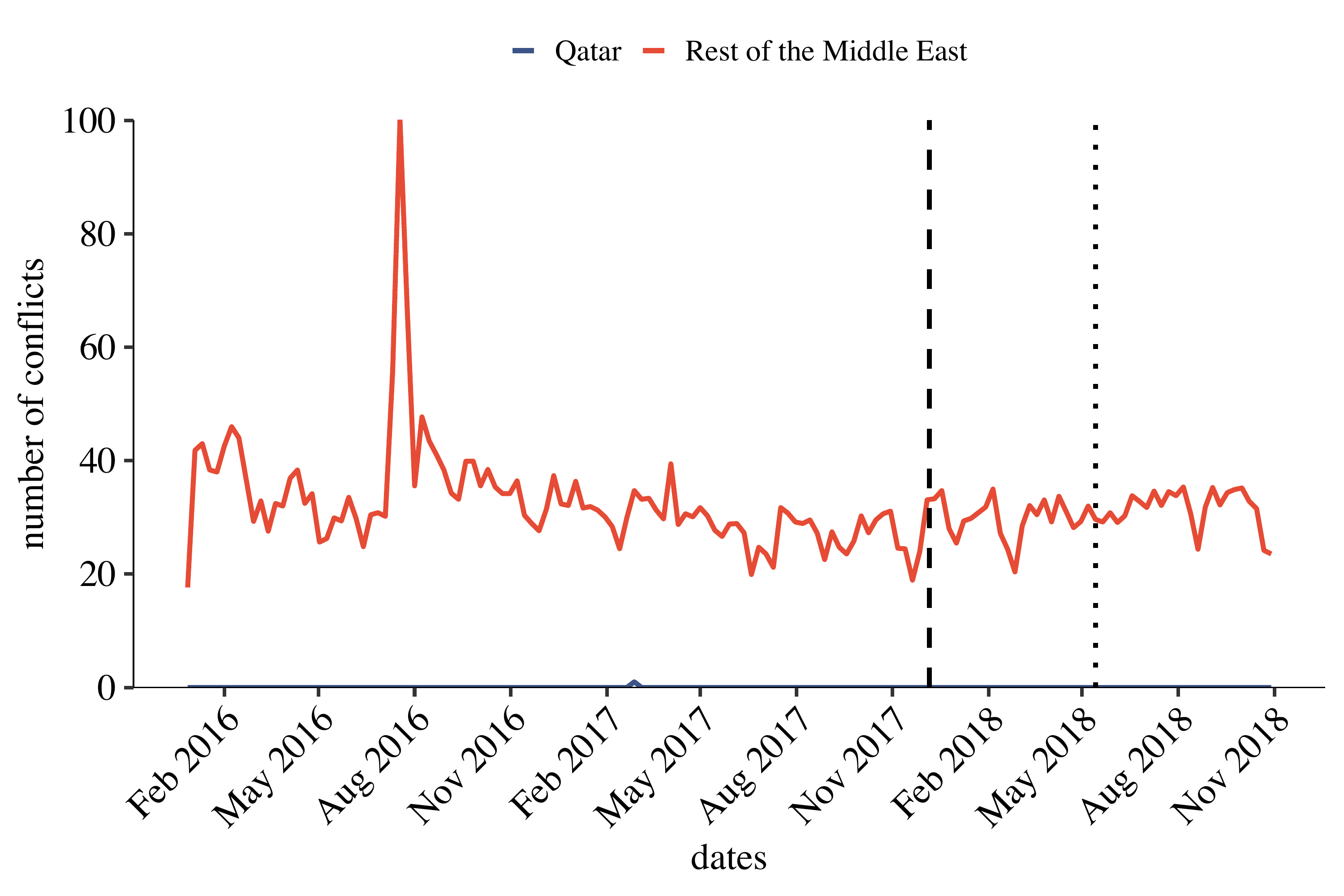}
		\end{adjustbox}	
		\caption{Conflict trends in Qatar}
	\label{Fig:descriptive_Qatar}
	\end{subfigure}
\begin{subfigure}{.5\textwidth}
		\begin{adjustbox}{max totalsize = {\textwidth}{0.9\textheight}, center}
			\includegraphics[width = \textwidth]{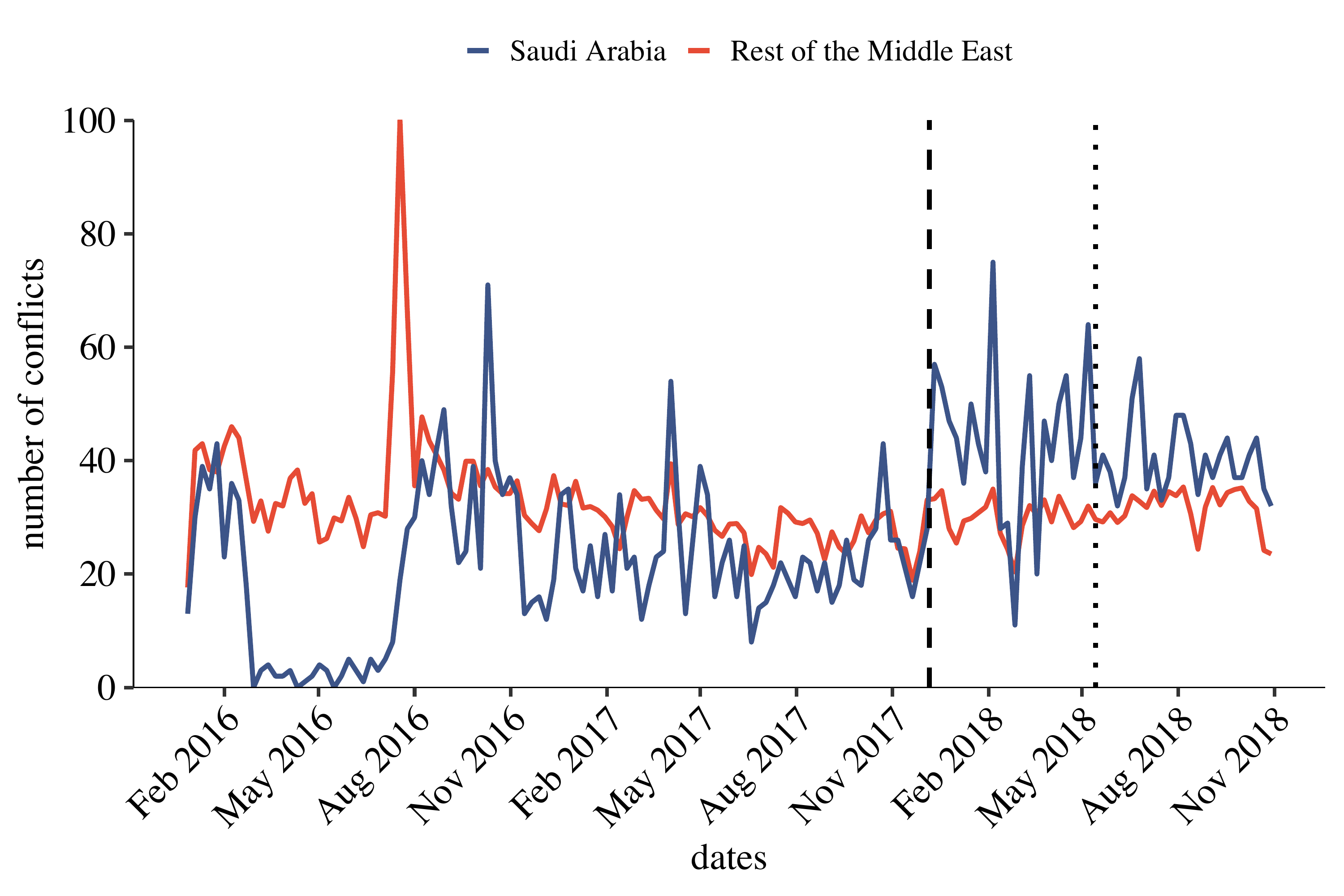}
		\end{adjustbox}	
		\caption{Conflict trends in Saudi Arabia}
		\label{Fig:descriptive_Saudi_Arabia}
	\end{subfigure}\hfill
\begin{subfigure}{.5\textwidth}
		\begin{adjustbox}{max totalsize = {\textwidth}{0.9\textheight}, center}
			\includegraphics[width = \textwidth]{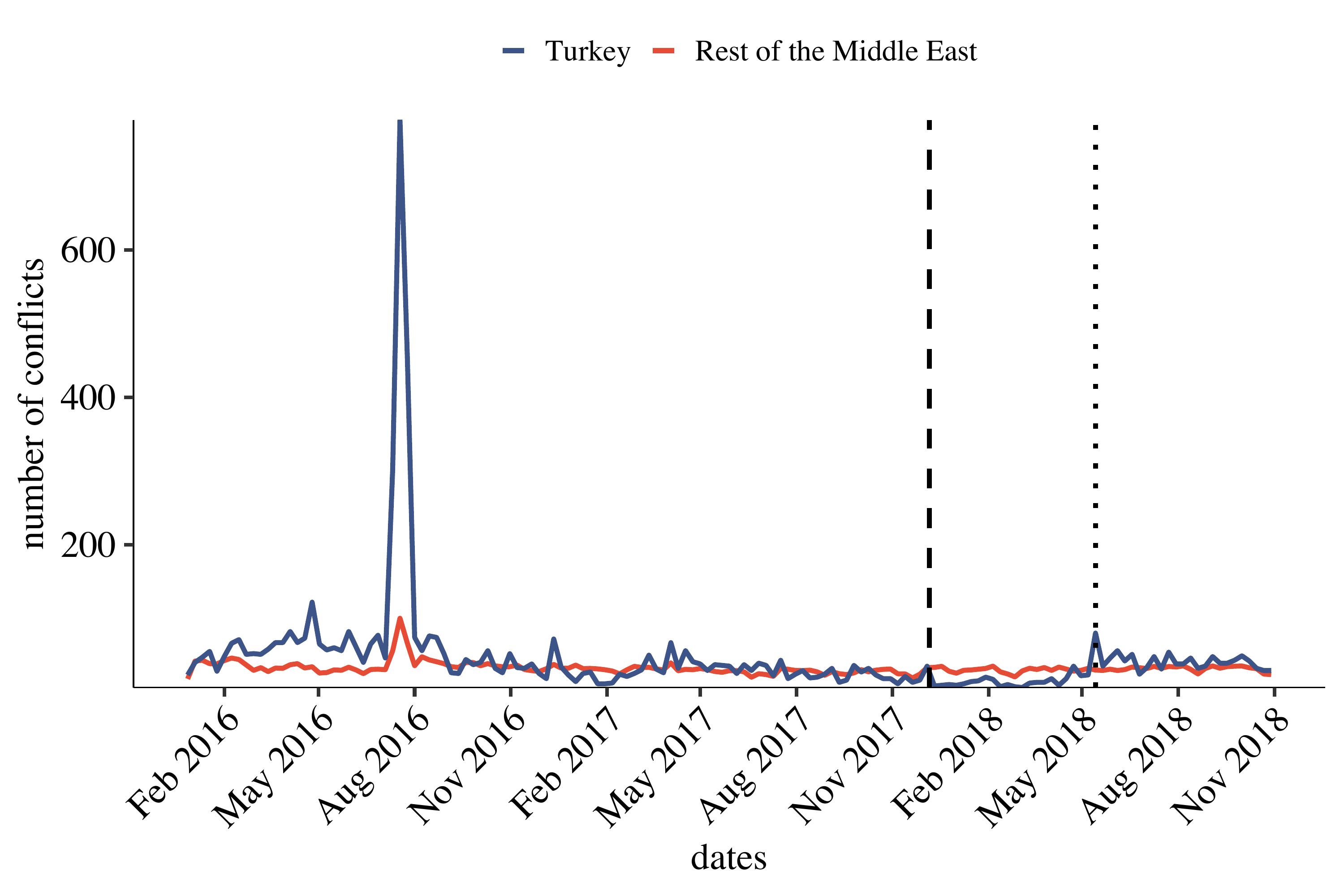}
		\end{adjustbox}
		\caption{Conflict trends in Turkey}
		\label{Fig:descriptive_Turkey}
	\end{subfigure}
\begin{subfigure}{.5\textwidth}
		\begin{adjustbox}{max totalsize = {\textwidth}{0.9\textheight}, center}
			\includegraphics[width = \textwidth]{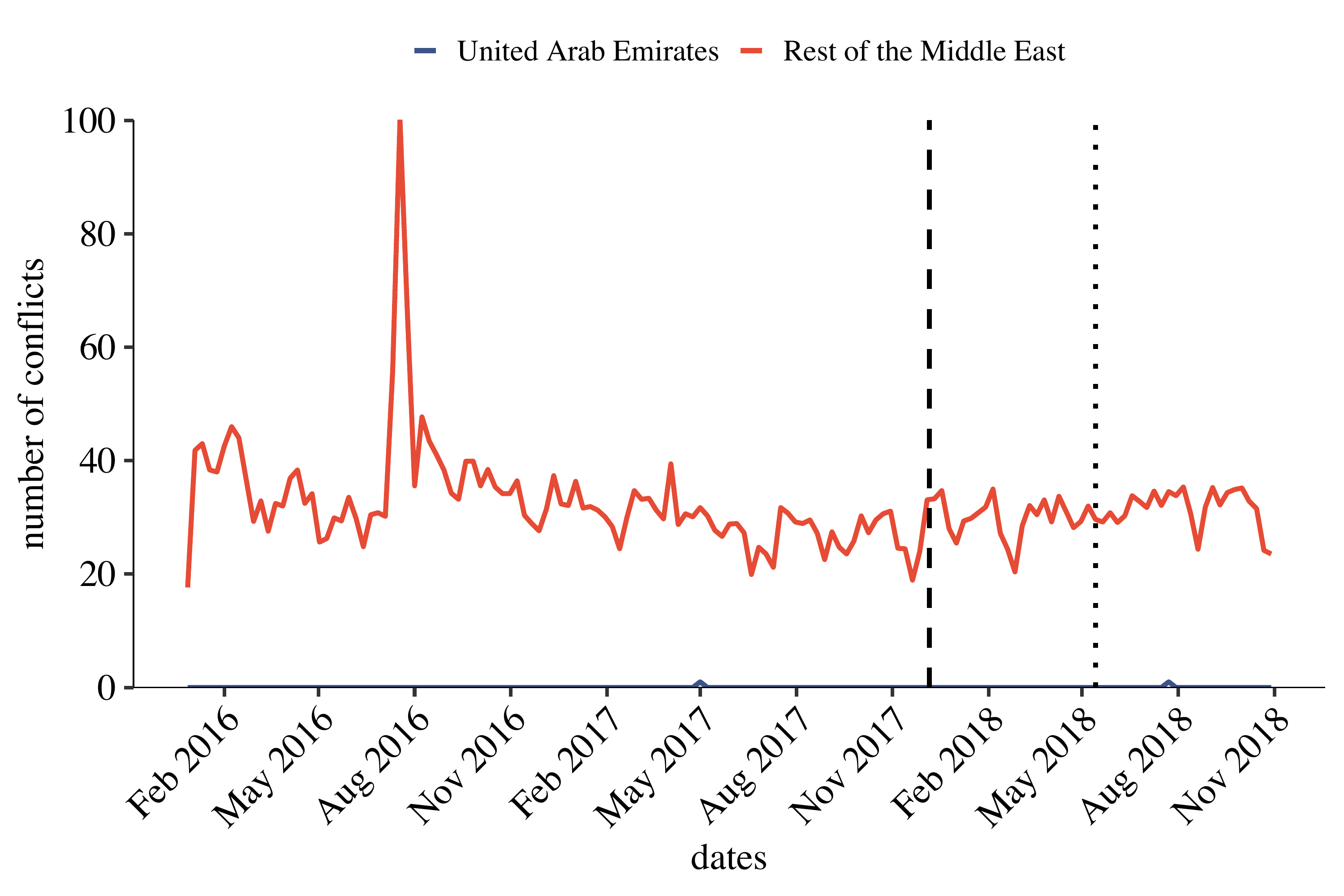}
		\end{adjustbox}	
		\caption{Conflict trends in United Arab Emirates}
		\label{Fig:descriptive_United_Arab_Emirates}
	\end{subfigure}\hfill
\begin{subfigure}{.5\textwidth}
		\begin{adjustbox}{max totalsize = {\textwidth}{0.9\textheight}, center}
			\includegraphics[width = \textwidth]{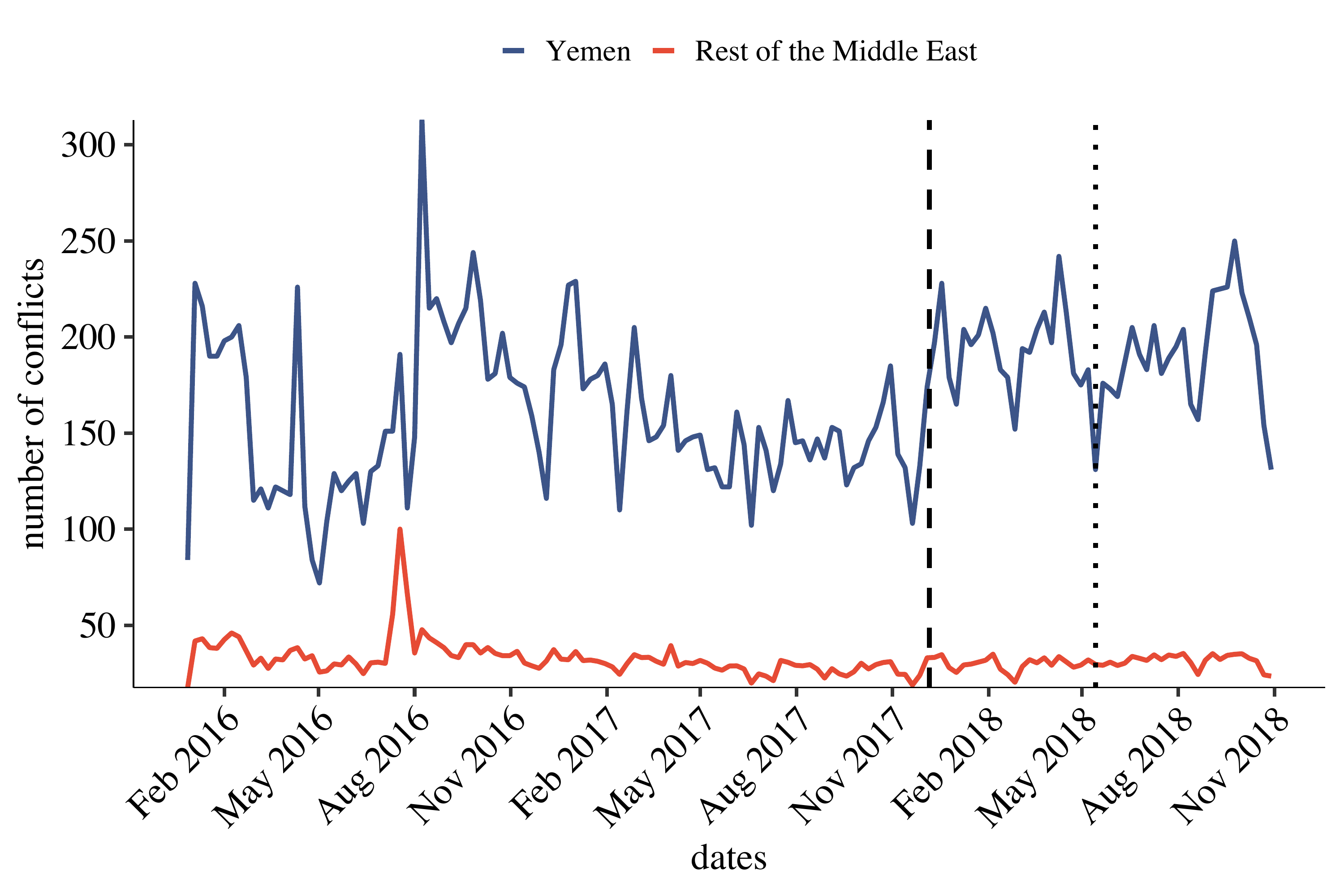}
		\end{adjustbox}
		\caption{Conflict trends in Yemen}
		\label{Fig:descriptive_Yemen}
	\end{subfigure}
\caption{Weekly Number of Conflicts in the Middle East (\RomNum{2})}
\label{fig: descriptive_conflicts_middle_east_2}
\vspace{0.166667in}
{\footnotesize
\textit{Notes:} Weekly number of conflicts
in each of the control countries in the Middle East together with
Iran (blue line) in addition to the average of the control countries
in the Middle East (red line). The vertical dashed and dotted lines
represent the date when the relocation of the US embassy was announced and
the date of the actual relocation, respectively.\par}
\end{figure} 
\clearpage
\section{Further considerations}\label{Sec:Appendix_practical}
\subsection{Implementation}
Choosing the best parametrization of the highly flexible tree-based
model is essential to avoid overfitting to the pre-intervention period.
To see this, imagine a single regression tree that is fully grown. Hence,
every leaf contains only one observation. Using this particular tree
in the pre-intervention period delivers a mean squared error of exactly
zero because it can fit every single observation perfectly, which
is not ideal. The same applies to random forests.

\vspace{0.166667in}
Therefore, we split the pre-intervention period further into an estimation sample and
a validation sample of relative sizes equal to 80\% and 20\%, respectively,
keeping the temporal ordering. We estimate the model on the estimation
sample and select the model complexity on the validation sample by
tuning hyperparameters. We tune the number of control units selected
for each tree, namely $m$, and set all other hyperparameters to their default.
Alternatively, out-of-bag predictions could be used to tune the hyperparameters as the out-of-bag error approximates well the generalization error \citep{Breiman2001}, leading to an efficient use of data. For dependent data, however, we choose the temporal sample split as it is more conservative.
Via this data splitting approach, we control
the bias and variance of the model. Similar ideas of sample splitting
have been suggested by \citet{Chernozhukov2018} and \citet{Chernozhukov2018a}.
We note, however, that we obtain essentially identical results using
default settings which is $m=\sqrt{N}$. We grow $B=500$ trees and
implement our tree-based method using the \textbf{sklearn} library
in Python. Similarly, one could implement the method using the \textbf{randomForest}
or \textbf{ranger} package in \textrm{R}, and the \textbf{TreeBagger}
class in MATLAB.

\vspace{0.166667in}
As a note on the relative performance of tree-based controls and other synthetic controls, recall that the objective is to estimate $\tau_{t}$ in the post-treatment period, which can be achieved as above by proper estimation of $f(x)$. Thus, the relative performance of these methods may be assessed by comparative studies of the underlying methods, e.g., least squares, elastic net, and random forests. Many studies assess the empirical performance of these methods, and for instance, both \cite{Medeiros2019} and \cite{Gu2020} claim that random forests achieve superior performance in most cases. In fact, \cite{howard2012two} claim that the method has been the most successful general-purpose algorithm in modern times, which underpins the need for nonparametric methods in program evaluation too. 

\subsection{Extensions}
\paragraph*{Multiple treated units}
Recent work on synthetic controls focuses on the case of multiple
treated units given its relevance in empirical applications (see, e.g., \citet{Hainmueller2012,Cavallo2013,Robbins2017}).
Incorporating multiple treated units into our framework would entail to
extending the univariate random forests model with a multivariate loss function and splitting rule. For
instance, \citet{Glenn2002} defines multivariate regression trees
analogously to a regression tree with the extension that the loss function
is the multivariate sum of squared error losses. The idea of partitioning
the space of the explanatory variables into disjoint regions and assigning
a constant to each region remains intact.

\vspace{0.166667in}
Another extension is provided
by \citet{Segal2011} who propose multivariate random forests. Again,
the core idea is the same, and the extension entails minimizing
a covariance-weighted loss of the multivariate sum of squared error
losses, where the covariance matrix is based on the multivariate response
function. The multivariate random forests have for instance been applied
by \citet{Pierdzioch2018} to forecasting multiple metal returns.
To estimate the treatment effects on multiple units, we suggest applying
the multivariate random forests directly instead of the random forests.
This would lead to a vector of counterfactual outcomes for the treated
units in each of the post-treatment periods.

\vspace{0.166667in}
\paragraph*{Importance of control units}
A key advantage of regression-based estimators and, in particular,
classical synthetic controls is the transparency of the resulting
counterfactual prediction due to the estimated weights. In the case
of synthetic controls, the counterfactual is a convex combination
of control units and a natural generalization of difference-in-differences.
In contrast, many nonparametric methods optimized for prediction and,
in particular, machine learning methods do not come with such transparency
and are often viewed as non-interpretable black boxes. We briefly
explain two approaches that would allow one to recover part of the
transparency. Particularly for forests, variable importance measures have been centered around split counts or total impurity decrease contributed by all splits for a given predictor variable. \cite{Lundberg2017} claim, however, that both methods are inconsistent, meaning that increasing the importance of a predictor may lead to a lower score using these two measures. In contrast, permutation-based variable importance or SHAP values are consistent ways to assessing variable importance. The former randomly permutes the values of a predictor variable and compute the changes in the objective function (see, e.g., the conditional variable importance measure by \cite{Strobl2008}), whereas the latter computes the marginal change in the predictions when the predictor variable is added to the model averaged over all predictor permutations \citep{Lundberg2017}. Both approaches would allow researchers
to assess which of the control units that drive the counterfactual prediction. Note that because our tree-based counterfactual prediction
is not a weighted average of control units but an average of treated
outcomes in the pre-treatment period, the two approaches would rather
assess which of the control units that are important drivers for computing
the similarities between subperiods and eventually group them.

\paragraph*{Treatment effects beyond the mean}
We comment on the ability of the model to recover treatment
effects beyond the mean. The random forests model estimates the conditional mean by an
averaged prediction of $B$ regression trees, which is essentially a
weighted mean over the observations of $Y_t^0$ for $t \leq T_0$. Likewise, one could define an approximation
to $\mathbb{E}\left[\mathds{1}\left\{ Y_1^0\leq y\right\} \vert X_1=x\right]$
by the weighted average over observations of $\mathds{1}\left\{ Y_1^0\leq y\right\}$. This approximation is suggested by \cite{Meinshausen2006}, leading
to quantile regression forests. Quantile regression forests are a consistent
estimator of the conditional distributions and the quantile functions.
To estimate other distributional properties of the treatment effect than the mean using tree-based controls,
we recommend using the quantile random forests
and estimate the treatment effects over a range of quantiles. 

\clearpage
\section{Proofs}\label{Sec:Appendix_proof}
Before turning to the proof, we rigorously define regression tree and corresponding forests. To this end, let $(\mathcal{P}_n)$ be a sequence of partitions of the input space $\mathcal{X}$ defined by starting from $\mathcal{P}_1 = \{\mathcal{X}\}$ and then, for each $n\geq 1$, obtaining $\mathcal{P}_{n+1}$ from $\mathcal{P}_n$ by replacing a node $A\in \mathcal{P}_n$ by $A_L= \{x\in A\, :\, x^{(i)} \leq z\}$ and $A_R= \{x\in A\, :\, x^{(i)} > z\}$. Here the split direction $i\in \{1,\dots, N\}$ and position $z \in \{x^{(i)} \, :\, x \in A\}$ are determined by a given set of rules, which is allowed to depend on both the data $\mathcal{D}_{T_0}$ as well as a randomization parameter $\Theta$ (the latter is important in order to be able to construct a forest consisting of diverse trees). Furthermore, $x^{(i)}$ refers to the $i$th entry of $x$. A partition $\Lambda$ is said to be \emph{recursive} if $\Lambda = \mathcal{P}_n$ for some $n \geq 1$ where $\mathcal{P}_1,\dots,\mathcal{P}_n$ are obtained as above. Recursive partitions are exactly those which can be depicted as (binary) trees. Hence, given any recursive partition $\Lambda$, the corresponding regression tree $T_\Lambda$ is defined as
\begin{equation*}
T_\Lambda (x) = \frac{1}{\vert \{t\in [T_0]\, :\, X_t \in A_\Lambda (x)\}\vert}\sum_{t\in [T_0]\colon X_t \in A_\Lambda (x)} Y^0_t,\qquad x\in \mathcal{X},
\end{equation*}
where $A_\Lambda (x)$ refers to the unique set in $\Lambda$ which contains $x$, and $[T_0] = \{1,\dots, T_0\}$. Given a family of recursive partitions, ${\boldsymbol \Lambda} = \{\Lambda_1,\dots, \Lambda_B\}$, the random forest $H_{\boldsymbol \Lambda}$ is an average across the corresponding $B\geq 1$ regression trees, i.e.,
\begin{equation}\label{forests}
H_{\boldsymbol \Lambda} (x) = \frac{1}{B}\sum_{b=1}^B T_{\Lambda_b} (x).
\end{equation}
We remark that, in contrast to the original random forests by \cite{Breiman2001}, all trees $T_{\Lambda_1},\dots, T_{\Lambda_B}$ are built on the same data set $\mathcal{D}_{T_0}$. In particular, we exclude the bootstrap step in the theoretical analysis, meaning that randomness must be injected through $\Theta$. In relation to this we note that, since observations cannot be assumed to be independent in our time series setting, one should apply a different bootstrap approach than the i.i.d. version (e.g., growing trees on block bootstrap samples). While we have imposed no requirements on the splitting scheme used to obtain regression trees in this appendix, we will restrict the attention to the $(\alpha,k,m)$-forests which were introduced in \ref{as1}--\ref{as4} of Section~\ref{consSection}.

\subsection{Consistency of $(\alpha,k,m)$-forests}
To prove Theorem~\ref{consistency} we introduce the following assumptions on the data-generating process.

\begin{assump}\label{dataAssump}
	\begin{enumerate}[label=(A\arabic*)]
		\item\label{a1} The sequence $(X_t,Y^0_t)_{t\in \mathbb{Z}}$ is stationary and has exponentially decaying strong mixing coefficients, that is,
		\begin{equation*}
		\alpha (t) = \sup_{A\in \mathcal{F}_0,\, B\in \mathcal{F}^{t}} \vert \mathbb{P} (A\cap B) - \mathbb{P}(A)\mathbb{P}(B)\vert \leq e^{-\gamma t}, \qquad t \geq 1,
		\end{equation*}
		for some $\gamma \in (0,\infty)$ using the notation $\mathcal{F}_t =  \sigma ((X_s,Y^0_s)\, :\, s\leq t)$ and $\mathcal{F}^t = \sigma ((X_s,Y^0_s)\, :\, s\geq t)$.
		
		\item\label{a2} The target variable $Y^0_t$ is bounded, that is, $\vert Y^0_t \vert \leq M$ almost surely for some $M\in (0,\infty)$.
		
		\item\label{a4} The regression function $f(x) = \mathbb{E}[Y^0_1 \mid X_1 = x]$ is Lipschitz continuous, i.e.,
		\begin{equation*}
		\vert f(x) - f(x^\prime)\vert \leq C\lVert x- x^\prime \rVert\qquad \text{for all $x,x^\prime \in \mathcal{X}$},
		\end{equation*}
		where $C\in (0,\infty)$ is a constant and $\lVert \: \cdot \: \rVert$ is any norm on $\mathbb{R}^N$.
		
		\item\label{a3} There exist monotone bijections $\iota_i\colon \mathcal{X}_i\to [0,1]^N$, $i=1,\dots, N$, such that the transformed vector of covariates $Z_t = (\iota_1 (X^{(1)}_t),\dots, \iota_N (X^{(N)}_t))$ admits a density $f_Z\colon [0,1]^N\to [0,\infty)$ with
		\begin{equation}\label{copulaRestrict}
		\zeta^{-1}\leq f_Z(z) \leq \zeta,\qquad z \in [0,1]^N,
		\end{equation}
		for some $\zeta \in (1,\infty)$.
	\end{enumerate}
\end{assump}
As already mentioned in Section~\ref{sec:Tree-based-Control-Methods}, part~\ref{a1} is classical when proving asymptotic results, particularly when one is able to obtain convergence rates as is the case in Lemma~\ref{treesConcentration} presented below. The second part, \ref{a2}, is not really a restriction from a practical point of view as $M$ can be chosen arbitrarily large.  It can indeed be relaxed considerably by instead requiring that the regression function $f$ is bounded and imposing a couple of other technical conditions. Part~\ref{a4} is used to obtain pointwise consistency, since then we have that $\mathbb{E}[Y^0_1\mid X_1\in N(x)]$ is close to $f(x)$ when $N(x)\subseteq \mathcal{X}$ is a ``small'' neighborhood of $x$. Finally, \ref{a3} is a technical condition that allows us to restrict attention to recursive partitions of $[0,1]^N$ and a vector of covariates $Z_t$ whose distribution is of the ``same order'' as the Lebesgue measure, i.e.,
\begin{equation*}
\zeta^{-1} \Leb (A) \leq \mathbb{P}(Z_t\in A) \leq \zeta \Leb (A).
\end{equation*} 
Roughly speaking, the condition means that entries in $X_t = (X^{(1)}_t,\dots, X^{(N)}_t)$ (i.e., the different covariates at a fixed point in time) are not too dependent. For instance, if $X^{(i)}_t$ admits a strictly positive density $f_i\colon \mathbb{R}\to (0,\infty)$ for $i=1,\dots, N$, then one can take $\mathcal{X}_i = [-\infty,\infty]$ and 
\begin{equation*}
\iota_i (y) = \int_{-\infty}^y f_i (z)\, \dd z,\qquad y \in \mathbb{R},
\end{equation*}
$\iota_i (-\infty) = 0$, and $\iota_i(\infty) =1$. With this choice, \eqref{copulaRestrict} simply means that the copula density of $X_t$ is bounded from below and above by suitable positive constants. A special case of such weak dependence is when $X^{(1)}_t,\dots, X^{(N)}_t$ are independent. While the above choice of $\iota_i$ may seem to be the only natural one, there are indeed other situations where alternative specifications are relevant; e.g., in the nonlinear autoregressive setting $X^{(i)}_t = Y^0_{t-i}$ it is shown in \citet[Lemma~1]{davis2020rf} that \ref{a3} is always satisfied for a different choice of $\iota_i$ under suitable assumptions on the data-generating process of $(Y^0_t)$. As a final remark we emphasize that, although \ref{a3} requires the existence of such bijections, there is no need to transform the data in practice before feeding it into the algorithm. 

Theorem~\ref{consistency} makes use of the same ideas as in \citet{davis2020rf} and \citet{wager2015adaptive}, but we have imposed a different set of assumptions and need to modify the proofs accordingly. This means that, while we write out most details here, we will also sometimes instead give specific references to those papers whenever it is appropriate. The key to consistency is Lemma~\ref{treesConcentration} below which shows that, with high probability, regression trees with ``non-negligible'' leaves concentrate around their theoretical counterpart in a uniform sense. To be precise, define for any recursive partition $\Lambda = \Lambda (\mathcal{D}_{T_0},\Theta)$ the corresponding \textit{partition-optimal} tree
\begin{equation*}
T_\Lambda^\ast (x) = \mathbb{E}_\Lambda [Y^0_1 \mid X_1 \in A_\Lambda (x)],\qquad x \in \mathcal{X}.
\end{equation*}
Here $\mathbb{E}_\Lambda$ denotes expectation with respect to the conditional probability measure $\mathbb{P}_\Lambda = \mathbb{P}(\: \cdot \: \mid \mathcal{D}_{T_0},\Theta)$. This means that $T^\ast_\Lambda (x)$ is simply the map $A\mapsto \mathbb{E}[Y\mid X\in A]$ evaluated at $A_\Lambda (x)$. Furthermore, let $\mathcal{V}_k$ denote the collection of all recursive partitions whose leaves contain at least $k$ observations, i.e.,
\begin{equation*}
\mathcal{V}_k = \{\Lambda\, :\, \text{$\Lambda$ is recursive and $\vert \{t \in [T_0]\, :\, X_t \in A\}\vert \geq k$ for all $A\in \Lambda$}\}.
\end{equation*}
We are now ready to formulate the uniform concentration result for regression trees.
\begin{lemma}\label{treesConcentration}
	Suppose that Assumption~\ref{dataAssump} is satisfied and that $k/(\log T_0)^4 \to \infty$ as $T_0\to \infty$. Then there exists a constant $\beta \in (0,\infty)$ such that, with probability at least $1-2T_0^{-1}$ for all sufficiently large $T_0$,
	\begin{equation}\label{tResult}
	\sup_{(x,\Lambda)\in \mathcal{X}\times \mathcal{V}_k} \vert T_\Lambda (x) - T_\Lambda^\ast (x)\vert \leq \beta \frac{\log T_0}{\sqrt{k}}.
	\end{equation}
\end{lemma}

\begin{proof}
Define $Z_t$ as in part~\ref{a3} of Assumption~\ref{dataAssump}. For any (hyper)rectangle $R\subseteq [0,1]^N$, set $\# R = \vert\{t\in [T_0]\, :\, Z_t\in R \}\vert$, $\mu (R) = \mathbb{P}(Z_1\in R)$, and $\xi (R) = \mathbb{E}[Y^0_1\mid Z_1\in R]$. Furthermore, put $\varepsilon = k^{-1/2}$ and $w = k/(2\zeta T_0)$, where $\zeta \in (1,\infty)$ is given such that \eqref{copulaRestrict} holds. Then, according to \citet{wager2015adaptive}, there exists a set of rectangles $\mathcal{R}$, whose cardinality $\vert \mathcal{R}\vert$ satisfies $\log \vert \mathcal{R}\vert = O(\log T_0)$, and which can approximate any set in
\begin{equation*}
\mathcal{R}_w = \{R\, :\, \text{$R\subseteq [0,1]^N$ is a rectangle with $\Leb (R)\geq w$}\}.
\end{equation*}
More precisely, given any $R\in \mathcal{R}_w$, there exist $R_-,R_+ \in \mathcal{R}$ such that
\begin{equation}\label{approx}
R_-\subseteq R\subseteq R_+\qquad \text{and}\qquad e^{-\varepsilon}\Leb (R_+)\leq \Leb (R)\leq e^\varepsilon\Leb (R_-).
\end{equation}
Here, for two sequences $(a_t)_{t\geq 1}$ and $(b_t)_{t\geq 1}\subseteq (0,\infty)$, the notation $a_t =O (b_t)$ means that $\limsup_{t\to \infty} \tfrac{\vert a_t\vert}{b_t}<\infty$. Following the arguments of \citet[pp.~13--14]{davis2020rf}, we have that
\begin{equation}\label{upperBounds}
\begin{aligned}
\MoveEqLeft\sup_{(x,\Lambda)\in \mathcal{X}\times \mathcal{V}_k} \vert T_\Lambda (x) - T_\Lambda^\ast (x)\vert \\
&\leq \sup_{R\in \mathcal{R}_w} \vert \xi (R) - \xi (R_-)\vert 
 + \sup_{R\in \mathcal{R}_w} \Bigl\vert\frac{1}{\# R_-}\sum_{t\in [T_0]\colon Z_t\in R_-} Y^0_t - \xi (R_-) \Bigr\vert \\
&\quad + \sup_{R\in \mathcal{R}_w} \Bigl\vert \frac{1}{\# R}\sum_{t\in [T_0]\colon Z_t\in R} Y^0_t - \frac{1}{\# R_-}\sum_{t\in [T_0]\colon Z_t\in R_-} Y^0_t \Bigr\vert
\end{aligned}
\end{equation}
on any event $\mathcal{A}$ satisfying
\begin{equation}\label{dominatingEvent}\mathcal{A} \supseteq \bigl\{\text{Any rectangle $R\subseteq [0,1]^N$ with $\# R \geq k$ has $\Leb (R) \geq w$}\bigr\}.
\end{equation}
In \eqref{upperBounds} it is implicitly understood that $R_-\in \mathcal{R}$ refers to the inner approximation of $R\in \mathcal{R}_w$ satisfying \eqref{approx} and that we use the convention $\tfrac{1}{0} \sum_{\emptyset} = 0$. It follows that it is sufficient to argue that, when $T_0$ is large, we can find an event $\mathcal{A}$, such that (i) $\mathbb{P}(A)\geq 1- 2T^{-1}_0$, (ii) each term on the right-hand side of \eqref{upperBounds} is bounded by $C\log T_0 / \sqrt{k}$ on $\mathcal{A}$ for a sufficiently large (absolute) constant $C$, and (iii) \eqref{dominatingEvent} holds. The specific event that we will consider is $\mathcal{A} = \mathcal{A}_1\cap \mathcal{A}_2$, where
\begin{align*}
\mathcal{A}_1 &= \biggl\{\sup_{R\in \mathcal{R}\colon \text{Leb} (R)\geq e^{-1}w}\frac{\vert \# R - T_0 \mu (R)\vert}{\sqrt{T_0 \mu (R)}} \leq c_1 \log T_0 \biggr\},\\
\mathcal{A}_2 &= \biggl\{\sup_{R\in \mathcal{R}\colon \text{Leb} (R)\geq e^{-1}w} \frac{\bigl\vert\frac{1}{T_0}\sum_{t\in [T_0]\colon Z_t\in R} Y^0_t - \mathbb{E}[Y^0_1\mathds{1}_R(Z_1)] \bigr\vert}{\mu (R)}\leq c_2 \frac{\log T_0}{\sqrt{k}}\biggr\}.
\end{align*}
Here $c_1,c_2\in (0,\infty)$ are suitably chosen constants. Concerning (i), we can rely on the exact same arguments as in the proof of \citet[Lemma~3]{davis2020rf} to deduce that $\mathbb{P}(\mathcal{A}_1) \geq 1-T_0^{-1}$. This holds for all sufficiently large $T_0$ and $c_1$ as long as we impose Assumption~\ref{dataAssump}, part \ref{a1}, and assume that $k/(\log T_0)^4 \to \infty$ when $T_0 \to \infty$. To show that $\mathbb{P}(\mathcal{A}_2) \geq 1- T_0^{-1}$, and thus $\mathbb{P}(\mathcal{A})\geq 1- 2T_0^{-1}$, consider an arbitrary rectangle $R\in\mathcal{R}$ with $\text{Leb} (R)\geq e^{-1}w$. Since the sequence $(Y^0_t\mathds{1}_R(Z_t))_{t\in\mathbb{Z}}$ is bounded (by part \ref{a2} of Assumption~\ref{dataAssump}) and its strong mixing coefficients are dominated by those of $(X_t,Y_t^0)_{t\in \mathbb{Z}}$, we can apply a Bernstein-type inequality for weakly dependent sequences (see \citet[Theorem~2]{merlevede2009bernstein}) to deduce that
\begin{equation}\label{f1}
\log \mathbb{P}\Bigl(\Bigl\vert \frac{1}{T_0} \sum_{t\in [T_0]\colon Z_t\in R}Y^0_t - \mathbb{E}[Y^0_1\mathds{1}_R(Z_1)]\Bigr\vert > x\Bigr) \leq \frac{-\delta x^2T_0}{\nu_R + T_0^{-1} + x (\log T_0)^2}
\end{equation}
for any $x\in (0,\infty)$ and small enough (generic) $\delta \in (0,\infty)$, where
\begin{equation}\label{varExpression}
\nu_R^2  = \text{Var}(Y^0_1\mathds{1}_R(Z_1)) + 2 \sum_{t=1}^\infty \vert \text{Cov}(Y^0_{t+1}\mathds{1}_R(Z_{t+1}),Y^0_1\mathds{1}_R(Z_1))\vert.
\end{equation}
Since $\inf\{y \in [0,\infty)\, :\, \mathbb{P}(\vert Y^0_1\vert \mathds{1}_R(Z_1) >y)\leq u\} \leq M \mathds{1}_{\{u\leq \mu (R)\}}$, Rio's covariance inequality (\citet[Theorem~1.1]{rio1993covariance}) implies
\begin{equation}\label{fact1}
\vert \text{Cov}(Y^0_{t+1}\mathds{1}_R(Z_{t+1}),Y^0_1\mathds{1}_R(Z_1))\vert \leq 4 M^2 \min \{e^{-\gamma t}, \mu (R)\}.
\end{equation}
Here we have used \ref{a1} and~\ref{a2} of Assumption~\ref{dataAssump}. By part~\ref{a3} of Assumption~\ref{dataAssump} we have $\mu (R)\geq k/(2e\zeta^2 T_0)$ (by the choice of $R$), and hence $\mu (R)\geq e^{-\gamma t/2}$ when $t \geq \tilde{T}\coloneqq \lceil 2\gamma^{-1}\log (2e\zeta^2 T_0/k) \rceil$. Consequently,
\begin{equation}\label{fact2}
\sum_{t=1}^\infty \min \{e^{-\gamma t},\mu (R)\} \leq \mu (R)\tilde{T} + \sum_{t=\tilde{T}+1}^\infty e^{-\gamma t} \leq \mu (R) \Bigl(\tilde{T} + \frac{1}{e^{\gamma/2}-1} \Bigr).
\end{equation}
By combining \eqref{varExpression}--\eqref{fact2}, we find that $\nu_R^2 = O(\mu (R)\log T_0)$ which, in view of \eqref{f1}, shows that 
\begin{equation*}
\log \mathbb{P}\Bigl(\Bigl\vert \frac{1}{T_0} \sum_{t\in [T_0]\colon Z_t\in R}Y^0_t - \mathbb{E}[Y^0_1\mathds{1}_R(Z_1)]\Bigr\vert > x\Bigr) \leq \frac{-\delta x^2T_0}{\max \{\mu (R) \log T_0, x (\log T_0)^2\}}
\end{equation*}
for a small $\delta$. In particular, with $\delta$ being sufficiently small,
\begin{equation*}
\mathbb{P}\Bigl(\Bigl\vert \frac{1}{T_0} \sum_{t\in [T_0]\colon Z_t\in R}Y^0_t - \mathbb{E}[Y^0_1\mathds{1}_R(Z_1)]\Bigr\vert > \delta^{-1} x\Bigr) \leq \frac{1}{\vert \mathcal{R}\vert T_0}
\end{equation*}
for
\begin{equation}\label{lowerX}
x = \max \biggl\{\frac{(\log T_0)^2\log (\vert \mathcal{R}\vert T_0)}{T_0}, \sqrt{\frac{\mu (R)}{T_0}\log T_0 \log (\vert \mathcal{R}\vert T_0)}\biggr\}.
\end{equation}
By the choice of $R$, the second term of the maximum in \eqref{lowerX} is largest when
\begin{equation}\label{kRestriction}
k \geq 2e\zeta^2 (\log T_0)^3 \log (\vert \mathcal{R}\vert T_0)
\end{equation}
Since $\log \vert \mathcal{R}\vert = O(\log T_0)$ and $k/(\log T_0)^4 \to\infty$ as $T_0 \to \infty$, \eqref{kRestriction} is satisfied for all sufficiently large $T_0$. Consequently, as we also have
\begin{equation*}
\frac{1}{\mu (R)}\sqrt{\frac{\mu (R)}{T_0} \log T_0\log (\vert \mathcal{R}\vert T_0)} = O\Bigl(\frac{\log T_0}{\sqrt{k}}\Bigr),
\end{equation*}
we conclude that
\begin{equation*}
\mathbb{P}\biggl(\frac{\bigl\vert \frac{1}{T_0} \sum_{t\in [T_0]\colon Z_t\in R}Y^0_t - \mathbb{E}[Y^0_1\mathds{1}_R(Z_1)]\bigr\vert}{\mu (R)} > c_2 \frac{\log T_0}{\sqrt{k}}\biggr) \leq \frac{1}{\vert \mathcal{R}\vert T_0}
\end{equation*}
for a suitable constant $c_2\in (0,\infty)$ and all sufficiently large $T_0$. By a union bound over all $R\in \mathcal{R}$ with $\text{Leb} (R)\geq e^{-1}w$ this shows that $\mathbb{P}(\mathcal{A}_2)\geq 1- T_0^{-1}$.

Now we turn the attention to (iii) and argue that $\mathcal{A}$ satisfies \eqref{dominatingEvent}. For a general $R\in \mathcal{R}_w$, the inequality of $\mathcal{A}_1$ applies to its outer approximation $R_+$, and this can be used to establish
\begin{equation}\label{allRectangles}
\sup_{R\in \mathcal{R}_w}\frac{\# R - e^{\zeta^2 \varepsilon}T_0 \mu (R)}{\sqrt{T_0 \mu (R)}} \leq c_1 e^{\zeta^2 \varepsilon/2}\log T_0.
\end{equation}
(See also \citet[Eq.~(5.23)]{davis2020rf}.) If we assume that $\Leb (R)< w$, \eqref{allRectangles} can always be applied to a larger rectangle $\tilde{R}\supseteq R$ with $\text{Leb}(\tilde{R}) = w$ to deduce that
\begin{equation}\label{numberOfObs}
\# R \leq \Bigl(\frac{e^{\zeta^2 \varepsilon}}{2} + \frac{c_1 e^{\zeta^2\varepsilon/2}\log T_0}{\sqrt{2k}} \Bigr)k
\end{equation}
Since $\log T_0 /\sqrt{k}\to 0$ by assumption, \eqref{numberOfObs} shows that $\# R<k$ as long as $T_0$ is exceeds a certain threshold (which does not depend on $R$). In other words, for $T_0$ sufficiently large, $\mathcal{A}_1$ (and thus $\mathcal{A}$) meets \eqref{dominatingEvent}.

Finally, we will establish (ii) which amounts to arguing that each of the three terms on the right-hand of \eqref{upperBounds} is smaller than $C \log T_0 / \sqrt{k}$. The first term is handled by part~\ref{a2} of Assumption~\ref{dataAssump}, \eqref{copulaRestrict}, and \eqref{approx}, i.e.,
\begin{equation}\label{term1}
\vert \xi (R) - \xi (R_-)\vert \leq 2M\frac{\mu (R\setminus R_-)}{\mu (R)}\leq 2M\zeta^2 (1-e^{-\varepsilon})\leq \frac{2M\zeta^2}{\sqrt{k}}.
\end{equation}
For the third term on the right-hand side of \eqref{upperBounds} we have that
\begin{align}\label{thirdTerm}
\MoveEqLeft \sup_{R\in \mathcal{R}_w}\Bigl\vert \frac{1}{\# R}\sum_{t\in [T_0]\colon Z_t\in R} Y^0_t - \frac{1}{\# R_-}\sum_{t\in [T_0]\colon Z_t\in R_-} Y^0_t \Bigr\vert \notag\\
&\leq \sup_{R\in \mathcal{R}_w}\biggl\{\Bigl(\frac{1}{\# R_-}-\frac{1}{\# R} \Bigr)\sum_{t\in [T_0]\colon Z_t \in R_-} \vert Y^0_t\vert + \frac{1}{\# R} \sum_{t\in [T_0]\colon Z_t \in R\setminus R_-}\vert Y^0_t\vert\biggr\} \notag \\
&\leq 2M \sup_{R\in \mathcal{R}_w} \frac{\# R - \# R_-}{\# R}.
\end{align}
To bound $(\# R - \# R_-)/\# R$ for $R\in \mathcal{R}_w$, apply the inequality of $\mathcal{A}_1$ to $R_-$ (which is allowed since $\Leb (R_-)\geq e^{-1}w$) and \eqref{allRectangles} to $R$, i.e., 
\begin{align}\label{firstPair}
\begin{array}{lcl}
\# R &\leq&  e^{\zeta^2 \varepsilon}T_0\mu (R) + c_1 e^{\zeta^2\varepsilon/2}\log T_0 \sqrt{T_0 \mu (R)},\\
\# R_- &\geq &  T_0 \mu (R_-)- c_1 \log T_0 \sqrt{T_0 \mu (R_-)}.
\end{array}
\end{align}
Since $T_0 \mu (R_-)\geq k/(4\zeta T_0)$, the last inequality of \eqref{firstPair} shows that $\sqrt{T_0 \mu (R_-)}\leq 2 \zeta \# R/\sqrt{k} + c_1 \log T_0$. Thus, by combining this inequality and the last one of \eqref{firstPair}, we have
\begin{equation}\label{secondIneq}
\# R \geq \frac{T_0 \mu (R_-)- c_1^2 (\log T_0)^2}{1+2c_1\zeta \log T_0/\sqrt{k}} \geq \frac{T_0 \mu (R_-)- c_1^2 (\log T_0)^2}{2}
\end{equation}
as long as $T_0$ is sufficiently large. Now, having the estimates \eqref{firstPair} and~\eqref{secondIneq} at hand, we can rely on the exact same arguments as in the proof of \citet[Lemma~4]{davis2020rf} to establish the bound
\begin{equation*}
\frac{\# R - \# R_-}{\# R} \leq \frac{3\zeta^2 +6c_1\log T_0}{\sqrt{k}}.
\end{equation*}
Together with \eqref{thirdTerm} this shows that the third term of \eqref{upperBounds} is bounded by $6M(\zeta^2 + 2c_1\log T_0)/\sqrt{k}$. For the second (and last) term on the right-hand side of \eqref{upperBounds} we note initially that, due to the inequality of $\mathcal{A}_1$,
\begin{equation*}
\# R \geq T_0 \mu (R)\Bigl(1 - \frac{c_1\log T_0}{\sqrt{T_0 \mu (R)}} \Bigr) \geq \frac{T_0 \mu (R)}{2}
\end{equation*}
for any $R\in\mathcal{R}$ with $\Leb (R)\geq e^{-1}w$ as long as $T_0$ is sufficiently large. Thus, on the set $\mathcal{A}$ (where both inequalities of $\mathcal{A}_1$ and $\mathcal{A}_2$ apply) we find that
\begin{align*}
\MoveEqLeft\sup_{R\in \mathcal{R}_w} \Bigl\vert\frac{1}{\# R_-}\sum_{t\in [T_0]\colon Z_t\in R_-} Y^0_t - \xi (R_-) \Bigr\vert\\
&\leq M\sup_{R\in \mathcal{R}\colon \Leb (R)\geq e^{-1}w}\frac{\vert \# R - T_0 \mu (R)\vert}{\# R}\\
&\quad + \sup_{R\in \mathcal{R}\colon \Leb (R)\geq e^{-1}w}  \frac{T_0}{\# R}\Bigl\vert \frac{1}{T_0}\sum_{t\in [T_0]\colon Z_t \in R}Y_t^0-\mathbb{E}[Y^0_1\mathds{1}_R(Z_1)] \Bigr\vert\\
&\leq 2M\sup_{R\in \mathcal{R}\colon \Leb (R)\geq e^{-1}w}\frac{\vert \# R - T_0 \mu (R)\vert}{T_0\mu (R)}\\
&\quad + 2\sup_{R\in \mathcal{R}\colon \Leb (R)\geq e^{-1}w}  \frac{\bigl\vert \frac{1}{T_0}\sum_{t\in [T_0]\colon Z_t \in R}Y_t^0-\mathbb{E}[Y^0_1\mathds{1}_R(Z_1)] \bigr\vert}{\mu (R)}\\
&\leq \bigl(2M\zeta c_1\sqrt{2e} + 2c_2 \bigr)\frac{\log T_0}{\sqrt{k}}.
\end{align*}
We conclude that all three terms on the right-hand side of \eqref{upperBounds} are bounded by a term of the form $C\log T_0/\sqrt{k}$ and, hence, the proof is complete.
\end{proof}

Lemma~\ref{treesConcentration} ensures that the empirical averages in leaves become close to their theoretical versions; in particular, for a tree associated to a partition $\Lambda$, its prediction at $x$ is close to $T^\ast_\Lambda (x) = \mathbb{E}[Y\mid X\in A_{\Lambda}(x)]$. While this holds for arbitrary trees, whose leaves contain at least $k$ observations (when Assumption~\ref{dataAssump} is satisfied and $k/(\log T_0)\to \infty$ as $T_0 \to \infty$), we need it to be close to $f(x)$, and this is the reason that we restrict the attention to $(\alpha,k,m)$-forests. Indeed, trees of an $(\alpha, k,m)$-forest do not only have a minimum number $k$ of observations in the leaves, but meet \ref{as1}--\ref{as4} of Section~\ref{consSection}.

\begin{proof}[Proof of Theorem~\ref{consistency}]
	Write $\hat{f}(x) = \frac{1}{B}\sum_{b=1}^BT_{\Lambda_b}(x)$ for suitable $B\geq 1$ and recursive partitions $\Lambda_1,\dots, \Lambda_B$ constructed using the rules \ref{as1}--\ref{as4} of Section~\ref{consSection}. With
	\begin{align*}
	\begin{array}{lcl}
	\delta_1 & = & \sup_{(x,\Lambda)\in \mathcal{X}\times \mathcal{V}_k} \vert T_\Lambda(x) - T^\ast_\Lambda (x)\vert, \\
	\delta_2(x) & = & \max_{b=1,\dots, B} \vert T^\ast_{\Lambda_b}(x)-f(x)\vert,
	\end{array}
	\end{align*}
	we clearly have that
	\begin{equation*}
	\vert \hat{f}(x)-f(x)\vert \leq \delta_1 + \delta_2(x),
	\end{equation*}
	and both $\delta_1$ and $\delta_2(x)$ are bounded by $2M$. Furthermore, Lemma~\ref{treesConcentration} implies that $\delta_1\to 0$ in probability as $T_0 \to \infty$. Now, let $\Lambda = \Lambda_b$ for an arbitrary $b\in \{1,\dots, B\}$, and note that
	\begin{equation*}
	\vert T^\ast_{\Lambda}(x) - f(x) \vert \leq \frac{\mathbb{E}_\Lambda[\vert f(X_1) - f(x)\vert \mathds{1}_{A_\Lambda (x)}(X_1)]}{\mathbb{P}_\Lambda (X_1\in A_{\Lambda}(x))} \leq C \text{diam} (A_\Lambda (x))
	\end{equation*}
	by part~\ref{a4} of Assumption~\ref{dataAssump}. Here $\text{diam}(A) = \sup_{x,x^\prime\in A}\lVert x- x^\prime \rVert$ refers to the diameter of a set $A$. In \citet[Lemma~6]{davis2020rf} it was argued that $\text{diam} (A_\Lambda (x))\to 0$ almost surely, so we conclude $\delta_2(x) \to 0$. Since almost sure convergence implies convergence in probability, \eqref{classicalCons} follows immediately, and the proof is complete.
\end{proof}

\subsection{Consistency of tree-based synthetic control methods}
To prove Theorem \ref{tauConsistency}, we need slightly stronger assumptions than those presented in Assumption~\ref{dataAssump}.
\begin{assump}\label{dataAssump2}
	\begin{enumerate}[label=(B\arabic*)]
	\item\label{aa1} The sequence $(Y^1_t)_{t\in \mathbb{Z}}$ is ergodic, and $\mathbb{E}[\vert Y^1_t\vert ]<\infty$.
	\item\label{aa2} The sequence $(X_t,Y_t^0)_{t\in \mathbb{Z}}$ has exponentiall decaying $\beta$-mixing coefficients, that is, 
	\begin{equation*}
	\beta (t) = \mathbb{E}\bigl[\sup_{B\in \mathcal{F}^t}\vert \mathbb{P}(B\mid \mathcal{F}_0)-\mathbb{P}(B)\vert \bigr]\leq e^{-\gamma t},\qquad t \geq 1,
	\end{equation*}
	for some $\gamma \in (0,\infty)$.
	\item\label{aa3} Parts~\ref{a2}--\ref{a3} of Assumption~\ref{dataAssump} are satisfied.
	\end{enumerate}
\end{assump}

\begin{proof}[Proof of Theorem~\ref{tauConsistency}]
	First observe that 
	\begin{align}\label{tauDecomposition}
	\begin{aligned}
	\hat{\tau}-\tau &= \frac{1}{T-T_0}\sum_{t=T_0+1}^T \bigl(Y^1_t - \mathbb{E}[Y^1_t] \bigr) + \frac{1}{T-T_0}\sum_{t = T_0 + 1}^T \bigl(\mathbb{E}[Y^0_t]-f(X_t)\bigr)\\
	&\quad + \frac{1}{T-T_0}\sum_{t=T_0+1}^T \bigl(f(X_t)-\hat{f}(X_t)\bigr),
	\end{aligned}
	\end{align}
	and the first two terms on the right-hand side of \eqref{tauDecomposition} converge to $0$ as $T-T_0\to \infty$ by Birkhoff's ergodic theorem (cf.\ parts~\ref{aa1} and~\ref{aa2} of Assumption~\ref{dataAssump2}). We will show that the third term converges to $0$ in mean. To do so, we use \citet[Lemma~2.6]{yu1994rates} which implies that
	\begin{equation*}\bigl\vert \mathbb{E}[\vert f (X_t) - \hat{f} (X_t)\vert] - \mathbb{E}[\vert f (X) - \hat{f} (X)\vert]\bigr\vert\leq 2 M e^{-\gamma(t-T_0)}
	\end{equation*}
	for an arbitrary $t\geq T_0 +1$, where $X$ has the same distribution as $X_t$, but is independent of $\mathcal{D}_{T_0}$. Here we have made use of part~\ref{aa2} of Assumption~\ref{dataAssump2} and the boundedness assumption on $Y^0_t$ imposed in Assumption~\ref{dataAssump}. In particular, we obtain the estimate
	\begin{align}\label{L1conv}
	\begin{aligned}
	\mathbb{E}\biggl[ \biggl\vert \frac{1}{T-T_0}\sum_{t=T_0 + 1}^T \bigl(f(X_t) - \hat{f}(X_t) \bigr) \biggr\vert \biggr] & \leq \frac{1}{T-T_0}\sum_{t=T_0+ 1}^T \mathbb{E}[\vert f(X_t) - \hat{f}(X_t) \vert]\\
	&\leq \frac{2M}{(e^{\gamma}-1)(T-T_0)} + \mathbb{E}[\vert \hat{f}(X) - f(X)\vert].
	\end{aligned}
	\end{align}
	Now we invoke Theorem~\ref{consistency} to obtain
	\begin{equation}\label{expBound}
	\mathbb{E}[\vert \hat{f}(X) - f(X)\vert] \leq \mathbb{E}[\delta_1] + \mathbb{E}[\delta_2(X)],
	\end{equation}
	where $\delta_1$ and $\delta_2 (x)$ are given as in the statement of the result. Since $\delta_1$ is bounded and converges to $0$ in probability (as $T_0 \to \infty$), $\mathbb{E}[\delta_1]\to 0$. By Tonelli's theorem we also have that $\delta_2(X)\to 0$ almost surely, and hence $\mathbb{E}[\delta_2 (X)] \to 0$ by boundedness of $\delta_2 (X)$ and Lebesgue's theorem on dominated convergence. By combining \eqref{L1conv} and \eqref{expBound} we conclude that $\frac{1}{T-T_0}\sum_{t=T_0 + 1}^T (f(X_t) - \hat{f}(X_t))$ converges to $0$ in mean, and this concludes the proof.
\end{proof}

\end{document}